\documentclass[12pt]{article}
\usepackage{amsmath}
\usepackage{amssymb}
\usepackage{amsthm}
\usepackage{fullpage}
\usepackage[utf8]{inputenc}
\usepackage{mathtools}
\usepackage{url}
\usepackage[usenames,svgnames]{xcolor}
\usepackage{cite}
\usepackage{enumitem}
\usepackage{comment}
\usepackage{mathdots}
\usepackage{graphicx}
\usepackage{float}
\usepackage{enumitem}
\usepackage{datetime}
\usepackage[titletoc,title]{appendix}
\usepackage{comment}
\usepackage[symbol]{footmisc}

\usepackage[margin=1.3in]{geometry}

\usepackage{ tikz, pgfplots, todonotes,color}
\usetikzlibrary{decorations.markings,arrows,math}

\usepackage{url}

\usepackage{stackengine}
\usepackage{scalerel}

\usepackage[pagebackref=true]{hyperref} 
\usepackage[all]{hypcap} 

\theoremstyle{plain}
\theoremstyle{definition}
\newtheorem{theorem}{Theorem}[section]
\newtheorem{lemma}[theorem]{Lemma}
\newtheorem{definition-theorem}[theorem]{Definition-Theorem}
\newtheorem{definition-proposition}[theorem]{Definition-Proposition}

\newtheorem{example}{Example}[section]
\newtheorem{examples}{Example}[subsection]

\newtheorem{remark}{Remark}[section]

\newtheorem{definition}{Definition}[section]

\numberwithin{equation}{section} 

\numberwithin{equation}{section} 

\DeclareMathOperator{\diag}{diag}
\DeclareMathOperator{\cyc}{cyc}

\DeclareMathOperator{\aut}{aut}

\def\ra{{\rightarrow}}

\def\tr{\mathrm {tr}}
\def\det{\mathrm {det}}

\def\diag{\mathrm {diag}}

\def\&{&{\hskip -20pt}}

\DeclarePairedDelimiter{\abs}{|}{|}

\def\be{\begin{equation}}
\def\ee{\end{equation}}

\def\bea{\begin{eqnarray}}
\def\eea{\end{eqnarray}}

\def\bt{\begin{theorem}}
\def\et{\end{theorem}}

\def\bex{\begin{example}\small \rm}
\def\eex{\end{example}}

\def\bexs{\begin{examples}\small \rm}
\def\eexs{\end{examples}}

\def\ra{\rightarrow}
\def\ss{\subset}
\def\deq{\coloneqq}

\def\br{\begin{remark}\small \rm}
\def\er{\end{remark}}

\def\CC{{\mathcal C}}

\def\GG {{\mathcal G}}

\def\MM{{\mathcal M}}

\def\SS{{\mathcal S}}
\def\TT {{\mathcal T}}

\def\VV{{\mathcal V}}
\def\WW{{\mathcal W}}

\def\Ib{\mathbf{I}}

\def\Nb{\mathbf{N}}

\def\Pb{\mathbf{P}}

\def\Zb{\mathbf{Z}}

\def\cb{\mathbf{c}}
\def\db{\mathbf{d}}

\def\mub{\boldsymbol{\mu}}
\def\nub{\boldsymbol{\nu}}

\def\grM{\mathfrak{M}}

\begin{document}
\baselineskip 16pt

\medskip
\begin{center}
\begin{Large}\fontfamily{cmss}
\fontsize{17pt}{27pt}
\selectfont
	\textbf{Constellations and $\tau$-functions for rationally weighted Hurwitz numbers}
	\end{Large}

\bigskip \bigskip
\begin{large} 
J. Harnad$^{1, 2}$\footnote[1]{e-mail:harnad@crm.umontreal.ca}  
and B. Runov$^{1, 2}$\footnote[2]{e-mail:boris.runov@concordia.ca}
 \end{large}
 \\
\bigskip
\begin{small}
$^{1}${\em Centre de recherches math\'ematiques, Universit\'e de Montr\'eal, \\C.~P.~6128, succ. centre ville, Montr\'eal, QC H3C 3J7  Canada}\\
$^{2}${\em Department of Mathematics and Statistics, Concordia University\\ 1455 de Maisonneuve Blvd.~W.~Montreal, QC H3G 1M8  Canada}
\end{small}
 \end{center}
\medskip
\begin{abstract}
Weighted constellations give graphical representations of weighted branched
coverings of the Riemann sphere. They were  introduced to provide a combinatorial
interpretation of the $2$D Toda $\tau$-functions of hypergeometric type
serving as generating functions for weighted Hurwitz numbers in the case of polynomial
weight generating functions.  The product over all vertex and edge weights of a given weighted
constellation, summed over all configurations, reproduces the $\tau$-function. 
In the present work, this is generalized to constellations in which the weighting parameters
are determined by a rational weight generating function. The associated $\tau$-function may be expressed 
as a sum over the weights of  doubly labelled weighted constellations, with two types of weighting parameters
 associated to each equivalence  class of branched coverings.
The double labelling of branch points, referred to as ``colour'' and ``flavour'' indices, is required by the fact that,
in the Taylor expansion of the weight generating function, a particular colour from amongst the denominator 
parameters may appear multiply, and the flavour labels indicate this multiplicity.
 \end{abstract}

\section{$2$D Toda $\tau$-functions, weighted Hurwitz numbers and constellations}
\label{tau_hurwitz_constell}

There is a growing literature on the use of KP or $2$D Toda $\tau$-functions \cite{Sa, SW, Ta, Takeb, UTa} of hypergeometric 
type \cite{Or, OrSc1, OrSc2} as generating functions for various types of weighted Hurwitz numbers 
\cite{Ok, Pa, GJ, MMN, AC1, AC2, AMMN1, AMMN2, GGN1, GGN2, H1, H2, H3, HO2, GH1, GH2}.
The latter are weighted sums over enumerations of Riemann surfaces,  realized as $N$-sheeted branched 
covers of the Riemann sphere, having specified classes of branching
structures \cite{Hur1, Hur2}. Equivalently, they may be viewed as enumerations of factorizations of 
the identity element $\Ib_N \in \SS_N$ in the symmetric group $\SS_N$  on $N$ elements into a product of elements within specified conjugacy
classes \cite{Frob1, Frob2, Sch}. The two are related through the monodromy representation of the fundamental group of the sphere punctured
at the branch points, obtained by lifting homotopy classes of loops from the base to the covering surface \cite{LZ}.

Both interpretations involve an enumeration of more precisely defined structures, in which not
just the conjugacy classes in $\SS_N$ are specified, but the actual monodromy group elements. 
Branched covers with specified monodromy,  or equivalently, factorizations of the identity element in $\SS_N$, 
have nice combinatorial representations as special types of marked bipartite graphs on the surface, called {\em constellations} \cite{LZ},
which contain in their markings all the data needed to uniquely identify the full ramification  structure
 of the branched cover or, equivalently, the monodromy group elements at each of the branch points.

By assigning suitable weights to the vertices and edges, we obtain a {\em weighted} constellation, 
depending on  parameters that characterize the nature of the weighting. Taking a product over all the
vertex and edge weights to obtain the total weight of the constellation, and summing over all branched coverings, 
taking into account the number of configurations corresponding to a given set of conjugacy classes (the pure Hurwitz numbers), 
we reconstruct, in a purely combinatorial fashion, the specific KP or $2$D Toda $\tau$-function that serves
as generating function for the corresponding weighted Hurwitz numbers.

Such weighted graphical representations have previously only been implemented for weights
determined by a polynomial weight generating function \cite{ACEH1, ACEH2}.  One class of such weighted 
coverings that has been considerably studied consists of Belyi curves, in which  there are just three branch points, with either one 
or two of them assigned specific ramification profiles. Removing one of the three branch points (or one of the factors in the
group product), the constellations reduce to Grothendieck's {\em dessins d'enfants} \cite{Z, KZ, AC1, AC2}.

More general weight generating functions are needed, however, to include the various types of Hurwitz numbers 
of interest, such as: {\em simple} Hurwitz numbers \cite{Pa, Ok}, for which the weight generating function is the exponential function;
{\em weakly monotonic} Hurwitz numbers \cite{GGN1, GGN2},  which enumerate weakly monotonic paths in the
Cayley graph of $\SS_N$ generated by transpositions or, equivalently, 
 {\em signed} enumerations of branched covers with  arbitrary branching  \cite{GH1, GH2, HO2},
  for which the weight generating function is the inverse of a linear function;
  more general families of weighted Hurwitz numbers, combining enumerations of weakly and strongly monotonic
  paths in the Cayley graph \cite{GH2, HO2}, for which the weight generating function is rational 
 and {\em quantum} weighted Hurwitz numbers \cite{H1, H3, HOrt, HR }, for which the generating function 
 is a quantum exponential.

The purpose of the present work is to extend the notion of weighted constellations
from the polynomially weighted case to more general {\em rationally weighted} 
Hurwitz numbers \cite{GH2, H1, HO2, BH}. Although the graphical interpretation of
the underlying unweighted constellations remains the same, the labelling is refined
 to include not only a distinction between two species of vertices, associated to the numerator  and denominator
 weighting parameters, and a specification of ``colour'', to distinguish the various  weighted branch points, 
 but also a second, non-negative integer label,  which we call ``flavour'', that takes into account the 
 degree of the monomials in the denominator  parameters that occur in the Taylor expansion of the 
 rational  weight generating function.  

Generating functions for rationally weighted Hurwitz numbers may all be given matrix integral representations
 when the KP and $2$D Toda flow parameters are evaluated at the trace invariants of a pair of ``external coupling''
 matrices  \cite{BH, GH1, GGN1, GGN2}. In particular, for the case of weakly monotonic  double Hurwitz numbers, 
 where the weight generating function is the inverse of a first degree factor,  the generating $\tau$-function
 is identifiable as the Harish-Chandra-Itzykson-Zuber matrix integral  \cite{HC, IZ, GGN1, GGN2}. The next simplest case, where both the 
 numerator and the denominator factors are linear \cite{GH1},  also has a matrix integral representation \cite{HO1, GH1}
 and may be interpreted combinatorially as enumeration  of a $2$-species hybrid of weakly and strongly monotonic paths 
 in the Cayley graph generated by transpositions. In general, by refining the parametrization, or taking into account the
coefficients of monomial terms in the weighting parameters in the expansion, rationally weighted Hurwitz numbers may be interpreted
as enumeration of ``multispecies'' branched covers, or equivalently, various hybrid
classes of monotonic paths in the Cayley graph, with specified step lengths \cite{H2}.

Section \ref{tau_function_gen}  recalls the definitions,  both geometrical and group theoretical, 
of {\em pure}  Hurwitz numbers and  {\em weighted} Hurwitz numbers for arbitrary weight generating functions.
The  construction of $2$D Toda $\tau$-functions of hypergeometric type \cite{Or, OrSc1, OrSc2}  that serve as their generating functions
 is also recalled (Theorem \ref{prop:tauHurwitz}), with particular focus on the case of rational weight generating functions. 
In Section \ref{rational_constel_tau}, we define the corresponding doubly labelled weighted constellations and show
(Theorem \ref{tau_function_rat_weighted_constell}) that, summing over their total weights, we recover the 
$2$D Toda $\tau$-functions that generate rationally weighted Hurwitz numbers. Section \ref{examples_rat_constel} provides a number
of examples of rationally weighted constellations,  including the two cases mentioned above, 
of weakly monotonic Hurwitz numbers and hybrids of weakly and strongly monotonic ones \cite{GH1}.  

\section{$2$D Toda $\tau$-functions as generating functions for weighted Hurwitz numbers}
\label{tau_function_gen}

\subsection{Hurwitz numbers}

We begin by recalling the definition of {\em pure Hurwitz numbers} in the group theoretical sense
of Frobenius  and Schur \cite{Frob1, Frob2, LZ, Sch} (called {\em Frobenius numbers} in \cite{Za}).
\begin{definition} 
\label{frob_pure_hurwitz}
For a pair $N, k \in \Nb^+$ of positive integers, and a set of $k$ integer partitions $\{\mu^{(1)}, \dots, \mu^{(k)}\}$
of $N$, the {\em pure Hurwitz number} $H(\mu^{(1)}, \dots, \mu^{(k)})$ is defined as $1/N!$ times the number of distinct factorizations of the 
identity element $\Ib_N \in \SS_N$ in the symmetric group on $N$ elements as a product of  $k$ elements
\be
\Ib_N = h_1 \cdots h_k, 
\label{factor_k_Id}
\ee
 such that, for each $i=1, \dots, k$,  $h_i \in \SS_N$ belongs to the conjugacy class $\cyc(\mu^{(i)})\ss \SS_N$ 
 whose cycle lengths are equal to the parts of $\mu^{(i)}$. The connected version of the pure Hurwitz numbers,
 denoted $\tilde{H}(\mu^{(1)}, \dots, \mu^{(k)})$  is determined by including in the enumeration only those factorizations
 (\ref{factor_k_Id}) in which the group generated by the factors acts transitively on the set of elements $(1, \dots , N)$ that
 $\SS_N$ permutes.
\end{definition}
An equivalent way of defining  $H(\mu^{(1)}, \dots, \mu^{(k)})$ stems from the work of Hurwitz \cite{Hur1, Hur2}
on enumeration of $N$-sheeted branched covers of the Riemann sphere.
 \begin{definition}
 \label{hurwitz_pure_hurwitz}
 The pure Hurwitz number $H(\mu^{(1)}, \dots, \mu^{(k)})$ may alternatively be defined as the number
 of inequivalent $N$-sheeted branched covers  $\Gamma \ra \Pb^1$ of the Riemann sphere having
 $k$ (possible) branch points $(Q^{(1)}, \dots , Q^{(k)})$, such that the ramification profiles of the $Q^{(i)}$'s are given
 by the partitions $\{\mu^{(i)}\}_ {i=1, \dots, k}$, divided by the order $|\aut(\Gamma)|$ of the automorphism
 group $\aut(\Gamma)$. The connected version  $\tilde{H}(\mu^{(1)}, \dots, \mu^{(k)})$ is similarly determined
 by enumerating only branched covers that  cannot be split into a disjoint union of connected components.
  \end{definition}
 The equivalence of these two definitions \cite{LZ}  follows from the monodromy representation
 \be
 \MM: \pi_1(\Pb^1\backslash(Q^{(1)}, \dots, Q^{(k)})) \ra \SS_N
 \ee
  of  the fundamental group of the  punctured sphere $\Pb^1\backslash(Q^{(1)}, \dots, Q^{(k)})$
 obtained by lifting homotopy classes of loops in  $\Pb^1\backslash(Q^{(1)}, \dots, Q^{(k)})$ to $\Gamma$, 
 and noting that a simple loop surrounding all the  branch points is homotopic to the identity element.

The Frobenius-Schur formula  \cite{Sch, Frob1, Frob2, Za} 
\be
 H(\mu^{(1)}, \dots, \mu^{(k)}) = \sum_{\lambda, \ |\lambda|=N} (h(\lambda))^{k-2} \prod_{i=1}^k {\chi_\lambda(\mu^{(i)}) \over z_{\mu^{(i)}}},
 \label{Frob_Schur_Hurwitz}
 \ee
 gives an explicit expression  for $H(\mu^{(1)}, \dots, \mu^{(k)})$  in terms of the characters $\chi_\lambda(\mu^{(i)})$ of $\SS_N$  
 corresponding to irreducible representations  with symmetry class $\lambda$, evaluated on the conjugacy classes
 $\{\cyc(\mu^{(i)})\}_{i=1 \dots, k}$, consisting of elements of cycle type $\mu^{(i)}$.
The sum in (\ref{Frob_Schur_Hurwitz}) is over all partitions $\lambda \vdash {\hskip -2 pt} N$ 
 (i.e., over all irreducible representations of $\SS_N$),  $h(\lambda)$ denotes the product of  
 hook lengths of $\lambda$ and
 \be
 z_\mu = \prod_{i=1}^{\mu_1} m_i(\mu))! i^{m_i(\mu)}
 \label{z_mu_def}
 \ee
 is the order of the stability subgroup of any element in the conjugacy class $\cyc(\mu)$,
 where $m_i(\mu)$ is the number of parts of $\mu$ equal to $i$. 

\subsection{Weighted Hurwitz numbers}
\label{weighted_hurwitx}

We recall the notion of {\em weighted} Hurwitz numbers \cite{GH2, HO2} associated to a weight generating function
\be
G(z) = 1 + \sum_{i=1}^\infty g_i z^i,
\ee
which may be viewed as a formal power series or, alternatively, a formal infinite product
\be
G(z) = G_{\bf c}(z):= \prod_{i=1}^\infty (1+c_i z)
\label{G_inf_prod}
\ee
or, in {\em dual form} 
\be
G(z) = \tilde{G}_{\bf d}(z):=\prod_{i=1}^\infty (1-d_i z)^{-1}.
\label{G_dual_inf_prod}
\ee
In the first case, the Taylor coefficients are identified as the {\em elementary symmetric functions} \cite{Mac}
\be
g_i = e_i(\cb)
\ee
of the infinite sequence of parameters
\be
\cb = (c_1, c_2, \dots),
\ee
while in the second, they are identified as the {\em complete symmetric functions} \cite{Mac}
\be
g_i = h_i(\db)
\ee
of the infinite sequence 
\be
\db = (d_1, d_2, \dots).
\ee

In some cases, the weight generating function series may be summed to an analytic function, such
as the exponential function 
\be
G(z) = G_{exp}(z):=e^z,
\ee
which occurs in the case of {\em simple} Hurwitz numbers \cite{Pa, Ok, GH1, GH2}.
In this work, we focus on {\em rational} weight generating functions, of the form
\be
	G(z) = G_{{\bf c}, {\bf d}}(z) :=\frac{\prod_{i=1}^{L}(1+c_i  z)}{\prod_{i=1}^{M}(1-d_i z)}\,,
	\label{Gi_rat}
\ee
for which the Taylor series coefficients are
\be
g_i=g_i({\bf c}, {\bf d}):= \sum_{j=0}^i e_j({\bf c}) h_{i-j}({\bf d}),
\label{rational_g_i}
\ee
in terms of the finite set of variables 
\be
{\bf c} =(c_1, \dots, c_L), \quad {\bf d} = (d_1, \dots, d_M). 
\ee

 Given a weight generating function $G(z)$ and a pair of partitions $(\mu, \nu)$ of $N$,  
 we consider weighted branched covers $\Gamma \ra \Pb^1$ of the Riemann sphere with possible branch points 
 at $0$ and $\infty$, having ramification profiles $\nu$ and $\mu$ respectively, and a further set of branch points 
 at  $k$ distinct, finite, nonzero points $(Q^{(1)}, \dots, Q^{(k)})$, whose ramification 
 profiles are $(\mu^{(1)}, \dots, \mu^{(k)})$.
Recall that the sum 
\be
d :=\sum_{i=1}^k\ell^*(\mu^{(i)})
\label{d_def}
\ee
of the colengths 
\be
\ell^*(\mu^{(i)}) := N- \ell(\mu^{(i)}) 
\ee
of the ramification profiles $(\{\mu^{(i)}\}_{i=1, \dots, k}, \mu, \nu) $ 
determines the Euler characteristic $\chi(\Gamma)$ of the covering surface 
through the Riemann-Hurwitz formula
\be
\chi(\Gamma) = 2 - 2g(\Gamma) = \ell(\mu) +\ell(\nu) -d,
\label{riemann_hurwitz}
\ee
where for connected covers $\Gamma\ra \Pb$,  $g(\Gamma)$ is the genus.
For a given choice $(G(z), d, \mu, \nu)$, the corresponding weighted  double Hurwitz number 
$H^d_G(\mu, \nu)$ is defined as follows.
 \begin{definition}[Weighted double Hurwitz numbers]
 For $d \in \Nb^+$, 
 \be
H^d_G(\mu, \nu) := 
\sum_{k=0}^d  \sideset{}{'}\sum_{\mu^{(1)}, \dots \mu^{(k)}, \atop  |\mu^{(i)}| =N , \ \sum_{i=1}^k \ell^*(\mu^{(i)}) =d} 
\WW_G(\mu^{(1)}, \dots, \mu^{(k)}) H(\mu^{(1)}, \dots, \mu^{(k)}, \mu, \nu),
\label{H_d_G_def}
\ee
where $\sideset{}{'}\sum$ denotes a sum over all  $k$-tuples of partitions 
$(\mu^{(1)}, \dots, \mu^{(k)})$ of $N$ other than the cycle type of the identity element $(1^N)$,
\be
\ell^*(\mu^{(i)}) := |\mu^{(i)}| - \ell(\mu^{(i)})= N - \ell(\mu^{(i)})
\ee
is the {\em colength} of the partition $\mu^{(i)}$, which may take values between $1$ (simple branching) and $N-1$ (complete
ramification),  and  $\WW_G(\mu^{(1)}, \dots, \mu^{(k)})$ is the weight
given to the configuration  of ramification profiles $(\mu^{(1)}, \dots, \mu^{(k)}, \mu, \nu)$. (Note that
this is independent of the choice of the two partitions $(\mu, \nu)$). For $d=0$, we  set
\be
H^0_G(\mu, \nu) := H(\mu, \nu).
\ee

For the case of a weight generating function represented by an infinite product,
as in eq.~(\ref{G_inf_prod}), for $k>0$, the weight factor is defined to be \cite{GH2, HO2}
\bea
 \WW_{G_{\bf c}}(\mu^{(1)}, \dots, \mu^{(k)}) &\&:=
 {1\over k!}
 \sum_{\sigma \in \SS_{k}} 
 \sum_{1 \le a_1 < \cdots < a_{k } } 
  c_{a_{\sigma(1)}}^{\ell^*(\mu^{(1)})} \cdots c_{a_{\sigma(k)}}^{\ell^*(\mu^{(k)})} \cr
 &\&= {|\aut(\lambda)|\over k!} m_\lambda ({\bf c}).
 \label{WG_def}
\eea
and, for $k=0$
\be
\WW_G(\emptyset) := 1.
\ee
Here  $m_\lambda ({\bf c}) $ is the monomial symmetric function \cite{Mac} of the 
parameters \hbox{${\bf c}:= (c_1, c_2, \dots)$}
\be
m_\lambda ({\bf c}) = {1\over |\aut(\lambda)|}\sum_{\sigma \in \SS_{k}} \sum_{1 \le a_1 < \cdots < a_{k }}
 c_{a_{\sigma(1)}}^{\lambda_1} \cdots c_{a_{\sigma(k)}}^{\lambda_{k}} ,
  \label{m_lambda}
\ee
 indexed by the partition $\lambda$ of weight $|\lambda|=d$ and length  $\ell(\lambda) =k$, whose 
parts $\{\lambda_i\}_{i=1, \dots, d}$  are equal to the colengths $\{\ell^*(\mu^{(i)})\}_{i=1, \dots, d}$ 
expressed in weakly decreasing order
\be
\{\lambda_i\}_{i=1, \dots k} \sim \{\ell^*(\mu^{(i)})\}_{i=1, \dots k},  \quad \lambda_1 \ge \cdots \ge \lambda_k >0,
\label{lambda_colengths_mu}
\ee
and
 \be
 |\aut(\lambda)|:=\prod_{i\geq 1} m_i(\lambda)!.
 \label{def_aut_lambda}
 \ee
We similarly denote by
\be
\tilde{H}^d_G(\mu, \nu) := 
\sum_{k=1}^d \sideset{}{'}\sum_{\mu^{(1)}, \dots \mu^{(k)}, \atop  |\mu^{(i)}| =N, \ \sum_{i=1}^k \ell^*(\mu^{(i)}) =d} 
\WW_G(\mu^{(1)}, \dots, \mu^{(k)}) \tilde{H}(\mu^{(1)}, \dots, \mu^{k)}, \mu, \nu)
\label{H_d_G_connected_def}
\ee
the {\em connected} weighted Hurwitz numbers corresponding to the weight generating function $G(z)$.

For the case of weight generating functions of the dual infinite product type (\ref{G_dual_inf_prod}),
the weight $\WW_{\tilde{G}_{\bf d}}(\mu^{(1)}, \dots, \mu^{(k)})$ is similarly defined \cite{GH2, HO2}, but with the
monomial symmetric functions $m_\lambda({\bf c})$ appearing in eq.~(\ref{WG_def}) replaced by the
``forgotten'' symmetric function $f_\lambda({\bf d})$, as defined in \cite{Mac}.

For rational weight generating functions (\ref{Gi_rat}), there are two types of branch points,
$(\mu^{(1)}, \dots, \mu^{(l)}; \nu^{(1)}, \dots, \nu^{(m)})$, corresponding to the  weighting parameters
$({\bf c},{\bf d})$ appearing in the numerator and denominator of (\ref{Gi_rat}) respectively, and
the weight becomes \cite{GH2, HO2}
\bea
&\& \WW_{G_{{\bf c}, {\bf d}}}(\mu^{(1)}, \dots, \mu^{(l)}; \nu^{(1)}, \dots, \nu^{(m)}) \cr
&\& :=
{(-1)^{\sum_{j=1}^l \ell^*(\nu^{(j)}) -m }\over l! m!}
 \sum_{\sigma \in \SS_l \atop \sigma' \in \SS_m}   \sum_{1 \leq a_1 < \cdots < a_{l } \leq L\atop 1 \leq b_1\cdots \leq b_m\leq M } 
  c_{a_{\sigma(1)}}^{\ell^*(\mu^{(1)})} \cdots c_{a_{\sigma(l)}}^{\ell^*(\mu^{(l)})}  
  d_{b_{\sigma'(1)}}^{\ell^*(\nu^{(1)})} \cdots d_{b_{\sigma'(m)}}^{\ell^*(\nu^{(m)})}.
 \label{WG_cd_def}
\eea
The corresponding rationally weighted double Hurwitz numbers are denoted
 \bea
H^d_{G_{{\bf c}, {\bf d}}}(\mu, \nu)&\& := 
\sum_{l=0}^L\sum_{m=0}^d  \sideset{}{'}\sum_{{\mu^{(1)}, \dots \mu^{(l)} ,\nu^{(1)}, \dots \nu^{(m)}, \atop    
\sum_{i=1}^l \ell^*(\mu^{(i)})+ \sum_{j=1}^m \ell^*(\nu^{(j)})  =d}\atop  |\mu^{(i)}|  =   |\nu^{(j)}|=N }
{\hskip -30 pt} \WW_{G_{{\bf c}, {\bf d}}}(\mu^{(1)}, \dots, \mu^{(l)}; \nu^{(1)}, \dots, \nu^{(m)})  \cr
&\&{\hskip 120 pt}  \times H(\mu^{(1)}, \dots, \mu^{(l)}, \nu^{(1)}, \dots, \nu^{(m)}, \mu, \nu).\cr
&\&
\label{H_d_G_rat_def}
\eea
\end{definition}

Note that the weight factor $ \WW_{G_{{\bf c}, {\bf d}}}(\mu^{(1)}, \dots, \mu^{(l)}; \nu^{(1)}, \dots, \nu^{(m)}) $ 
depends only on the parameters ${\bf c}, {\bf d}$
and on the colengths $\{\ell^*(\mu^{(i)}), \ell^*(\nu^{(j)}\}$ of the weighted partitions, not on the pair $(\mu, \nu)$,
nor on any further details relating to the parts of $\{\mu^{(i)}, \nu^{(j)}\}$.

\subsection{$2$D Toda $\tau$-functions of hypergeometric type}
\label{Toda_hypergeom_G}

For a given weight generating function $G(z)$,  and a small nonzero complex constant  $\beta$, 
we assume that $G(z)$ may be evaluated  at $z=j\beta$, for all integers $j\in \Zb$, 
and define the following infinite sequence of parameters:
\be
r_j^{(G, \beta)} := G(j \beta) .
\label{rj_beta}
\ee
For any partition $\lambda$ of weight $|\lambda|= N$,
we define $r_\lambda^{(G, \beta)}$ by the following {\em content product} formula \cite{HO2, GH2}
\be
r_\lambda^{(G, \beta)}  \deq   \prod_{(i,j)\in \lambda} r_{j-i}^{(G, \beta)},
\label{r_lambda_G}
 \ee
where  $(i,j)\in \lambda$ refers to the position of a box in the Young diagram of the partition $\lambda$ 
 in matrix index notation.

Following \cite{GH1, H1, HO2}, we introduce a parametric family $\tau^{(G, \beta, \gamma)}({\bf t},{\bf s})$ 
of 2D Toda $\tau$-functions  of hypergeometric type   \cite{ KMMM, Or, OrSc1, OrSc2} (at the lattice point $0$)
associated to the weight generating function $G(z)$,  defined by the double Schur function series
\begin{eqnarray}
\tau^{(G, \beta, \gamma)}({\bf t},{\bf s})
:= \sum_{N=0}^\infty  \gamma^N\sum_{\lambda, |\lambda|=N} r^{(G, \beta)}_\lambda s_\lambda({\bf t}) s_\lambda({\bf s}),
\label{tau_schur_G}
\end{eqnarray}
where 
\be
{\bf t} =(t_1, t_2, \dots),  \quad {\bf s}= (s_1, s_2, \dots)
\label{2DToda_flow_vars}
\ee
 are the two sets of 2D Toda flow parameters, identified, within normalization, with the power sums
 \be
 p_i = i t_i, \quad p'_i = i s_i
 \ee
 over a pair of sequences of auxiliary variables,  $\gamma$ is an additional nonzero parameter,
 introduced to keep track of the weights $|\lambda |$, and the sum is taken over all integer partitions (including $\lambda=\emptyset$).
These will serve as generating functions for the weighted double Hurwitz numbers as explained below.

Making a change of basis  from the Schur functions to the power sum symmetric functions 
\be
p_\mu({\bf t}):= \prod_{i=1}^{\ell(\mu) }p_{\mu_i} =  \prod_{i=1}^{\ell(\mu) } \mu_i t_{\mu_i}, \quad  p_\nu({\bf s}) := \prod_{i=1}^{\ell(\nu)} p'_{\nu_i} =  \prod_{i=1}^{\ell(\nu)} \nu_i s_{\nu_i},
\ee 
 using the Frobenius character formula \cite{FH, Mac},
\be
s_\lambda({\bf t}) = \sum_{\mu, \, \abs{\mu} = \abs{\lambda}} z_\mu^{-1} \chi_\lambda(\mu) p_\mu ({\bf t}),
\label{frobenius_character}
\ee
allows us to express $  \tau^{(G, \beta, \gamma)} ({\bf t}, {\bf s})$ equivalently as a double series
in the power sum symmetric functions, whose coefficients are equal to the $H^d_G(\mu, \nu)$'s (see ~\cite{GH1, GH2, HO2, H1} for 
details).
\begin{theorem}{\phantom{x}{\hskip -3 pt}\cite{GH1, GH2, HO2}.}
\label{prop:tauHurwitz}
The $\tau$-function $\tau^{(G, \beta, \gamma)} ({\bf t},{\bf s})$ has the equivalent series expansion
\bea
  \tau^{(G, \beta, \gamma)} ({\bf t}, {\bf s}) &\& =  
\sum_{N=0}^\infty  \gamma^N\sum_{\substack{\mu, \nu \\ |\mu|=|\nu| =N}}
\sum_{d=0}^\infty \beta^d  H^d_G(\mu, \nu) p_\mu({\bf t}) p_\nu({\bf s}).
\label{tau_G_H}
\eea
while its logarithm gives the corresponding expansion for connected Hurwitz numbers
\bea
 \ln( \tau^{(G, \beta, \gamma)} ({\bf t}, {\bf s})) &\& =  
\sum_{N=0}^\infty  \gamma^N\sum_{\substack{\mu, \nu \\ |\mu|=|\nu| =N}}
\sum_{d=0}^\infty \beta^d  \tilde{H}^d_G(\mu, \nu) p_\mu({\bf t}) p_\nu({\bf s}).
\label{ln_tau_G_H}
\eea
\end{theorem}
Thus $\tau^{(G, \beta, \gamma)}({\bf t},{\bf s})$ is interpretable as a generating function for weighted double Hurwitz numbers $H^d_G(\mu, \nu)$, with the exponents of the variables $\gamma$ and $\beta$ equal  to the weight $N=|\mu|=|\nu|$ 
and the total colength $d$,  as defined in eq.~(\ref{d_def}), respectively.

In particular, the $2$D Toda $\tau$-function  that serves as generating function for rationally weighted Hurwitz numbers is:
\bea
\tau^{(G_{{\bf c}, {\bf d}}, \beta, \gamma)}({\bf t},{\bf s})
&\&:= \sum_{N=0}^\infty \gamma^N
\sum_{\lambda, \, |\lambda| = N } r^{(G_{{\bf c}, {\bf d}},  \beta)}_\lambda s_\lambda({\bf t}) s_\lambda({\bf s})
\label{tau_schur_G_cd} \\
&\& =\sum_{N=0}^\infty \gamma^N \sum_{\substack{\mu, \nu \\ |\mu|=|\nu| = N}} 
\sum_{d=0}^\infty \beta^d  H^d_{G_{{\bf c}, {\bf d}}}(\mu, \nu) p_\mu({\bf t}) p_\nu({\bf s}).
\label{tau_G_cd_H}
\eea

\section{Rationally weighted constellations and $\tau$-functions}
\label{rational_constel_tau}


\subsection{Constellations and branched coverings}
\label{constellation_singl_double_label}

We begin by recalling the combinatorial representation of {\em constellations} as special types of marked bipartite
 graphs on a Riemann surface $\Gamma$ containing all the combinatorial information required to realize it 
 as an $N$-fold branched cover $\Gamma \ra \Pb^1$  of the Riemann sphere \cite{Ch, LZ}.
Equivalently, these may be viewed as  graphical representations of a factorization of the identity element $\Ib_N\in \SS_N$ in
the symmetric group
\be
h_0 h_1 \cdots h_{k+1} = \Ib_N,
\label{factor_Id_N}
\ee
where, for purposes of the weighted enumerations to follow, it will be convenient to choose the number of factors as
$k+2$, with indexing labels $(0,1, \dots, k+1)$. These will be thought of equivalently as labelling the $k+2$ branch points
$(Q^{(0)},  Q^{(1)}, \dots, Q^{(k+1)})$  in $\Pb^1$ under the projection map $\Gamma \ra \Pb^1$ whose monodromy 
representations are  generated by the factors $(h_0, h_1, \dots, h_{k+1})$, giving the monodromies of the lifts of the simple,
positively oriented loop, passing through a generic  point $P\in \Pb^1$, and going once around the branch points $Q^{(i)}$ for 
$i=0, \dots, k+1$. Conventionally, the points $(Q^{(0)}, Q^{(k+1)})$ will be identified as $(0, \infty)$
and the cycle types of $(h_0, h_{k+1})$ denoted as $(\nu, \mu)$.
 
 The definition below of {\em singly labelled constellations} is essentially the same as Definition 3.1 in \cite{ACEH2}. 
 It differs from the one in \cite{LZ} by the fact that, in  the latter, one of the factors in (\ref{factor_Id_N}), say $h_0$,
 is omitted, and  so are the  associated vertices of the graph that form the preimage of $Q^{(0)}$.
  But since this element is  determined by the factorization formula (\ref{factor_Id_N}), there is a unique way to 
  restore the missing vertices, and the corresponding edges, by adding these to  the ``amputated'' graphs of \cite{LZ}, 
  to obtain the constellations defined here;   i.e.,  the two graphs can be uniquely transformed into each other.
 Another slight difference is that we label the $N$ elements of the set $D$ (the ``darts'') upon which
 $\SS_N$ acts as $(1, 2, \dots, N)$, while there is no such ordered labelling in \cite{LZ}.
 What are there defined as {\em star} vertices therefore do not carry any particular numbering. However, as shown in \cite{LZ}, any
 permutation of the elements of $D$ gives rise to an isomorphic constellation, in which
 the factors in (\ref{factor_Id_N}) are replaced by their conjugates under the permutation.
 To recover the product (\ref{factor_Id_N}) explicitly, not  just up to conjugation, it is of course
 necessary to specify the numbering. Finally, we do not impose the requirement that the graph be connected,
 or equivalently, that the group generated by the factors in (\ref{factor_Id_N}) act transitively on $D$.

\begin{definition}[Singly labelled constellations]
\label{singly_labelled_constell}
	  A (singly labelled) constellation $\CC_{N, k+2}(\Gamma)$ of degree $N$ and length $k+2$
	 is a graph  on a  Riemann surface $\Gamma$ consisting of two types of vertices: $N$ {\em star} vertices, 
	  labelled $(\underline{1}, \dots, \underline{N})$ or, equivalently $(P_1, \dots, P_N)$, which carry no colour,  
	  and  $k+2$  consecutively labelled {\em coloured} vertices, which may appear multiply,
	  carrying the colour labels $i=0, \dots,  k +1$ (where those labelled $k+1$ are conventionally located at
	  points over $\infty$, and those labelled $0$ over the origin, when viewed as defined as a branched
	  cover of the $\Pb^1$),  together with edges between  pairs of these satisfying:
	\begin{enumerate}
	\item\label{star_colour}  All edges connect  pairs  consisting of a star vertex and a coloured vertex.
		\item \label{sedgerule} There is at least one edge joining every coloured vertex to a star vertex.
		\item \label{sdegrule}The degree of each star vertex is $k+2$, with each  connected to exactly one vertex of each colour.
		\item \label{starrule}For any given star vertex, the labels of the coloured vertices to which it is connected
		 are in strictly increasing order when enumerated in a counterclockwise sense.
		\item \label{sfacerule} There are exactly $N$ faces, all homeomorphic to a disc. The boundary of each face contains one angular sector
		 of each colour  (i.e., a sequence: star vertex $\ra$ coloured vertex $\ra$ star vertex).
		 		\end{enumerate}
	\end{definition}	
	\begin{remark} 
	It follows from the above, and the fact that the surface $\Gamma$ is oriented, that the vertices of each face
	appear in decreasing order as it is traversed in a counterclockwise sense starting from $k+1$.
	\end{remark}
		
	\begin{remark} \label{star_disjoint_union} The star vertices $(P_1, \dots, P_N)$,
	  viewed  as points of $\Gamma$,  are the preimages of the generic (nonbranching) point $P$ under the projection map 
	  $\Gamma\ra \Pb^1$.   The coloured vertices carrying the label $i$ are the preimages
	of the branch points $(Q^{(0)}, \dots, Q^{(k+1)})$. The ordered sequence of star vertices they are connected to,
	taken in the clockwise sense, determines a cycle in $\SS_N$, of length $\mu^{(i)}_j$.
		It follows from  requirements \ref{sedgerule} and \ref{sdegrule} that, for any given $i$, 
these cycles are disjoint, and their union is the full set $(P_1, \dots, P_n)$. 
	The product of  the cycles for  $j=1, \dots , \ell(\mu^{(i)})$ defines the group element $h_i$, which belongs to the 
	conjugacy class $\cyc(\mu^{(i)})$, whose cycle lengths  are equal to the parts of the partition 
	$\mu^{(i)}=(\mu^{(i)}_1, \dots, \mu^{(i)}_{\ell(\mu^{(i)})})$ corresponding to distinct vertices of the same colour $i$.
	These are the preimages of the point $Q^{(i)}$, and for illustrative purposes are denoted by $(Q^{(i)}_1, \dots, Q^{(i)}_{\ell(\mu^{(i)})})$
	in Figure \ref{path_lift_cycle_single}.
 The group element $h_i$ is the monodromy representation of the simple, positively oriented loop through $P$, going once around the 
branch point $Q^{(i)}$. Taking an ordered product over all these therefore gives the identity element, as in eq.~(\ref{factor_Id_N}).
	Constellations can be divided into equivalence classes,  determined by the ordered set of partitions 
$(\mu^{(0)}, \mu^{(1)}, \dots \mu^{(k+1)})$, where
\be
\mu^{(0)}=:\nu, \quad \mu^{(k+1)} =: \mu.
\ee
 By Definition \ref{frob_pure_hurwitz}, the number of elements of this equivalence class, 
divided by $N!$, is the pure Hurwitz number $H( \mu^{(1)}, \dots , \mu^{(k)}, \mu, \nu)$.
\end{remark}

\begin{remark}
The edges of the constellation are determined by lifting the $N$ simple, positively oriented loops starting and ending at $P$,
going once around each branch point, and dividing the corresponding lifted paths into cycles. Each edge  corresponds to
a pair consisting of an ``incoming'' $1/2$ path and an ``outgoing'' $1/2$-path, connecting a coloured vertex to one of the points
$(P_1, \dots, P_N)$ that represent the star vertices. Figure \ref{path_lift_cycle_single} gives a visualization 
of this monodromy interpretation of the edges for one  $3$-cycle $(abc)$.
\end{remark}


 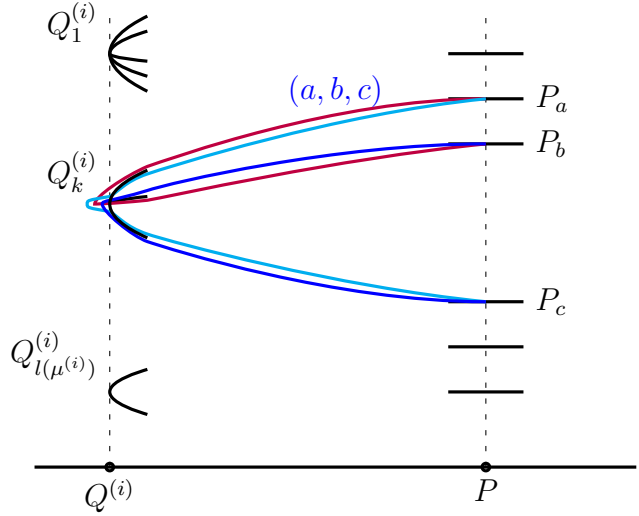
\begin{figure}[H]
 \label{path_lift_cycle_single}
\begin{center}
	\begin{tikzpicture}
		\tikzmath{
			\y1=1;
			\y2=3.5;
			\y3=5.5;}
		\draw[very thick] (0,0) -- (8,0);
		\draw[dash pattern=on 2pt off 4pt] (1,0) -- (1, 6);
		\draw[very thick,rotate=270] (-\y1-0.3,1+0.5) parabola bend (-\y1,1) (-\y1+0.3,1+0.5);
		\draw[very thick] (1,\y1) node[anchor=south east]{$Q^{(i)}_{l(\mu^{(i)})}$};
		\draw[very thick] (1,\y2) node[anchor=south east]{$Q^{(i)}_k$};
		\draw[very thick,rotate=270] (-\y2,1) parabola (-\y2-0.1,1+0.5);
		\draw[very thick,rotate=270] (-\y3-0.3,1+0.5) parabola bend (-\y3,1) (-\y3+0.3,1+0.5);
		\draw[very thick,rotate=270] (-\y3-0.5,1+0.5) parabola bend (-\y3,1) (-\y3+0.5,1+0.5);
		\draw[very thick] (1,\y3) node[anchor=south east]{$Q^{(i)}_1$};
		\draw[very thick,rotate=270] (-\y3,1) parabola (-\y3+0.1,1+0.5);
		\draw[very thick] (6-0.5,\y1) -- (6+0.5,\y1);
		\draw[very thick] (6-0.5,\y1+0.6) -- (6+0.5,\y1+0.6);
		\draw[very thick] (6-0.5,\y1+2*0.6) -- (6+0.5,\y1+2*0.6) node[right]{$P_c$};
		\draw[very thick] (6-0.5,\y3) -- (6+0.5,\y3);
		\draw[very thick] (6-0.5,\y3-0.6) -- (6+0.5,\y3-0.6) node[right] {$P_a$};
		\draw[very thick] (6-0.5,\y3-2*0.6) -- (6+0.5,\y3-2*0.6) node[right]{$P_b$};
		\draw[dash pattern=on 2pt off 4pt] (6,0) -- (6, 6);
		\draw[purple,very thick] (6,\y3-0.6) .. controls (4,\y3-0.6) and (1+0.8,\y2+0.6)..(1+0.5, \y2+0.5) ..controls (1+0.2,\y2+0.37) and (1-0.2,\y2+0.1) .. (1-0.2,\y2) .. controls (1-0.2,\y2)  and (1,\y2)..(1+0.5,\y2+0.05) ..controls (1+0.8,\y2+0.1) and (4,\y3-2*0.6 -0.2).. (6,\y3-2*0.6);
		\draw[cyan,very thick] (6,\y3-0.6) .. controls (4,\y3-0.6-0.2) and (1+0.8,\y2+0.5) .. (1+0.5,\y2+0.4) ..controls(1+0.2,\y2+0.3) and (1+0.1,\y2+0.15) ..(1,\y2+0.1)
		..controls (1-0.3,\y2+0.05) .. (1-0.3,\y2) ..
		controls (1-0.3,\y2-0.05)  .. (1,\y2-0.1) .. controls (1+0.1,\y2-0.15) and (1+0.2,\y2-0.3)  .. (1+0.5,\y2-0.4 ) .. controls (1+0.8,\y2-0.5) and (4,\y1+2*0.6+0.2) .. (6,\y1+2*0.6);
		\draw[blue,very thick] (6,\y3-2*0.6) .. controls (4,\y3-2*0.6) and (1+0.8,\y2+0.28) .. (1+0.5,\y2+0.18) .. controls (1,\y2+0.03) and (1-0.1,\y2+0.05).. (1-0.1,\y2)..
		controls (1-0.1,\y2-0.1) and (1+0.3,\y2-0.43) .. (1+0.5,\y2-0.5) .. controls (1+0.8,\y2-0.6) and (4,\y1+2*0.6) .. (6,\y1+2*0.6);
		\draw[very thick] (1,0) circle (0.05) node[below] {$Q^{(i)}$};
		\draw[very thick] (6,0) circle (0.05) node[below] {$P$};
		\draw[blue] (4,5) node {$(a,b,c)$};
		\draw[very thick,rotate=270] (-\y2-0.45,1+0.5) parabola bend (-\y2,1) (-\y2+0.45,1+0.5);
	\end{tikzpicture}
\end{center}

	 \caption{\footnotesize Lifted loops and edges. Here, $(a,b,c)$ represents a (typical) cycle  in the monodromy factorization at the indicated branch point.}\label{fig:pathABCTer}

\end{figure}
\bigskip

	Figure \ref{fig:planar_N5_constell_1} gives an example of a planar constellation with $N=5$, $k=3$ appearing
	in  \cite{ACEH1, ACEH2}. The factorization  \ref{factor_Id_N} for this case is:
	\be
	h_0 h_1 h_2 h_3 h_4 = \Ib_5 ,
	\label{factor_N5k3}
	\ee
 where 
\be
h_0 = (123),  \ h_1  = (153), \ h_2 =  (15) (23), \ h_3 = (14),  \  h_4 = h_\infty = (14)
      \ee
The corresponding partitions are
 \bea
\mu^{(0)}  &\&= (3, 1, 1) ), \quad  \mu^{(1)} = (3, 1, 1),  \ \mu^{(2)} = (2, 2, 1), \cr
\  \mu^{(3)} &\&= (2, 1,1,1), \  \,  \mu^{(4)}= (2, 1, 1, 1).
\eea

 \begin{figure}[H]
 \label{fig:planar_N5_constell_1}
\begin{center}
	\begin{tikzpicture}
		\tikzmath{
			\sts=0.13;
			\x1=0;
			\y1=0;
			\x2=2.0;
			\y2=-1.2;
			\x3=4.3;
			\y3=-2.4;
			\x4=-1.5;
			\y4=0.1;
			\x5=2.6;
			\y5=2.8;
			}
		\draw[fill=black] (\x5,\y5)  circle   (\sts)  node[anchor=south west]{\underline{5}};
		\draw[fill=black] (\x1,\y1)  circle   (\sts)  node[anchor=north east]{\underline{1}};
		\draw[fill=black] (\x2,\y2)  circle   (\sts)  node[anchor=north east]{\underline{2}};
		\draw[fill=black] (\x3,\y3)  circle   (\sts)  node[anchor=north east]{\underline{3}};
		\draw[fill=black] (\x4,\y4)  circle   (\sts)  node[anchor=north east]{\underline{4}};
		\draw[very thick,fill=white] (\x1 +0.3,\y5-0.3) circle (0.3);
		\draw (\x1+0.3,\y5-0.3) node(v215) {$2$};
		\draw[very thick,fill=white] (\x5-0.3,\y1+0.3) circle (0.3);
		\draw (\x5-0.3,\y1+0.3) node(v115) {$1$};
		\draw[very thick,fill=white] (\x5-1.3,\y5-0.3) circle (0.3);
		\draw (\x5-1.3,\y5-0.3) node(v35) {$3$};
		\draw[very thick,fill=white] (\x5-1.6,\y5-1.6) circle (0.3);
		\draw (\x5-1.6,\y5-1.6) node(vinf5) {$4$};
		\draw[very thick,fill=white] (\x5-0.1,\y5-1.3) circle (0.3);
		\draw (\x5-0.1,\y5-1.3) node(v05) {$0$};
		\draw[very thick,fill=white] (0.5*\x1+0.5*\x4,\y4+0.7) circle (0.3);
		\draw (0.5*\x1+0.5*\x4,\y4+0.7) node (v314) {$3$};
		\draw[very thick,fill=white] (\x4-0.8,\y4+0.7) circle (0.3);
		\draw (\x4-0.8,\y4+0.7) node (v04) {$0$};
		\draw[very thick,fill=white] (\x4-0.9,\y4-0.2) circle (0.3);
		\draw (\x4-0.9,\y4-0.2) node (v14) {$1$};
		\draw[very thick,fill=white] (\x4+0.4,\y4-0.9) circle (0.3);
		\draw (\x4+0.4,\y4-0.9) node (v24) {$2$};
		\draw[very thick,fill=white] (\x4+0.2,\y4-2) circle (0.3);
		\draw (\x4+0.2,\y4-2) node (vinf4) {$4$};
		\draw[very thick,fill=white] (\x1*0.6+\x2*0.4,\y2-0.5) circle (0.3);
		\draw (\x1*0.6+\x2*0.4,\y2-0.5) node (v012) {$0$};
		\draw[very thick,fill=white] (\x2+0.4,\y2-0.8) circle (0.3);
		\draw (\x2+0.4,\y2-0.8) node (v12) {$1$};
		\draw[very thick,fill=white] (\x3+0.5,\y2+1.1) circle (0.3);
		\draw (\x3+0.5,\y2+1.1) node (v223) {$2$};
		\draw[very thick,fill=white] (\x2+ 1,\y2+0.8) circle (0.3);
		\draw (\x2+1,\y2+0.8) node (v32) {$3$};
		\draw[very thick,fill=white] (\x2-0.8,\y2+0.8) circle (0.3);
		\draw (\x2-0.8,\y2+0.8) node (vinf2) {$4$};
		\draw[very thick,fill=white] (\x3 +0.3,\y3+1.3) circle (0.3);
		\draw (\x3+0.3,\y3+1.3) node (v33) {$3$};
		\draw[very thick,fill=white] (\x3-1.0,\y3+0.8) circle (0.3);
		\draw (\x3-1.0,\y3+0.8) node (vinf3) {$4$};
		\draw[very thick] (v35) -- (\x5,\y5);
		\draw[very thick] (vinf5) -- (\x5,\y5);
		\draw[very thick] (v05) -- (\x5,\y5);
		\draw[very thick] (\x4,\y4) -- (v314) -- (\x1,\y1);
		\draw[very thick] (v04) -- (\x4,\y4);
		\draw[very thick] (v14) -- (\x4,\y4);
		\draw[very thick] (v24) -- (\x4,\y4);
		\draw[very thick] (v12) -- (\x2,\y2);
		\draw[very thick] (v32) -- (\x2,\y2);
		\draw[very thick] (vinf2) -- (\x2,\y2);
		\draw[very thick] (v33) -- (\x3,\y3);
		\draw[very thick] (vinf3) -- (\x3,\y3);
		\draw[very thick] (\x1,\y1) -- (v012) -- (\x2,\y2);
		\draw[very thick] (\x1,\y1) --  (v215).. controls (\x1+1.5,\y5+0.4) and  (\x5-0.7,\y5+0.4) .. (\x5,\y5);
		\draw[very thick] (\x1,\y1) --  (v115).. controls (\x5+0.5,\y1+1) and (\x5+0.5,\y5-0.8) .. (\x5,\y5);
		\draw[very thick] (\x4,\y4) .. controls (\x4-0.25,\y4+2) and (\x4-1.75,\y4+2) .. (\x4-2,\y4) .. controls (\x4-2,\y4-1.75) .. (vinf4) .. controls (\x1+0.3,\y1-1.6) and (\x1-0.1,\y1-1.5) .. (\x1,\y1);
		\draw[very thick] (v012) .. controls (\x2*0.7+\x3*0.3, \y3-0.4) .. (\x3,\y3) .. controls (\x3+2,\y3+0.2) and (\x3+2,\y1+0.5) .. (\x3+1,\y1+0.5)
		.. controls (\x3-0.5,\y1+0.7)  .. (v115);
		\draw[very thick] (\x2,\y2) .. controls (\x2*0.5+\x3*0.5+0.05,\y2+0.15).. (v223) .. controls (\x3+2,0.7*\y2+0.3*\y3) and (\x3+1,0.5*\y2+0.5*\y3) .. (\x3,\y3);
	\end{tikzpicture}
\end{center}
\caption{\footnotesize  Example of a planar constellation with $N=5$, $k=3$.}
\end{figure}
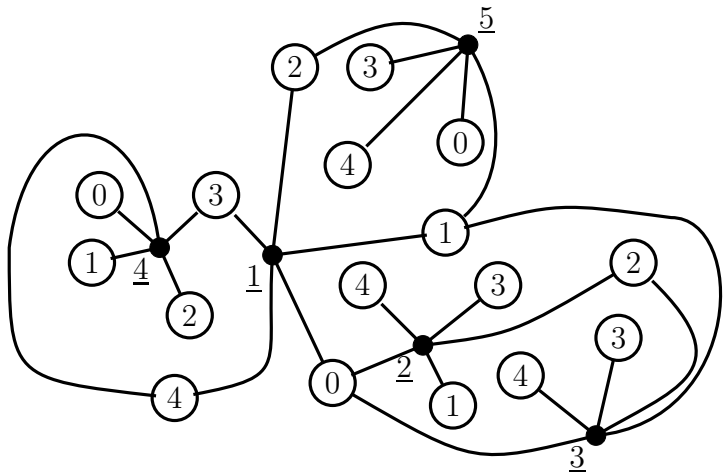

\subsection{Doubly labelled constellations}
\label{double_labelled_constell}

Doubly labelled constellations are constellations  $\CC_{N, k+2}(\Gamma)$
in which  the labelling of the  consecutive colours $(0, 1, \dots, k+1)$ is replaced by a more structured one, consisting 
of pairs $(i, j_i)$ of nonnegative integers,  called ``(colour, flavour)'' pairs.
 Their definition  is just a refinement of the one for singly labelled constellations, with enhanced details regarding the labelling.
We choose a pair of positive integers $(L, M)$, with $L\le k$, and a further set $(J_1, \dots , J_M)$ of positive integers
such that
\be
k = L +\sum_{i=1}^M J_i
\ee
with
\bea
0 &\&\le i \le L+1 \ \text{ if}\  j_i =0 ,
\label{i_ji_range_ji0}
\\
\ \text{ and}  \ 1&\&\le i \le M \ \text{ if}\   1 \le  j_i \le J_i.
\label{i_ji_range_ji1}
\eea
(The pair of integers $(L,M)$ will later be identified with those appearing in the rational weight generating function $G_{{\bf c}, {\bf d}}$
in (\ref{Gi_rat}).) The $M$-tuple of positive integers
\be
J :=(J_1, \dots , J_M)
\ee
will be called the {\em spectrum} of $\CC_{N, k+2}(\Gamma)$. Its cardinality is denoted
\be
\ell(J) := M
\ee
and the sum of its components 
\be
|J| := \sum_{j=1}^M J_j,
\ee
as with partitions.  We therefore have a total of
\be
k+2 = L + |J| + 2 
\ee
coloured vertices.

As  for any $k+2$-constellation $\CC_{N, k+2}(\Gamma)$, of degree $N$, the
graph of the constellation can be drawn on an $N$-sheeted branched cover $\Gamma \ra \Pb^1$
of the Riemann sphere with (at most) $k+2$ branch points. There is, however, also the double labelling information 
contained in the choice of integers $(L, M, (J_1, \dots , J_M))$. Moreover,  instead of just two types of vertices, there
are three: the ``star'' vertices and two types of ``coloured'' vertices.
\begin{itemize}
	\item Star vertices (denoted by black dots), still carry an integer label $\underline{k}$, with $1\le k \le N$
	(which correspond to points on  $\Gamma$ that can be interpreted as preimages of an arbitrarily
	chosen non-branch point $P$.) They may also be labelled by the symbols  $(P_1, \dots, P_N)$, where
	each $P_i$ is viewed as a point over $P$ lying on a different sheet of the branched cover.
	\end{itemize}
	 The two coloured types are:
	\begin{itemize}
	\item Round vertices, which carry the pair of labels $(i,0)$, where $i$ is integer-valued, 
	with $0\le i \le L+1$, enclosed in a circle.  These are interpreted as the primages of a set of  $L+2$ points
\be
	(Q^{(0,0)}=0, Q^{(1,0)}, \dots,  Q^{(i,0)}, \dots , Q^{(L+1,0)}=\infty),
\ee
where, for $1\le i \le L$, the $Q^{(i,0)}\in  \Pb^1$'s are distinct, finite, nonzero points  and the label  $i$ refers to the coefficient $c_i$ in (\ref{Gi_rat}).
	These are assigned one of the colours of the spectrum,  and a ``flavour'' that is set equal to $0$.
	Round vertices carrying the label $(0, 0)$ are assigned the colour ``black''
and interpreted as the preimages of the point $Q^{(0,0)} $ (which is chosen, conventionally,
as  $z=0$). Those carrying the label $(L+1, 0)$ are assigned the colour ``white'' and are interpreted 
as the preimages of the point $Q^{(L+1,0)}$ (which is chosen  conventionally as $z=\infty)$.
\item Square vertices, which carry an ordered pair $(i,j_i)$, of positive integer valued labels within the range (\ref{i_ji_range_ji1}),
enclosed in a square box. These are interpreted as the primages of a set of  $\sum_{i=1}^M J_i$  further (possible) branch points
\be
(Q^{(1, 1)}, \dots, Q^{(1, J_1)},  \dots , Q^{(M, 1)}, \dots, Q^{(M, J_M)}),
\ee
where the subscript  $i$ refers to the coefficient $d_i$ in (\ref{Gi_rat}).
The indices are again referred to as {\em colour} $(i)$ and {\em flavour }$(j_i)$, with  $1\le i \le M$,  where, for each $i$, 
we have $1\le j_i \le J_i$. The colours of the spectrum  assigned to the $i$'s  with $j_i\ge 1$ (i.e., to the $d_i$'s) 
are chosen to be different from those assigned to the round vertices.
\end{itemize}

As in general, we also refer  to the  black and white vertices as ``coloured'', and the (colour, flavour) 
pairs $(i,j_i)$  will be understood as ordered, in the following  sense.
\begin{definition}
By CF ordering of the (colour, flavour) pairs, we mean:
\bea
&\&{(0,0)} < {(1,0)} \cdots <{(L,0)} <{(1, 1)} \cdots  < {(1, J_1)} < 
 \cdots < {(M, 1)} \cdots {(M, J_M)} <{(L+1,0)} .\cr
 &\&
\label{CF_order}
\eea
\end{definition}

 The coloured vertices are understood as the distinct  points in the covering surface
 over the branch points $(Q^{(0,0)}, Q^{(1,0)}, \dots, Q^{(L,0)}, Q^{(1, 1)}, \dots, Q^{(1, J_1)},  \dots,   Q^{(M, 1)}, \dots, Q^{(M, J_M)}, Q^{(L+1,0)})$,
  with the pair $(i,j_i)$  of (colour, flavour) labels indicating the branch point to which they project.

\begin{definition}\label{double_label_constellation}
For a given $(L, M, J)$, a {\em doubly labelled} $L+|J|+2$-constellation $\CC_{N, L+|J|+2}(\Gamma)$ of degree $N$ is 
a graph  on $\Gamma$ formed from $N$ star vertices and $L+|J|+2$ coloured vertices that connect pairs of these, satisfying 
the following requirements:
	\begin{enumerate}
\item\label{dstar_colour}  The graph is bipartite, in the sense that all the edges connect  pairs 
consisting of a star vertex and a coloured vertex.
	\item \label{dedgerule}  Each coloured vertex is connected to at least one star vertex, and every type $(i, j_i)$
	of square vertex has at least one representative that is connected to two or more star vertices. 
\item \label{ddegrule} Every star vertex is connected to exactly one of the coloured vertices of type $(i,j_i)$
(and hence, is of degree $L + |J| + 2$).	
		\item\label{dstarrule} The coloured vertices $\{(i,j_i)\}$ associated to any given star vertex,  
		when enumerated  in a counterclockwise sense,  are in increasing CF  order,  starting from $(0,0)$ .	
	\item \label{dfacerule} There are exactly $N$ faces, all homeomorphic to a disc. Viewing the boundary of any face
	 as a positively oriented polygonal curve, its coloured vertices  consist of all possible pairs $(i,j_i)$.
	 \end{enumerate}
	 \begin{remark}
	 It again follows from the above, and the fact that the surface $\Gamma$ is oriented, that the vertices of each face
	appear in decreasing CF order as it is traversed in a counterclockwise sense.
	\end{remark}
	\end{definition}

\begin{remark} \label{star_disjoint_union} It also follows from  requirements \ref{dedgerule} and \ref{ddegrule}
that the distinct vertices with any given label $(i, j_i)$ are connected to a disjoint set of star vertices, 
whose union is the full set $(P_1, \dots, P_N)$. These therefore define an element $h_{(i, j_i)} \in\SS_N$, 
which is the monodromy of the simple, positively oriented loop through $P$ going once around the 
branch point $Q^{(i, j_i)}$. Taking a CF ordered product over all these therefore gives the identity element:
	\be
h_{(0,0)} \left(\overrightarrow{\prod}_{i=1}^Lh_{(i, 0)}\right) 
 \overrightarrow{\prod}_{k=1}^M \overrightarrow{\prod}_{j_k=1}^{J_k} h_{(k, j_k)} h_{(\infty,0)}= \Ib.
 \label{doubly_labelled_factoriz}
\ee
where none of the elements $\{h_{(k, j_k)}\}_{1 \le k \le M \atop 1\le j_k \le J_k}$ is the identity $\Ib_N$.
\end{remark}

Figure \ref{fig:planar_N5_constell_2} gives an example of a planar constellation with $N=5$, $k=3$
that is equivalent to the one in Figure \ref{fig:planar_N5_constell_1}, but with doubly labelled coloured vertices
with $(L,M,J) = (1,1,(2))$.
The factorization for this labelling is equivalent to (\ref{factor_N5k3}), but with double labelling:
	\be
	h_{(0,0))} h_{(1,0)} h_{(1,1)} h_{(1,2)} h_{(\infty,0)} = \Ib_5  ,
	\ee
 where 
\be
h_{(0,0)} =  (123),\  h_{(1,0)} = (153) , \  h_{(1,1)}  = (15)(23)   , \ h_{(1,2)} = (14) , \  h_{(2,0)}= (14).
\label{N5_L1_M1_factoriz}
      \ee
The corresponding partitions are again
 \bea
\mu^{(0,0)}  &\&= (3, 1, 1) ), \quad  \mu^{(1,0)} = (3, 1, 1)),  \ \mu^{(1,1)} = (2, 2, 1), \cr
\  \mu^{(1,2)} &\&= (2, 1,1,1), \  \,  \mu^{(\infty,0)}= (2, 1, 1, 1).
\eea

 \begin{figure}[H]
 \label{fig:planar_N5_constell_2}
\begin{center}
	\begin{tikzpicture}
		\tikzmath{
			\sts=0.13;
			\x1=0;
			\y1=0;
			\x2=2.0;
			\y2=-1.2;
			\x3=4.3;
			\y3=-2.4;
			\x4=-1.5;
			\y4=0.1;
			\x5=2.6;
			\y5=2.8;
			\rr=0.34;
			}
		\draw[fill=black] (\x5,\y5)  circle   (\sts)  node[anchor=south west]{\underline{5}};
		\draw[fill=black] (\x1,\y1)  circle   (\sts)  node[anchor=north east]{\underline{1}};
		\draw[fill=black] (\x2,\y2)  circle   (\sts);
		\draw (\x2-0.2,\y2-0.4) node{\underline{2}};
		\draw[fill=black] (\x3,\y3)  circle   (\sts)  node[anchor=north east]{\underline{3}};
		\draw[fill=black] (\x4,\y4)  circle   (\sts)  node[anchor=north east]{\underline{4}};
		\draw[very thick,fill=white] (\x1,\y5-0.6) rectangle (\x1+0.6,\y5);
		\draw (\x1+0.3,\y5-0.3) node(v215) {$1,1$};
		\draw[very thick,fill=white] (\x5-0.3,\y1+0.3) circle (\rr);
		\draw (\x5-0.3,\y1+0.3) node(v115) {$1,0$};
		\draw[very thick,fill=white] (\x5-1.6,\y5-0.6) rectangle (\x5-1,\y5);
		\draw (\x5-1.3,\y5-0.3) node(v35) {$1,2$};
		\draw[very thick,fill=white] (\x5-1.6,\y5-1.6) circle (\rr);
		\draw (\x5-1.6,\y5-1.6) node(vinf5) {$2,0$};
		\draw[very thick,fill=white] (\x5-0.1,\y5-1.3) circle (\rr);
		\draw (\x5-0.1,\y5-1.3) node(v05) {$0,0$};
		\draw[very thick,fill=white] (0.5*\x1+0.5*\x4-0.3,\y4+0.4) rectangle (0.5*\x1+0.5*\x4+0.3,\y4+1);
		\draw (0.5*\x1+0.5*\x4,\y4+0.7) node (v314) {$1,2$};
		\draw[very thick,fill=white] (\x4-0.8,\y4+0.7) circle (\rr);
		\draw (\x4-0.8,\y4+0.7) node (v04) {$0,0$};
		\draw[very thick,fill=white] (\x4-0.9,\y4-0.2) circle (\rr);
		\draw (\x4-0.9,\y4-0.2) node (v14) {$1,0$};
		\draw[very thick,fill=white] (\x4+0.1,\y4-1.2) rectangle (\x4+0.7,\y4-0.6);
		\draw (\x4+0.4,\y4-0.9) node (v24) {$1,1$};
		\draw[very thick,fill=white] (\x4+0.2,\y4-2) circle (\rr);
		\draw (\x4+0.2,\y4-2) node (vinf4) {$2,0$};
		\draw[very thick,fill=white] (\x1*0.6+\x2*0.4,\y2-0.5) circle (\rr);
		\draw (\x1*0.6+\x2*0.4,\y2-0.5) node (v012) {$0,0$};
		\draw[very thick,fill=white] (\x2+0.4,\y2-0.8) circle (\rr);
		\draw (\x2+0.4,\y2-0.8) node (v12) {$1,0$};
		\draw[very thick,fill=white] (\x3+0.2,\y2+0.8) rectangle (\x3+0.8,\y2+1.4);
		\draw (\x3+0.5,\y2+1.1) node (v223) {$1,1$};
		\draw[very thick,fill=white] (\x2+0.7,\y2+0.5) rectangle (\x2+1.3,\y2+1.1);
		\draw (\x2+1,\y2+0.8) node (v32) {$1,2$};
		\draw[very thick,fill=white] (\x2-0.8,\y2+0.8) circle (\rr);
		\draw (\x2-0.8,\y2+0.8) node (vinf2) {$2,0$};
		\draw[very thick,fill=white] (\x3,\y3+1) rectangle (\x3+0.6,\y3+1.6);
		\draw (\x3+0.3,\y3+1.3) node (v33) {$1,2$};
		\draw[very thick,fill=white] (\x3-1.0,\y3+0.8) circle (\rr);
		\draw (\x3-1.0,\y3+0.8) node (vinf3) {$2,0$};
		\draw[very thick] (\x5-1.0,\y5-0.25) -- (\x5,\y5);
		\draw[very thick] (\x5-1.35,\y5-1.35) -- (\x5,\y5);
		\draw[very thick] (v05) -- (\x5,\y5);
		\draw[very thick] (\x4,\y4) -- (v314) -- (\x1,\y1);
		\draw[very thick] (\x4-0.55,\y4+0.45) -- (\x4,\y4);
		\draw[very thick] (\x4-0.55,\y4-0.2) -- (\x4,\y4);
		\draw[very thick] (v24) -- (\x4,\y4);
		\draw[very thick] (v12) -- (\x2,\y2);
		\draw[very thick] (\x2+0.7,\y2+0.5)  -- (\x2,\y2);
		\draw[very thick] (\x2-0.55,\y2+0.55) -- (\x2,\y2);
		\draw[very thick] (\x3+0.3,\y3+1.0) -- (\x3,\y3);
		\draw[very thick] (\x3-0.75,\y3+0.55)  -- (\x3,\y3);
		\draw[very thick] (\x1,\y1) -- (\x1*0.6+\x2*0.4,\y2-0.15);
		\draw[very thick] (\x1*0.6+\x2*0.4+0.3,\y2-0.3) -- (\x2,\y2);
		\draw[very thick] (\x1,\y1) --  (v215);
		\draw[very thick] (\x1+0.6,\y5-0.1) .. controls (\x1+1.5,\y5+0.4) and  (\x5-0.7,\y5+0.4) .. (\x5,\y5);
		\draw[very thick] (\x1,\y1) --  (\x5-0.6,\y1+0.2);
		\draw[very thick] (\x5-0.1,\y1+0.6).. controls (\x5+0.5,\y1+1) and (\x5+0.5,\y5-0.8) .. (\x5,\y5);
		\draw[very thick] (\x4,\y4) .. controls (\x4-0.25,\y4+2) and (\x4-1.75,\y4+2) .. (\x4-2,\y4) .. controls (\x4-2,\y4-1.75) .. (\x4-0.15,\y4-2);
		\draw[very thick] (\x4+0.55,\y4-2) .. controls (\x1+0.3,\y1-1.6) and (\x1-0.1,\y1-1.5) .. (\x1,\y1);
		\draw[very thick](\x1*0.6+\x2*0.4+0.25,\y2-0.75)   .. controls (\x2*0.7+\x3*0.3, \y3-0.4) .. (\x3,\y3) .. controls (\x3+2,\y3+0.2) and (\x3+2,\y1+0.5) .. (\x3+1,\y1+0.5)
		.. controls (\x3-0.5,\y1+0.7)  .. (\x5-0.0,\y1+0.4);
		\draw[very thick] (\x2,\y2) .. controls (\x2*0.5+\x3*0.5+0.05,\y2+0.15).. (\x3+0.2,\y2+0.8);
		\draw[very thick] (\x3+0.8,\y2+0.8).. controls (\x3+2,0.7*\y2+0.3*\y3) and (\x3+1,0.5*\y2+0.5*\y3) .. (\x3,\y3);
	\end{tikzpicture}
\end{center}
\caption{\footnotesize  Example of a doubly labelled planar constellation with $N=5$, $L=1$, $M=1$, $J_1 =2$.}
\end{figure}
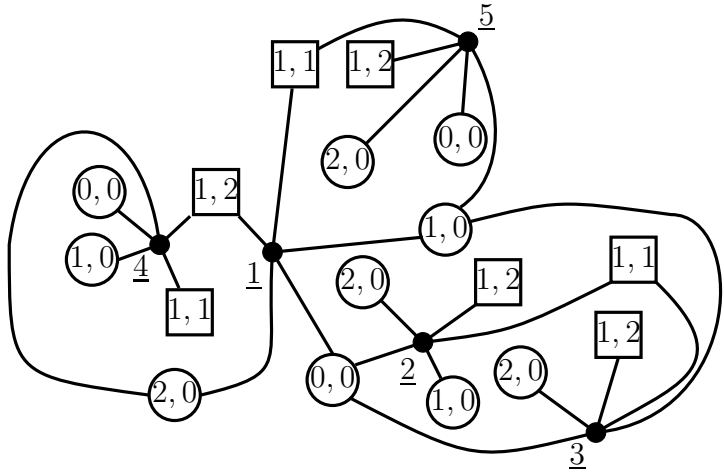
Figure \ref{path_lift_cycle} is a doubly labelled version of Figure \ref{path_lift_cycle_single}. 
It again  illustrates the lift of a closed, simple, positively oriented loop, starting at $P$,  going once around the
branch point $Q^{(i,j_i)}$,  and returning to $P$. The points $(P_a, P_b, P_c)$ are cyclically permuted
when the lifted path encircles the ramification point $Q^{(i, j_i)}_k$, three times, giving a contribution $(a b c)$ to the monodromy
group around this point. Each $1/2$ path contributes to the resulting three edges, with a pair of (inward, outward)
$1/2$ paths constituting each edge.


 \begin{figure}[H]
 \label{path_lift_cycle}
\begin{center}
	\begin{tikzpicture}
		\tikzmath{
			\y1=1;
			\y2=3.5;
			\y3=5.5;}
		\draw[very thick] (0,0) -- (8,0);
		\draw[dash pattern=on 2pt off 4pt] (1,0) -- (1, 6);
		\draw[very thick,rotate=270] (-\y1-0.3,1+0.5) parabola bend (-\y1,1) (-\y1+0.3,1+0.5);
		\draw[very thick] (1,\y1) node[anchor=south east]{$Q^{(i,j_i)}_{l(\nu^{(i,j)})}$};
		\draw[very thick] (1,\y2) node[anchor=south east]{$Q^{(i,j_i)}_k$};
		\draw[very thick,rotate=270] (-\y2,1) parabola (-\y2-0.1,1+0.5);
		\draw[very thick,rotate=270] (-\y3-0.3,1+0.5) parabola bend (-\y3,1) (-\y3+0.3,1+0.5);
		\draw[very thick,rotate=270] (-\y3-0.5,1+0.5) parabola bend (-\y3,1) (-\y3+0.5,1+0.5);
		\draw[very thick] (1,\y3) node[anchor=south east]{$Q^{(i,j_i)}_1$};
		\draw[very thick,rotate=270] (-\y3,1) parabola (-\y3+0.1,1+0.5);
		\draw[very thick] (6-0.5,\y1) -- (6+0.5,\y1);
		\draw[very thick] (6-0.5,\y1+0.6) -- (6+0.5,\y1+0.6);
		\draw[very thick] (6-0.5,\y1+2*0.6) -- (6+0.5,\y1+2*0.6) node[right]{$P_c$};
		\draw[very thick] (6-0.5,\y3) -- (6+0.5,\y3);
		\draw[very thick] (6-0.5,\y3-0.6) -- (6+0.5,\y3-0.6) node[right] {$P_a$};
		\draw[very thick] (6-0.5,\y3-2*0.6) -- (6+0.5,\y3-2*0.6) node[right]{$P_b$};
		\draw[dash pattern=on 2pt off 4pt] (6,0) -- (6, 6);
		\draw[purple,very thick] (6,\y3-0.6) .. controls (4,\y3-0.6) and (1+0.8,\y2+0.6)..(1+0.5, \y2+0.5) ..controls (1+0.2,\y2+0.37) and (1-0.2,\y2+0.1) .. (1-0.2,\y2) .. controls (1-0.2,\y2)  and (1,\y2)..(1+0.5,\y2+0.05) ..controls (1+0.8,\y2+0.1) and (4,\y3-2*0.6 -0.2).. (6,\y3-2*0.6);
		\draw[cyan,very thick] (6,\y3-0.6) .. controls (4,\y3-0.6-0.2) and (1+0.8,\y2+0.5) .. (1+0.5,\y2+0.4) ..controls(1+0.2,\y2+0.3) and (1+0.1,\y2+0.15) ..(1,\y2+0.1)
		..controls (1-0.3,\y2+0.05) .. (1-0.3,\y2) ..
		controls (1-0.3,\y2-0.05)  .. (1,\y2-0.1) .. controls (1+0.1,\y2-0.15) and (1+0.2,\y2-0.3)  .. (1+0.5,\y2-0.4 ) .. controls (1+0.8,\y2-0.5) and (4,\y1+2*0.6+0.2) .. (6,\y1+2*0.6);
		\draw[blue,very thick] (6,\y3-2*0.6) .. controls (4,\y3-2*0.6) and (1+0.8,\y2+0.28) .. (1+0.5,\y2+0.18) .. controls (1,\y2+0.03) and (1-0.1,\y2+0.05).. (1-0.1,\y2)..
		controls (1-0.1,\y2-0.1) and (1+0.3,\y2-0.43) .. (1+0.5,\y2-0.5) .. controls (1+0.8,\y2-0.6) and (4,\y1+2*0.6) .. (6,\y1+2*0.6);
		\draw[very thick,rotate=270] (-\y2-0.45,1+0.5) parabola bend (-\y2,1) (-\y2+0.45,1+0.5);
		\draw[very thick] (1,0) circle (0.05) node[below] {$Q^{(i,j_i)}$};
		\draw[very thick] (6,0) circle (0.05) node[below] {$P$};
		\draw[blue] (4,5) node {$(a,b,c)$};
	\end{tikzpicture}
\end{center}

	 \caption{\footnotesize Lifted loops and edges. Here, $(a,b,c)$ represents a (typical) cycle  in the monodromy factorization at the indicated branch point.}
	 \label{fig:pathABCTer}
\end{figure}
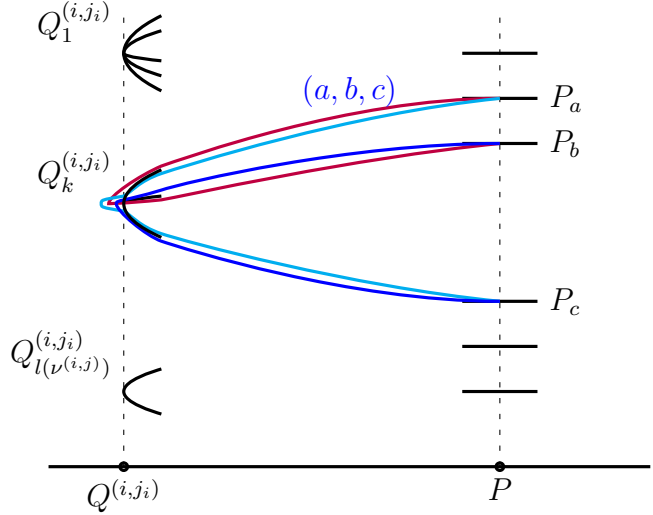
\bigskip

The doubly labelled constellations are divided, as in general, into equivalence classes  determined by the set of partitions
\be
\{\nu:=\mu^{(0)}, \mu^{(1)},...,\mu^{(L)},\nu^{(1,1)}, \dots \nu^{(1,J_1)} \dots, \nu^{(M, 1)} \dots \nu^{(M,J_M)},\mu:=\mu^{(L+1)}\},
\label{mu_nu_partitions}
\ee
with a partial ordering of the vertices carrying the same pair of labels $(i, j_i)$, defined  in the following way:
\begin{itemize}
	\item $\mu^{(i)}_k$ is the number of edges ending in the $k$-th round vertex of colour $i$
	\item $\nu^{(i,j_i)}_k$ is the number of edges ending in the $k$-th square vertex of color $i$ and flavor $j_i$
	\item $\mu_k$ is the number of edges ending in the $k$-th white vertex
	\item $\nu_k$ is the number of edges ending in the $k$-th black vertex
\end{itemize}
corresponding to the cycle lengths of the conjugacy classes of the elements 
\be
(\{h_{(i,0)}\}_{i=1, \dots L+1}, \quad \{h_{(i, j_i)}\}_{1 \le i \le M,  \ 1\le j_i \le J_i},
\ee
where none of the partitions $\{\nu^{(i, j_i)}\}$ is the trivial one $(1)^N$.
The partial ordering  is defined by the weakly decreasing order of the parts of each partition.

Two constellations are again regarded as equivalent (from the viewpoint of Hurwitz enumeration) if the
CF ordered, labelled sets of partitions (\ref{mu_nu_partitions}) 
corresponding to them are the same. The number of elements in the equivalence class is thus, again, $N!$ times
the Hurwitz number 
$H(\nu, \mu^{(1)},...,\mu^{(L)},\nu^{(1,1)}, \dots \nu^{(1,J_1)}, \dots, \nu^{(M, 1)} \dots \nu^{(M, J_M)},\mu)$.


\subsection{Rationally weighted doubly labelled constellations}
\label{double_labelled_constell}

We now define what is meant by a rationally weighted constellation.
\begin{definition}
\label{double_labelled_rational_weighting}
For a rational weight generating function $G_{{\bf c}, {\bf d}}$, a rationally weighted constellation $\GG^J_{{\bf c} ,{\bf d}}$
is a doubly labelled $L+|J|+2$-constellation of  type $(L, M, J)$ and degree $N$, with the following weights attached to its 
 vertices and edges:
\begin{itemize}
	\item The weight of each round vertex of color $i$, where $i =1, \dots, L$, is $(\beta c_i)^{-1}$.
	\item The weight of each square vertex of color $i$, where $i =1, \dots , M$,  is $(-\beta d_i)^{-1}$.
	\item The weight of each edge ending at a round vertex is $\beta c_i$.
	\item The weight of each edge ending at a square vertex is $-\beta d_i$.
	\item The weight of each star vertex is $\gamma$.
	\item The weight of each black vertex of degree $j$ is $p_j(s)=j s_j$.
	\item The weight of each white vertex of degree $j$ is $p_j(t)=j t_j$.
\end{itemize}
\end{definition}
\begin{definition}
The weight $W_{\GG^J_{{\bf c},{\bf d}}}$ of the doubly labelled weighted constellation 
is defined to be the product of  the weights of its edges and vertices multiplied by 
${(-1)^{|J|}\over N!}$.
\end{definition}
The following figures illustrate the labelling and weighting of the different types of vertices and 
edges.
Figure \ref{vertex_edge weights}  illustrates the five types of vertex weights, both for circle and square types of
coloured vertices, and for star vertices, as well as the weights for all edges connected to them.
Figure  \ref{fig:weightsVerticesAllTypesCover} places the weighted vertices and edges into a combined picture,
adding the branch points $\{Q^{(i, j_i)}\}$ of which the coloured vertices are
pre-images, as well as the generic point $P$, with its pre-images $(P_1, \dots, P_N)$.

 \begin{figure}[H]
  \label{vertex_edge weights}
		\tikzmath{
			\rr=0.34;
			\nst=6;
			\nlb=3.5;
			\x1=\rr*cos(72);
			\y1=\rr*sin(72);
			\x2=\rr*cos(130);
			\y2=\rr*sin(130);
			\x3=\rr*cos(216);
			\y3=\rr*sin(216);
			\x4=\rr*cos(288);
			\y4=\rr*sin(288);
			\x5=\rr;
			\y5=0;
			\of1=0;
			\sts=0.13;
			}
 \begin{minipage}{0.35\linewidth}
  \begin{center}
	\begin{tikzpicture}
		\draw (0,0) node {$i,0$};
		\draw[very thick] (0,0) circle(\rr) node[left=8pt,blue] {$(\beta c_i)^{-1}$};
		\draw[very thick] (\x1,\y1)  -- (\nst*\x1,\nst*\y1);
		\draw[very thick] (\x2,\y2)  -- (\nst*\x2,\nst*\y2);
		\draw[very thick] (\x3,\y3) -- (\nst*\x3,\nst*\y3);
		\draw[very thick] (\x4,\y4) -- (\nst*\x4,\nst*\y4);
		\draw[very thick] (\x5,\y5) -- (\nst*\x5,\nst*\y5);
		\draw (\nlb*\x1,\nlb*\y1-0.1) node[above,rotate=72,blue] {$\beta c_i$};
		\draw(\nlb*\x2,\nlb*\y2-0.1) node[above,rotate=180+130,blue] {$\beta c_i$};
		\draw(\nlb*\x3,\nlb*\y3-0.1) node[above,rotate=180+216,blue] {$\beta c_i$};
		\draw(\nlb*\x4,\nlb*\y4-0.1) node[above,rotate=288,blue] {$\beta c_i$};
		\draw(\nlb*\x5,\nlb*\y5-0.1) node[above,rotate=0,blue] {$\beta c_i$};
		\draw[fill=black] (\nst*\x1,\nst*\y1) circle (\sts) node[anchor=west,blue] {\Large$\gamma$};
		\draw[fill=black] (\nst*\x2,\nst*\y2) circle (\sts) node[anchor=west,blue] {$\gamma$};
		\draw[fill=black] (\nst*\x3,\nst*\y3) circle (\sts) node[anchor=west,blue] {$\gamma$};
		\draw[fill=black] (\nst*\x4,\nst*\y4) circle (\sts) node[anchor=west,blue] {$\gamma$};
		\draw[fill=black] (\nst*\x5,\nst*\y5) circle (\sts) node[anchor=west,blue] {$\gamma$};
	\end{tikzpicture}
  \end{center}
 \end{minipage}
 \begin{minipage}{0.315\linewidth}
  \begin{center}
	\begin{tikzpicture}
		\draw (\of1,0) node {$0,0$};
		\draw[thick] (\of1,0) circle(\rr) node[left=8pt,blue] {$p_j({\bf s})$};
		\draw[thick] (\of1+\x1,\y1) -- (\of1+\nst*\x1,\nst*\y1);
		\draw[thick] (\of1+\x2,\y2) -- (\of1+\nst*\x2,\nst*\y2);
		\draw[thick] (\of1+\x3,\y3) -- (\of1+\nst*\x3,\nst*\y3);
		\draw[thick] (\of1+\x4,\y4) -- (\of1+\nst*\x4,\nst*\y4);
		\draw[thick] (\of1+\x5,\y5) -- (\of1+\nst*\x5,\nst*\y5);
		\draw (\of1,\nst*\y4 -0.35) node {\sf{degree} $j$};
		\draw[fill=black] (\of1+\nst*\x1,\nst*\y1) circle (\sts) node[anchor=west,blue] {$\gamma$};
		\draw[fill=black] (\of1+\nst*\x2,\nst*\y2) circle (\sts) node[anchor=west,blue] {$\gamma$};
		\draw[fill=black] (\of1+\nst*\x3,\nst*\y3) circle (\sts) node[anchor=west,blue] {$\gamma$};
		\draw[fill=black] (\of1+\nst*\x4,\nst*\y4) circle (\sts) node[anchor=west,blue] {$\gamma$};
		\draw[fill=black] (\of1+\nst*\x5,\nst*\y5) circle (\sts) node[anchor=west,blue] {$\gamma$};
		\draw (\nlb*\x1,\nlb*\y1-0.1) node[above,rotate=72,blue] {$1$};
		\draw(\nlb*\x2,\nlb*\y2-0.1) node[above,rotate=180+130,blue] {$1$};
		\draw(\nlb*\x3,\nlb*\y3-0.1) node[above,rotate=180+216,blue] {$1$};
		\draw(\nlb*\x4,\nlb*\y4-0.1) node[above,rotate=288,blue] {$1$};
		\draw(\nlb*\x5,\nlb*\y5-0.1) node[above,rotate=0,blue] {$1$};
	\end{tikzpicture}
  \end{center}
 \end{minipage}
 \begin{minipage}{0.315\linewidth}
  \begin{center}
	\begin{tikzpicture}
		\draw (2*\of1,0) node {$\scriptstyle{\infty,0}$};
		\draw[thick] (2*\of1,0) circle(\rr) node[left=8pt,blue] {$p_j({\bf t})$};
		\draw[thick] (2*\of1+\x1,\y1) -- (2*\of1+\nst*\x1,\nst*\y1);
		\draw[thick] (2*\of1+\x2,\y2) -- (2*\of1+\nst*\x2,\nst*\y2);
		\draw[thick] (2*\of1+\x3,\y3) -- (2*\of1+\nst*\x3,\nst*\y3);
		\draw[thick] (2*\of1+\x4,\y4) -- (2*\of1+\nst*\x4,\nst*\y4);
		\draw[thick] (2*\of1+\x5,\y5) -- (2*\of1+\nst*\x5,\nst*\y5);
		\draw (2*\of1,\nst*\y4 -0.35) node {\sf{degree} $j$};
		\draw[fill=black] (2*\of1+\nst*\x1,\nst*\y1) circle (\sts) node[anchor=west,blue] {$\gamma$};
		\draw[fill=black] (2*\of1+\nst*\x2,\nst*\y2) circle (\sts) node[anchor=west,blue] {$\gamma$};
		\draw[fill=black] (2*\of1+\nst*\x3,\nst*\y3) circle (\sts) node[anchor=west,blue] {$\gamma$};
		\draw[fill=black] (2*\of1+\nst*\x4,\nst*\y4) circle (\sts) node[anchor=west,blue] {$\gamma$};
		\draw[fill=black] (2*\of1+\nst*\x5,\nst*\y5) circle (\sts) node[anchor=west,blue] {$\gamma$};
		\draw (\nlb*\x1,\nlb*\y1-0.1) node[above,rotate=72,blue] {$1$};
		\draw(\nlb*\x2,\nlb*\y2-0.1) node[above,rotate=180+130,blue] {$1$};
		\draw(\nlb*\x3,\nlb*\y3-0.1) node[above,rotate=180+216,blue] {$1$};
		\draw(\nlb*\x4,\nlb*\y4-0.1) node[above,rotate=288,blue] {$1$};
		\draw(\nlb*\x5,\nlb*\y5-0.1) node[above,rotate=0,blue] {$1$};
	\end{tikzpicture}
  \end{center}
 \end{minipage}
	\begin{minipage}{0.35\linewidth}
	 \tikzmath{
		 \nst=5;
		 \nlb=3;
		 }
		\begin{center}
	\begin{tikzpicture}
		\draw (3*\of1,0) node (a) {$i,j$};
		\draw[very thick] (3*\of1+\x1,\y1)  -- (3*\of1+\nst*\x1,\nst*\y1);
		\draw[very thick] (3*\of1+\x2,\y2)  -- (3*\of1+\nst*\x2,\nst*\y2);
		\draw[very thick] (3*\of1+\x3,\y3) -- (3*\of1+\nst*\x3,\nst*\y3);
		\draw[very thick] (3*\of1+\x4,\y4) -- (3*\of1+\nst*\x4,\nst*\y4);
		\draw[very thick] (3*\of1+\x5,\y5) -- (3*\of1+\nst*\x5,\nst*\y5);
		\draw[very thick,fill=white] (3*\of1-\rr,-\rr) rectangle (3*\of1+\rr,\rr);
		\draw (3*\of1,0) node (a) {$i,j_i$};
		\draw (a) node[left=0pt of a.west,blue] {$(-\beta d_i)^{-1}$};
		\draw (3*\of1+\nlb*\x1,\nlb*\y1) node[below,rotate=72,blue] {$-\beta d_i$};
		\draw(3*\of1+\nlb*\x2,\nlb*\y2-0.1) node[above,rotate=180+130,blue] {$-\beta d_i$};
		\draw(3*\of1+\nlb*\x3,\nlb*\y3) node[below,rotate=180+216,blue] {$-\beta d_i$};
		\draw(3*\of1+\nlb*\x4,\nlb*\y4-0.1) node[above,rotate=288,blue] {$-\beta d_i$};
		\draw(3*\of1+\nlb*\x5,\nlb*\y5-0.1) node[above,rotate=0,blue] {$-\beta d_i$};
		\draw[fill=black] (3*\of1+\nst*\x1,\nst*\y1) circle (\sts) node[anchor=west,blue] {$\gamma$};
		\draw[fill=black] (3*\of1+\nst*\x2,\nst*\y2) circle (\sts) node[anchor=east,blue] {$\gamma$};
		\draw[fill=black] (3*\of1+\nst*\x3,\nst*\y3) circle (\sts) node[anchor=east,blue] {$\gamma$};
		\draw[fill=black] (3*\of1+\nst*\x4,\nst*\y4) circle (\sts) node[anchor=west,blue] {$\gamma$};
		\draw[fill=black] (3*\of1+\nst*\x5,\nst*\y5) circle (\sts) node[anchor=west,blue] {$\gamma$};
	\end{tikzpicture}
		\end{center}
	\end{minipage}
\caption{\footnotesize  Weights. Left: coloured vertices and edges. Center: black vertices. Right: white vertices. 
The vertex degree is $j=5$ in all the examples.}
	 \label{fig:weightsVerticesAllTypes}
\end{figure}
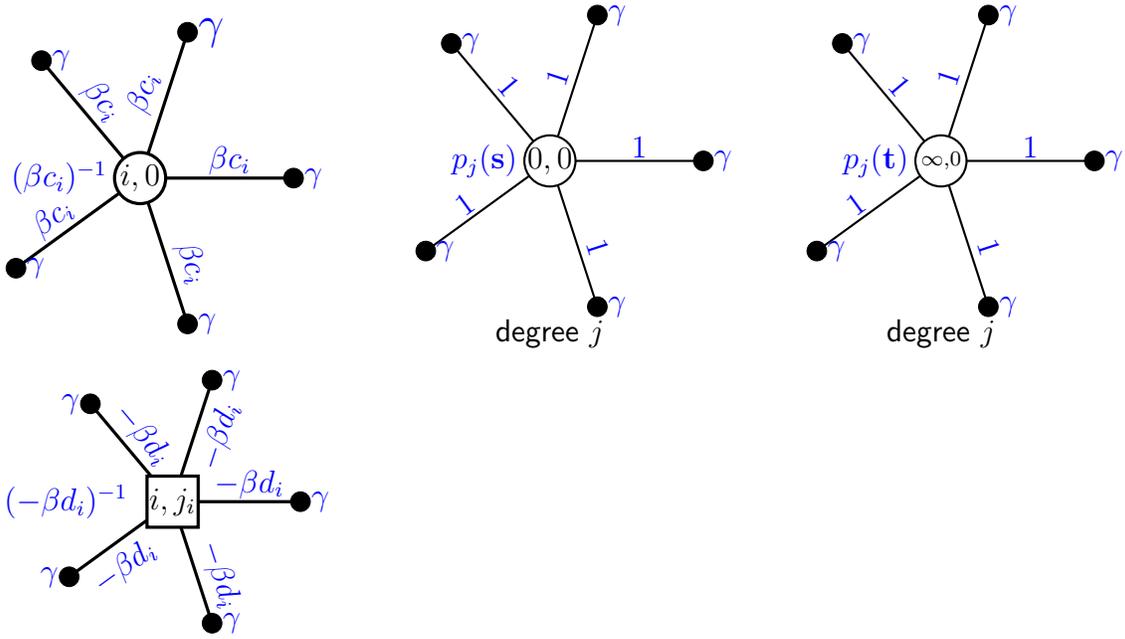
\bigskip

\begin{figure}[H]
\label{fig:weightsVerticesAllTypesCover}
 \begin{minipage}{0.4\linewidth}
  \begin{center}
	\begin{tikzpicture}
		\tikzmath{
			\y1=0;
			\y2=2;
			\y3=4;
			\y0=-1;
			\of1=3;
			\sts=0.13;
			\qsd=1.2*0.28;
			}
		\draw[very thick] (0,\y2) -- (\of1,0.5*\y2+0.5*\y3);
		\draw[very thick] (0,\y2) -- (\of1,\y2);
		\draw[very thick] (0,\y2) -- (\of1,\y1);
		\draw[blue] (0.5*\of1, 0.5*\y1+0.5*\y2+0.3) node[rotate=330] {$\beta c_i$};
		\draw[blue] (0.7*\of1, \y2+0.2) node {$\beta c_i$};
		\draw[blue] (0.6*\of1, 0.25*\y3+0.75*\y2+0.35) node[rotate=20] {$\beta c_i$};
		\draw[very thick,fill=white] (0,\y1) circle(\qsd);
		\draw[very thick,fill=white] (0,\y2) circle(\qsd);
		\draw[very thick,fill=white] (0,\y3) circle(\qsd);
		\draw[very thick] (0,\y1) node {$i,0$};
		\draw[very thick] (0,\y2) node {$i,0$};
		\draw[very thick] (0,\y3) node {$i,0$};
		\draw[fill=black] (\of1,0.5*\y2+0.5*\y3) circle (\sts);
		\draw[fill=black] (\of1,\y2) circle (\sts);
		\draw[fill=black] (\of1,\y1) circle (\sts);
		\draw[blue] (-0.3,\y1) node[left=3pt] {$(\beta c_i)^{-1}$};
		\draw[blue] (-0.3,\y2) node[left=3pt] {$(\beta c_i)^{-1}$};
		\draw[blue] (-0.3,\y3) node[left=3pt] {$(\beta c_i)^{-1}$};
		\draw[blue] (\of1,0.5*\y2+0.5*\y3) node[above=1pt] {$\gamma$};
		\draw[blue] (\of1,\y2) node[above=1pt] {$\gamma$};
		\draw[blue] (\of1,\y1) node[above=1pt] {$\gamma$};
		\draw[black] (\of1,0.5*\y2+0.5*\y3) node[right=2pt] {$\underline{a}$};
		\draw[black] (\of1,\y2) node[right=2pt] {$\underline{b}$};
		\draw[black] (\of1,\y1) node[right=2pt] {$\underline{c}$};
		\draw[black] (0+0.2,\y1+0.2) node[anchor=south west] {$Q^{(i,0)}_{\ell(\mu^{(i,0)})}$};
		\draw[black] (0+0.2,\y2+0.2) node[anchor=south west] {$Q^{(i,0)}_{k}$};
		\draw[black] (0+0.2,\y3+0.2) node[anchor=south west] {$Q^{(i,0)}_{1}$};
		\draw[black] (0,\y0) node {$Q^{(i,0)}$};
	\end{tikzpicture}
	\end{center}
	\end{minipage}
 \begin{minipage}{0.4\linewidth}
  \begin{center}
	\begin{tikzpicture}
		\tikzmath{
			\y1=0;
			\y2=2;
			\y0=-1;
			\y3=4;
			\of1=3;
			\sts=0.13;
			\qsd=1.2*0.28;
			}
		\draw[very thick] (0,\y2) -- (\of1,0.7*\y2+0.3*\y3);
		\draw[very thick] (0,\y2) -- (\of1,\y3);
		\draw[very thick] (0,\y2) -- (\of1,0.75*\y1+0.25*\y2);
		\draw[blue] (0.7*\of1, 0.5*\y1+0.5*\y2+0.2) node[rotate=330] {$1$};
		\draw[blue] (0.7*\of1, 0.3*\y3+0.7*\y2+0) node[rotate=20] {$1$};
		\draw[blue] (0.6*\of1, 0.7*\y3+0.3*\y2+0.1) node[rotate=35] {$1$};
		\draw[very thick,fill=white] (0,\y1) circle(\qsd);
		\draw[very thick,fill=white] (0,\y2) circle(\qsd);
		\draw[very thick,fill=white] (0,\y3) circle(\qsd);
		\draw[very thick] (0,\y1) node {$0,0$};
		\draw[very thick] (0,\y2) node {$0,0$};
		\draw[very thick] (0,\y3) node {$0,0$};
		\draw[fill=black] (\of1,\y3) circle (\sts);
		\draw[fill=black] (\of1,0.7*\y2+0.3*\y3) circle (\sts);
		\draw[fill=black] (\of1,0.75*\y1+0.25*\y2) circle (\sts);
		\draw[blue] (\of1,0.7*\y2+0.3*\y3) node[above=1pt] {$\gamma$};
		\draw[blue] (\of1,0.75*\y1+0.25*\y2) node[above=1pt] {$\gamma$};
		\draw[blue] (\of1,\y3) node[above=1pt] {$\gamma$};
		\draw[black] (\of1,0.7*\y2+0.3*\y3) node[right=2pt] {$\underline{a}$};
		\draw[black] (\of1,0.75*\y1+0.25*\y2) node[right=2pt] {$\underline{b}$};
		\draw[black] (\of1,\y3) node[right=2pt] {$\underline{c}$};
		\draw[blue] (-0.3,\y2) node[left=3pt] {$p_{\nu_j}({\bf t})$};
		\draw[black] (0+0.2,\y1+0.2) node[anchor=south west] {$Q^{(0,0)}_{\ell(\mu^{(i,0)})}$};
		\draw[black] (0+0.2,\y2+0.2) node[anchor=south] {$Q^{(0,0)}_{k}$};
		\draw[black] (0+0.2,\y3+0.2) node[anchor=south west] {$Q^{(0,0)}_{1}$};
		\draw[black] (0,\y0) node {$Q^{(0,0)}$};
	\end{tikzpicture}
	\end{center}
	\end{minipage}
	\newline
 \begin{minipage}{0.4\linewidth}
  \begin{center}
	\begin{tikzpicture}
		\tikzmath{
			\y1=0;
			\y2=2;
			\y3=4;
			\of1=3;
			\y0=-1;
			\sts=0.13;
			\qsd=1.2*0.28;
			}
		\draw[very thick] (0,\y2) -- (\of1,0.7*\y2+0.3*\y1);
		\draw[very thick] (0,\y2) -- (\of1,\y3);
		\draw[very thick] (0,\y2) -- (\of1,0.75*\y1+0.25*\y2);
		\draw[blue] (0.7*\of1, 0.5*\y1+0.5*\y2+0.8) node[rotate=350] {$-\beta d_i$};
		\draw[blue] (0.7*\of1, 0.75*\y1+0.25*\y2+0.1) node[rotate=335] {$-\beta d_i$};
		\draw[blue] (0.6*\of1, 0.5*\y3+0.5*\y2+0.45) node[rotate=35] {$-\beta d_i$};
		\draw[very thick,fill=white] (-\qsd,\y1-\qsd) rectangle(\qsd,\y1+\qsd);
		\draw[very thick,fill=white] (-\qsd,\y2-\qsd) rectangle(\qsd,\y2+\qsd);
		\draw[very thick,fill=white] (-\qsd,\y3-\qsd) rectangle(\qsd,\y3+\qsd);
		\draw[very thick] (0,\y1) node {$i,j_i$};
		\draw[very thick] (0,\y2) node {$i,j_i$};
		\draw[very thick] (0,\y3) node {$i,j_i$};
		\draw[fill=black] (\of1,\y3) circle (\sts);
		\draw[fill=black] (\of1,0.7*\y2+0.3*\y1) circle (\sts);
		\draw[fill=black] (\of1,0.75*\y1+0.25*\y2) circle (\sts);
		\draw[black] (0+0.2,\y1+0.2) node[anchor=south] {$Q^{(i,j_i)}_{\ell(\mu^{(i,0)})}$};
		\draw[black] (0+0.2,\y2+0.2) node[anchor=south] {$Q^{(i,j_i)}_{k}$};
		\draw[black] (0+0.2,\y3+0.2) node[anchor=south west] {$Q^{(i,j_i)}_{1}$};
		\draw[black] (0,\y0) node {$Q^{(i,j_i)}$};
		\draw[blue] (\of1,0.7*\y2+0.3*\y1) node[above=1pt] {$\gamma$};
		\draw[blue] (\of1,0.75*\y1+0.25*\y2) node[above=1pt] {$\gamma$};
		\draw[blue] (\of1,\y3) node[above=1pt] {$\gamma$};
		\draw[black] (\of1,0.7*\y2+0.3*\y1) node[right=2pt] {$\underline{a}$};
		\draw[black] (\of1,0.75*\y1+0.25*\y2) node[right=2pt] {$\underline{b}$};
		\draw[black] (\of1,\y3) node[right=2pt] {$\underline{c}$};
		\draw[blue] (-0.3,\y1) node[left=3pt] {$(-\beta d_i)^{-1}$};
		\draw[blue] (-0.3,\y2) node[left=3pt] {$(-\beta d_i)^{-1}$};
		\draw[blue] (-0.3,\y3) node[left=3pt] {$(-\beta d_i)^{-1}$};
	\end{tikzpicture}
	\end{center}
	\end{minipage}
 \begin{minipage}{0.4\linewidth}
  \begin{center}
	\begin{tikzpicture}
		\tikzmath{
			\y1=0;
			\y2=2;
			\y3=4;
			\of1=3;
			\sts=0.13;
			\qsd=1.2*0.28;
			\y0=-1;
			}
		\draw[very thick] (0,\y2) -- (\of1,0.7*\y2+0.3*\y1);
		\draw[very thick] (0,\y2) -- (\of1,\y3);
		\draw[very thick] (0,\y2) -- (\of1,0.75*\y1+0.25*\y2);
		\draw[blue] (0.7*\of1, 0.5*\y1+0.5*\y2+0.8) node[rotate=350] {$1$};
		\draw[blue] (0.7*\of1, 0.75*\y1+0.25*\y2+0.1) node[rotate=335] {$1$};
		\draw[blue] (0.6*\of1, 0.5*\y3+0.5*\y2+0.45) node[rotate=35] {$1$};
		\draw[very thick,fill=white] (0,\y1) circle(\qsd);
		\draw[very thick,fill=white] (0,\y2) circle(\qsd);
		\draw[very thick,fill=white] (0,\y3) circle(\qsd);
		\draw[very thick] (0,\y1) node {$\scriptstyle{\infty,0}$};
		\draw[very thick] (0,\y2) node {$\scriptstyle{\infty,0}$};
		\draw[very thick] (0,\y3) node {$\scriptstyle{\infty,0}$};
		\draw[fill=black] (\of1,\y3) circle (\sts);
		\draw[fill=black] (\of1,0.7*\y2+0.3*\y1) circle (\sts);
		\draw[fill=black] (\of1,0.75*\y1+0.25*\y2) circle (\sts);
		\draw[blue] (\of1,0.7*\y2+0.3*\y1) node[above=1pt] {$\gamma$};
		\draw[blue] (\of1,0.75*\y1+0.25*\y2) node[above=1pt] {$\gamma$};
		\draw[blue] (\of1,\y3) node[above=1pt] {$\gamma$};
		\draw[black] (\of1,0.7*\y2+0.3*\y1) node[right=2pt] {$\underline{a}$};
		\draw[black] (\of1,0.75*\y1+0.25*\y2) node[right=2pt] {$\underline{b}$};
		\draw[black] (\of1,\y3) node[right=2pt] {$\underline{c}$};
		\draw[blue] (-0.3,\y2) node[left=3pt] {$p_{\mu_j}({\bf s})$};
		\draw[black] (0+0.2,\y1+0.2) node[anchor=south west] {$Q^{(\infty,0)}_{\ell(\mu^{(i,0)})}$};
		\draw[black] (0+0.2,\y2+0.2) node[anchor=south] {$Q^{(\infty,0)}_{k}$};
		\draw[black] (0+0.2,\y3+0.2) node[anchor=south west] {$Q^{(\infty,0)}_{1}$};
		\draw[black] (0,\y0) node {$Q^{(\infty,0)}$};
	\end{tikzpicture}
	\end{center}
	\end{minipage}
	\caption{\footnotesize  Weights. Left: coloured vertices and edges. Right: black and white vertices.}
	 \label{fig:weightsVerticesAllTypes}
\end{figure}
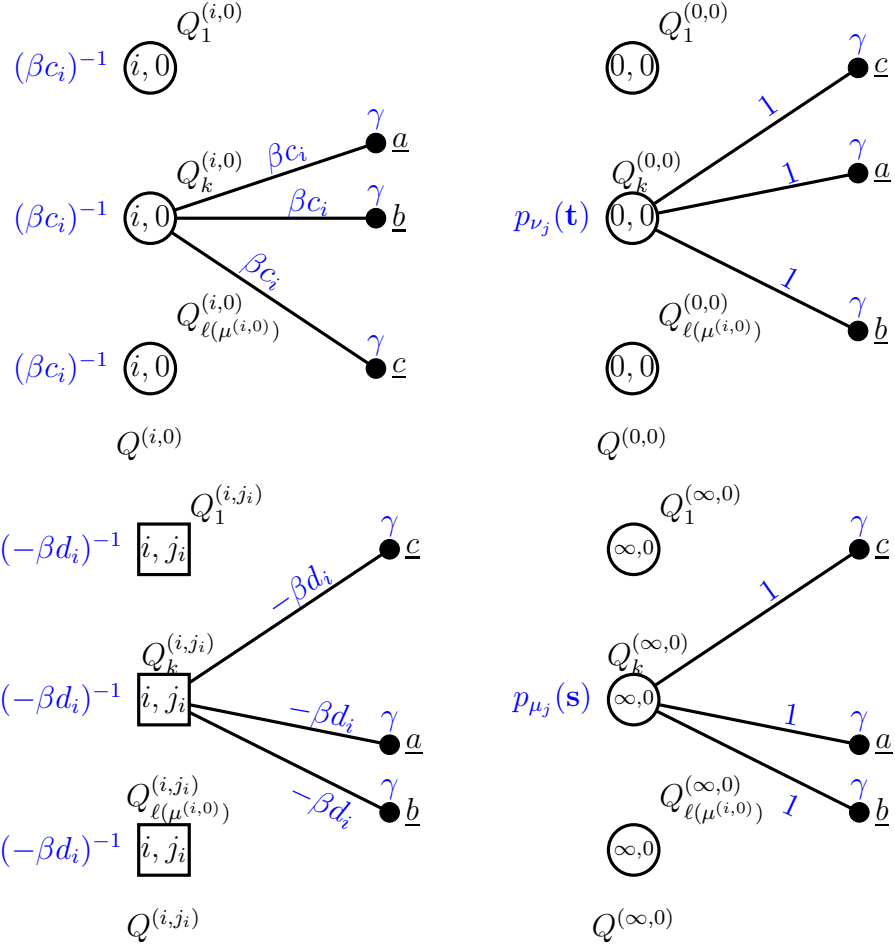

\begin{definition}
The sum of the colengths of the partitions $(\mu^{(1,0)},...,\mu^{(L,0)},\nu^{(1,1)}, \dots,\nu^{(M,J_M)})$
is denoted as
\be
d(\GG^J_{{\bf c}, {\bf d}}) := \sum_{i=1}^L \ell^*(\mu^{(i,0)}) + \sum_{i=1}^M\sum_{j_i=1}^{J_i} \ell^*(\nu^{(i, j_i)})
\label{d_GG_cd}
\ee
and the degree as
\be
N(\GG^J_{{\bf c}, {\bf d}}) =N.
\label{N_GG_cd}
\ee
\end{definition}
We also introduce the following abbreviated notation for ordered  sets of partitions of these two types
	\be
\mub(\GG^J_{{\bf c}, {\bf d}})):= (\mu^{(1)}, \dots, \mu^{(L)}), \quad \nub(\GG^J_{{\bf c}, {\bf d}}) := (\nu^{(1,1)}, \dots,\nu^{(M,J_M)}).
\ee

\begin{lemma}
	Let $\GG^J_{{\bf c},{\bf d}}$ be a doubly labelled, rationally weighted constellation
	in the equivalence class determined by $(\nu(\GG^J_{{\bf c}, {\bf d}}), \mub(\GG^J_{{\bf c}, {\bf d}}), \nub(\GG^J_{{\bf c}, {\bf d}}), \mu(\GG^J_{{\bf c}, {\bf d}}))$,
	with the  weights of vertices and edges determined by the weight generating function $G_{{\bf c},{\bf d}}$.
	Its total weight $W_{\GG^J_{{\bf c},{\bf d}}}$ is then given by
	\be
		W_{\GG^J_{\mathbf{c},\mathbf{d}}}={\gamma^{N}\over N!}\beta^{d(\GG^J_{{\bf c}, {\bf d}})} (-1)^{|J|+\sum_{i=1}^{M}\sum_{j_i=1}^{J_i} \ell^{*}(\nu^{(i,j_i)})}
	c^{\ell^{*}(\mu^{(1,0)})}_{1}\dots c^{\ell^{*}(\mu^{(L,0)})}_{L}
	\prod_{i=1}^{M}\prod_{j_i=1}^{J_i}d^{\ell^{*}(\nu^{(i,j_i)})}_{i} p_{\mu}({\bf t})p_{\nu}({\bf s}).
	\label{double_label_weight}
	\ee
\end{lemma}

\begin{proof}
	Consider a doubly labelled constellation with spectrum $J$ from the equivalence class
	$(\nu(\GG^J_{{\bf c}, {\bf d}}), \mub(\GG^J_{{\bf c}, {\bf d}}), \nub(\GG^J_{{\bf c}, {\bf d}}),\mu(\GG^J_{{\bf c}, {\bf d}}))$.
		The star vertices give an overall factor $\gamma^N$.
	Round vertices of colour $i$ connected to $m$ star vertices give $(\beta c_i)^{m-1}$. 
	The product over all round vertices of the same colour therefore gives 
	\be
	(\beta c_i)^{\sum_{k=1}^{\ell(\mu^{(i,0)})}(\mu^{(i,0)}_k-1)}=(\beta c_i)^{\ell^{*}(\mu^{(i,0)})}\, ,
	\ee
	while those over square vertices of the same colour gives
	\be
	(-1)^{\sum_{j=1}^{j_i}\ell^{*}(\nu^{(i,j_i)})}(\beta d_i)^{\sum_{j=1}^{j_i}\ell^{*}(\nu^{(i,j_i)})}\,.
	\ee
	A white vertex of degree $\mu_i$ contributes $p_{\mu_i}({\bf t})$, while a black vertex
	of degree $\nu_i$  gives $p_{\nu_i}({\bf s})$;  taking their product gives 
	\be
	p_{\mu}({\bf t})p_{\nu}({\bf s}) =\prod_{i=1}^{\ell(\mu)}p_{\mu_i}({\bf t}) \prod_{i=1}^{\ell(\nu)}p_{\nu_i}({\bf s})\,.
	\ee
	The weight of the entire graph is therefore given by (\ref{double_label_weight}). 
\end{proof}
\begin{remark}
Note that, for any choice of $(L, M,J)$ the weight (\ref{double_label_weight}) for a doubly weighted
constellation can be obtained from that for its singly weighted analog by simply replacing the successive
weighting parameters $(c_1, c_2, \dots)$ in the singly weighted sequence by the corresponding 
sequence of $c_i$'s and $d_i$'s, with the latter repeated $J_i$ times in the CF ordered sequence,
and multiplying by the sign factor $(-1)^{|J|+\sum_{i=1}^{M}\sum_{j_i=1}^{J_i} \ell^{*}(\nu^{(i,j_i)})}$.
\end{remark}

In the examples  above, we may view Figure \ref{fig:planar_N5_constell_1}
as displaying a doubly labelled constellation in which $L=3$ and $M=0$, with
the labels $((0), (1), (2), (3), (4))$ replaced by $((0,0), (1,0), (2,0), (3,0), (4,0))$. This then corresponds to a 
polynomial weight generating function $G_{(c_1,c_2,c_3),(\emptyset)}$ of third degree and,  
according to eq.~(\ref{double_label_weight}),  the total weight  of the constellation is
\be
W_{\GG^0_{(c_1,c_2,c_3), (\emptyset)}} = {1\over 5!} \gamma^5 \beta^5 c_1^2 c_2^2 c_3\ p_{(2,1,1,1)}({\bf t})\  p_{(3,1,1)}({\bf s}).
\ee
The doubly labelled planar constellation Figure \ref{fig:planar_N5_constell_2} corresponds to
$(L,M, J)=(1,1, (2))$ with weight generating function $G_{((c_1), (d_1))}$ and has total weight
\be
W_{\GG^2_{(c_1), (d_1)}} = -{1\over 5!} \gamma^5 \beta^5 c_1^2 d_1^3\   p_{(2,1,1,1)}({\bf t})\  p_{(3,1,1)}({\bf s}).
\ee


\subsection{Reconstruction of  $\tau^{(G_{{\bf c}, {\bf d}}, \beta, \gamma)} ({\bf t}, {\bf s})$}
\label{tau_func_constell}

We now show how to recover the generating function $\tau^{(G_{{\bf c}, {\bf d}}, \beta, \gamma)} ({\bf t}, {\bf s})$ defined in  eq.~(\ref{tau_G_cd_H}) by summing over the total weights of the doubly weighted constellations defined above.
	In the following, it is helpful to introduce the following notations for sets of positive integer vectors and partitions.
	\begin{definition}
	For positive integers $(l, m, L, M)$ with $1\le l \le L$, let
	\bea
		A_{l,L}&\&:=\{{\bf a}=(a_1, \dots, a_l) \in \Nb^l \, | \,1\leq a_1<a_2<\dots<a_l\leq L\}\,,\quad 1\leq l\leq L\,, \\
		B_{m,M}&\&:=\{{\bf b}= (b_1, \dots, b_m)\in \Nb^m| 1\leq b_1\leq b_2 \leq \dots\leq b_m\leq M\}.
	\eea
	\end{definition}
	\begin{definition}
	\label{def_m_german}
	For positive integers $(N, d, k)$, let 
	\be
	\mathfrak{M}^{(N)}_{k,d}:=\lbrace \mub\rbrace\,,\quad \mub=\left(\mu^{(1)},\dots,\mu^{(k)}\right)
	\ee
       denote  the set of (ordered) $k$-tuples of nontrivial partitions $\left(\mu^{(1)},\dots,\mu^{(k)}\right)$ of $N$ with total colength  
       $d$ as defined in eq.~(\ref{d_def}).
	The cardinality of  $\mub$ will be denoted  
	\be
	\#\mub=k.
	\ee
	More generally,  for a $k$-tuple of nonnegative integers
       \be
       J = (J_1, \dots, J_k), \quad J_i \in \Nb^+
       \ee
       let
       \be
	\mathfrak{M}^{(N)}_{J,d}:=\{\left(\mub^{(1)},\dots,\mub^{(k)}\right)\},
	\ee
	where
	\be
	\mub^{(i)} = (\mu^{(i, 1)}, \dots, \mu^{(i, J_i)})
	\ee
	are $J_i$-tuples of partitions of $N$
	and
	\be
	d =\sum_{i=1}^k \sum_{j_i=1}^{J_i}\ell^*( \mu^{(i, j_i)})
	\ee
	is the total colength.
	\end{definition}
	In general, boldface Greek letters will denote ordered multiples of partitions
	\be
	\nub=\left(\nu^{(1)},\dots,\nu^{(m)}\right)
	\ee
	of any cardinality
	\be
	\#(\nub) =m.
	\ee
	\begin{definition}
	\label{def_n_german}
	Define the subset $\mathfrak{N}^{(N)}_{k,d}\subset \mathfrak{M}^{(N)}_{k,d}$
	 to consist of weakly ordered $k$-tuples of partitions
	\be
		\mathfrak{N}^{(N)}_{k,d}:=\{\mub\in\mathfrak{M}^{(N)}_{k,d}\;|\;\mu^{(1)}\leq \mu^{(2)}\leq \dots \leq \mu^{(k)}\}\,.
	\ee
	   with respect to lexicographical ordering \cite{Mac}.
	      \end{definition}
	      The action of the symmetric group $\SS_k$ on $\mathfrak{M}^{(N)}_{k,d}$ is denoted
	\be
		\sigma(\mub)=\left(\mu^{(\sigma(1))},\mu^{(\sigma(2))},\dots,\mu^{(\sigma(k))}\right)\,,\quad \sigma \in \mathcal{S}_k\,.
	\ee
		We also use the abbreviated notations
	\bea
		H(\mub,\nub,\mu,\nu)&\&:=H(\mu^{(1)},\dots,\mu^{(l)},\nu^{(1)},\dots,\nu^{(m)},\mu,\nu)\\
		\WW_{G_{{\bf c},{\bf d}}}(\mub;\nub)&\&:=\WW_{G_{{\bf c},{\bf d}}}
		(\mu^{(1)},\dots,\mu^{(l)};\nu^{(1)},\dots,\nu^{(m)})
	\eea
	 for pure Hurwitz numbers and weights.

\begin{remark}
	 The definitions \ref{def_m_german}, \ref{def_n_german} immediately imply the following tautological properties.
		 For a triple of positive integers $(N,k,d)$  and an arbitrary map $X$
		 \be
			X: \mathfrak{M}^{(N)}_{k,d} \rightarrow \mathbb{C}
		\ee
		 the following identity holds:
		\be
		 \label{ordunord}
			\sideset{}{'}\sum_{\mub \in \mathfrak{M}^{(N)}_{k,d}}X(\mub)=
			\sideset{}{'}\sum_{\mub \in \mathfrak{N}^{(N)}_{k,d}} \frac{1}{|\aut(\mub)|}\sum_{\sigma \in \SS_{k}} X(\sigma(\mub))
		\ee
		 where $\aut(\mub)$ denotes the stabilizer of $\mub$ in $\SS_{k}$.
		 For an arbitrary element $\mub$ of the set $\mathfrak{M}^{(N)}_{k,d}$ the following is also true:
	\be
		 \label{sumaverage}
		 \sum_{\sigma_1 \in \SS_{k}}\sum_{\sigma_2 \in \SS_{k}} 
		 X (\sigma_1(\sigma_2(\mub)))=k! \sum_{\sigma_1 \in \SS_{k}}X(\sigma_1(\mub))\,.
	\ee
\end{remark}
	 
Define the numbers $W_{G_{{\bf c},{\bf d}}}(\mub;\nub)$ (not to be confused  with $\WW_{G_{{\bf c},{\bf d}}}$)  to be
\be
	W_{G_{{\bf c},{\bf d}}}(\mub;\nub) :=
	(-1)^{\sum_{j=1}^{\#\nub} \ell^*(\nu^{(j)}) -\#\nub }
	\sum_{{\bf a}\in A_{\#\mub,L}}\sum_{{\bf b}\in B_{\#\nub,M}}
	\prod_{i=1}^{\#\mub} c_{a_i}^{\ell^{*}(\mu^{(i)})}
	\prod_{j=1}^{\#\nub} d_{b_j}^{\ell^{*}(\nu^{(j)})} 
	\ee
\begin{lemma}
	The $\tau$-function $\tau^{(G_{\bf c,\bf d}, \beta, \gamma)}({\bf t},{\bf s})$ can be expressed as
	\bea
	\label{tausum}
	\tau^{(G_{\bf c,\bf d}, \beta, \gamma)}({\bf t},{\bf s})=
	\sum_{d=0}^{\infty}\beta^d \sum_{N=0}^{\infty}\gamma^{N}
	\sum_{\substack{\mu,\nu\\|\mu|=|\nu|=N}}
	\sum_{d_{+}=0}^{d}
		\sum_{l=0}^{d_{+}}
\sideset{}{'}
		\sum_{\mub \in \mathfrak{M}^{(N)}_{l,d_{+}}}
		\sum_{m=0}^{d-d_{+}}\sideset{}{'}
		\sum_{\nub\in \mathfrak{M}^{(N)}_{m,d-d_{+}}}
		&&
		W_{G_{{\bf c},{\bf d}}}(\mub;\nub)\cr
		&&\times H(\boldsymbol{\mu},\nub,\mu,\nu) p_{\mu}(\mathbf{s})p_{\nu}(\mathbf{t}).
		\cr
		&\&
	\eea
	(Recall that  $\sideset{}{'}\sum$ denotes summation only over nontrivial partitions.)
\end{lemma}
\begin{proof}
	We start with formula (\ref{tau_G_cd_H}) for $\tau^{(G_{\bf c,\bf d}, \beta, \gamma)}$:
	\be
	\label{tausym}
	\tau^{(G_{\bf c,\bf d}, \beta, \gamma)}({\bf t},{\bf s})
		=\sum_{d=0}^{\infty} \beta^d\sum_{N=0}^{\infty}\gamma^{N}\sum_{\substack{\mu,\nu\\|\mu|=|\nu|=N}} H^d_{G_{\mathbf{c},\mathbf{d}}}(\mu,\nu) p_{\mu}(\mathbf{s})p_{\nu}(\mathbf{t})
	\ee
	where, by eq.~(\ref{H_d_G_rat_def}), the weighted Hurwitz numbers $H^{d}_{G_{\mathbf{c},\mathbf{d}}}(\mu,\nu)$ are given by
	\be
		H^{d}_{G_{\mathbf{c},\mathbf{d}}}(\mu,\nu)=
		\sum_{d_{+}=0}^{d}\sum_{l=0}^{d_{+}}\sum_{m=0}^{d_-}
		\sideset{}{'}\sum_{\nub\in \mathfrak{M}^{(N)}_{l,d_{+}}}
		\sideset{}{'}\sum_{\nub\in \mathfrak{M}^{(N)}_{m,d_-}}
		\WW_{G_{{\bf c},{\bf d}}}(\boldsymbol{\mu};\nub) H(\boldsymbol{\mu},\nub,\mu,\nu).
      \label{whnsym}
	\ee
	where the last two sums are over pairs $(\mub, \nub)$ for which the sum of colengths is $d$
	\bea
	d_+ +d _- &\&=d, \cr
	d_+ = d_+(\mub) := \sum_{i=1}^{\#(\mub)}\ell^*( \mu^{(i)}), &\& \quad d_- = d_-(\nub) := \sum_{j=1}^{\#(\nub)} \ell^*(\nu^{(j)}).
	\eea
	By eq.~(\ref{WG_cd_def}), the weight $\WW_{G_{{\bf c},{\bf d}}}$ is
\bea
&\& \WW_{G_{{\bf c}, {\bf d}}}(\mu^{(1)}, \dots, \mu^{(l)}; \nu^{(1)}, \dots, \nu^{(m)}) \cr
&\& :=
	\frac{(-1)^{\sum_{j=1}^m \ell^*(\nu^{(j)}) -m }}{l!m!}
	\sum_{{\bf a}\in A_{l,L}}\sum_{{\bf b}\in B_{m,M}}
	\sum_{\sigma\in \mathcal{S}_l}\sum_{\sigma'\in \mathcal{S}_m} 
  c_{a_{\sigma(1)}}^{\ell^*(\mu^{(1)})} \cdots c_{a_{\sigma(l)}}^{\ell^*(\mu^{(l)})}  
  d_{b_{\sigma'(1)}}^{\ell^*(\nu^{(1)})} \cdots d_{b_{\sigma'(m)}}^{\ell^*(\nu^{(m)})}.\cr
  &\&
  \label{WW_Gcd_perm1}
\eea
	Since $\WW_{G_{{\bf c},{\bf d}}}$ is invariant under permutations of the partitions $\{\mu^{(i)}\}$ and $\{\nu^{(j)}\}$
	amongst themselves, we can write equivalently
\bea
&\& \WW_{G_{{\bf c}, {\bf d}}}(\mu^{(1)}, \dots, \mu^{(l)}; \nu^{(1)}, \dots, \nu^{(m)}) \cr
&\& :=
	\frac{(-1)^{\sum_{j=1}^m \ell^*(\nu^{(j)}) -m }}{l!m!}
	\sum_{{\bf a}\in A_{l,L}}\sum_{{\bf b}\in B_{m,M}}
	\sum_{\sigma\in \mathcal{S}_l}\sum_{\sigma'\in \mathcal{S}_m} 
	c_{a_1}^{\ell^*(\mu^{(\sigma(1))})} \cdots c_{a_{l}}^{\ell^*(\mu^{(\sigma(l))})}  
	d_{b_{1}}^{\ell^*(\nu^{(\sigma'(1))})} \cdots d_{b_{m}}^{\ell^*(\nu^{(\sigma'(m))})}.\cr
  &\&
  \label{WW_Gcd_perm2}
\eea
	$\WW_{G_{{\bf c},{\bf d}}}$ is thus an average of the weights $W_{G_{{\bf c},{\bf d}}}(\boldsymbol{\mu};\nub)$
over permutations of partitions
\be
	\WW_{G_{{\bf c},{\bf d}}}(\boldsymbol{\mu};\nub)
	=\frac{1}{(\#\boldsymbol{\mu}!)(\#\nub!)}\sum_{\sigma\in \mathcal{S}_{\#\boldsymbol{\mu}}}\sum_{\sigma'\in \mathcal{S}_{\#\nub}} 
	W_{G_{{\bf c},{\bf d}}}(\sigma(\boldsymbol{\mu});\sigma'(\nub)).
		\label{weightav}
\ee
For any pair of integer $(d_-,d_+)$ such that 
\be
l\leq d_+\leq d, \quad d_- + d_+ =d,
\ee 
the two right-most sums in  (\ref{whnsym}) can  therefore be rewritten as
\bea
	&\&\sideset{}{'}
	\sum_{\boldsymbol{\mu} \in \mathfrak{M}^{(|\mu|)}_{l,d_{+}}}\sideset{}{'}
	\sum_{\nub\in \mathfrak{M}^{(|\nu|)}_{m,d_-}}
	\WW_{G_{{\bf c},{\bf d}}}(\boldsymbol{\mu};\nub)
	H(\boldsymbol{\mu},\nub,\mu,\nu)  \cr
	=&\& \frac{1}{l!m!}\sideset{}{'}\sum_{\boldsymbol{\mu} \in \mathfrak{M}^{(|\mu|)}_{l,d_{+}}}\sideset{}{'}
	\sum_{\nub\in \mathfrak{M}^{(|\nu|)}_{m,d_-}} \sum_{\sigma\in \mathcal{S}_{l}}\sum_{\sigma'\in \mathcal{S}_{m}} 
	W_{G_{{\bf c},{\bf d}}}(\sigma(\boldsymbol{\mu});\sigma'(\nub)) H(\boldsymbol{\mu},\nub,\mu,\nu)\cr
		= &\&
	\sideset{}{'}\sum_{\boldsymbol{\mu} \in \mathfrak{N}^{(|\mu|)}_{l,d_{+}}}
\sideset{}{'}
	\sum_{\nub\in \mathfrak{N}^{(|\mu|)}_{m,d_-}}
	\frac{1}{|\aut(\boldsymbol{\mu})||\aut(\nub)|}\sum_{\sigma\in \mathcal{S}_{l}}
	\sum_{\sigma'\in \mathcal{S}_{m}     }
	W_{G_{{\bf c},{\bf d}}}(\sigma(\boldsymbol{\mu});\sigma'(\nub)) 
	H(\boldsymbol{\mu},\nub,\mu,\nu)\,, 
	\label{symmetrized_WW_H_sum}
\eea
where we have used the identities (\ref{ordunord}) and (\ref{sumaverage}) to pass from the second to the third line.
On the other hand, the identity (\ref{ordunord}) implies that
\bea
	&\&\sideset{}{'}
	\sum_{\boldsymbol{\mu} \in \mathfrak{M}^{(|\mu|)}_{l,d_{+}}}\sideset{}{'}
	\sum_{\nub\in \mathfrak{M}^{(|\nu|)}_{m,d_-}}
	W_{G_{{\bf c},{\bf d}}}(\boldsymbol{\mu};\nub)
	H(\boldsymbol{\mu},\nub,\mu,\nu)  \cr
		= &\&
	\sideset{}{'}\sum_{\boldsymbol{\mu} \in \mathfrak{N}^{(|\mu|)}_{l,d_{+}}}
\sideset{}{'}
	\sum_{\nub\in \mathfrak{N}^{(|\mu|)}_{m,d_-}}
	\frac{1}{|\aut(\boldsymbol{\mu})||\aut(\nub)|}\sum_{\sigma\in \mathcal{S}_{l}}
	\sum_{\sigma'\in \mathcal{S}_{m}     }
	W_{G_{{\bf c},{\bf d}}}(\sigma(\boldsymbol{\mu});\sigma'(\nub)) 
	H(\boldsymbol{\mu},\nub,\mu,\nu)\,, 
\label{symmetrized_W_H_sum}
\eea
Comparing (\ref{symmetrized_WW_H_sum}) with (\ref{symmetrized_W_H_sum}) we obtain the following relation:
\be
\label{sumunord}
	\sideset{}{'}\sum_{\boldsymbol{\mu} \in \mathfrak{M}^{(|\mu|)}_{l,d_{+}}}\sideset{}{'}\sum_{\nub\in \mathfrak{M}^{(|\mu|)}_{m,d_-}}
	\WW_{G_{{\bf c},{\bf d}}}(\boldsymbol{\mu};\nub)  H(\boldsymbol{\mu},\nub,\mu,\nu) =
	\sideset{}{'}\sum_{\boldsymbol{\mu} \in \mathfrak{M}^{(|\mu|)}_{l,d_{+}}}\sideset{}{'}\sum_{\nub\in \mathfrak{M}^{(|\mu|)}_{m,d_-}}
	W_{G_{{\bf c},{\bf d}}}(\boldsymbol{\mu};\nub)
	H(\boldsymbol{\mu},\nub,\mu,\nu).
\ee
Substituting (\ref{sumunord}) into eqs.~(\ref{tausym}),(\ref{whnsym}),  we recover the statement of the lemma.

\end{proof}
\begin{theorem}
\label{tau_function_rat_weighted_constell}
	The $\tau$-function $\tau^{(G_{{\bf c},{\bf d}}, \beta, \gamma)}({\bf s},{\bf t})$  
		is equal to the sum of the total weights of all doubly labelled, rationally weighted constellations
		with weight generating function $G_{{\bf c}, {\bf d}}$.
\end{theorem}
\begin{proof}
	Note first that terms with $\#\boldsymbol{\mu}>L$ are absent from the sum (\ref{tausum}) 
	since there cannot be more than $L$  distinct indices $a_i$.
	Given  a set of indices $\mathbf{a}\in A_{l,L}$ and an $l$-tuple of nontrivial partitions $\boldsymbol{\mu} \in \mathfrak{M}^{(N)}_{l,d_{+}}$,
	we construct an $L$-tuple of (possibly trivial) partitions $\tilde{\boldsymbol{\mu}}(\boldsymbol{\mu},{\bf a})$
	\be
		\tilde{\boldsymbol{\mu}}(\boldsymbol{\mu},{\bf a})=\left(\tilde{\mu}^{(1)},\dots,\tilde{\mu}^{(L)}\right),
	\ee
	where $l$ of the partitions $\{\tilde{\mu}^{(i)}\}_{i=1, \dots, L}$ coincide with the elements of $\boldsymbol{\mu}$, and the rest are trivial:
	\be
		\tilde{\mu}^{(a_i)}=\mu^{(i)}\,,\quad 1\leq i \leq l.
	\ee
	\be
	\tilde{\boldsymbol{\mu}}(\boldsymbol{\mu},{\bf a})=\left(\underbrace{1^N,\dots,1^{N}}_{a_1-1\;\text{times}},\mu^{(1)},\underbrace{1^{N},\dots,1^{N}}_{a_2-a_1-1\;\text{times}},\mu^{(2)},1^{N},\dots,1^{N},\mu^{(l)},\underbrace{1^{N},\dots,1^{N}}_{L-a_l\;\text{times}}\right)
	\ee
	Conversely, from an $L$-tuple of partitions $\tilde{\boldsymbol{\mu}}$ we construct a set of indices 
	${\bf a}(\tilde{\mub}) \in A_{l, L} $ and an $l$-tuple of partitions $\mub(\tilde{\mub})$,
	where $l$ is the number of nontrivial partitions in $\tilde{\mub}$, ${\bf a} \in A_{l, L}$ is the
	 set of indices $i$ such that $\tilde{\mu}^{(i)}$ is nontrivial and the $l$-tuple of partitions $\mub(\tilde{\mub})$ 
	 is obtained by removing all trivial elements from $\mub$  while retaining the ordering.
	Note that for any ${\bf a} \in A_{l, L} $  where $l=\#\mub\leq L$,  the pure Hurwitz numbers 
	and weights are unchanged if  $\mub$   is replaced by $\tilde{\mub}$:
	\bea
		H(\tilde{\mub}(\mub,{\bf a}),\dots)&\&=H(\mub,\dots)
		\label{Hexpand}  \\
	\label{Wexpand}
		W_{G_{{\bf c},{\bf d}}}(\mub;\dots)&\&=
		W_{G_{{\bf c},{\bf d}}}(\tilde{\mub}(\mub,{\bf a});\dots)
	\eea	
	
The denominator requires a bit more care. Given an $m$-tuple of possibly coinciding indices $\mathbf{b}\in B_{m,M}$,
and  an $m$-tuple of non-trivial partitions $\nub\in \mathfrak{M}^{(N)}_{m,d_{-}}$ for any integer $i$ between $1$ and $M$,
 we identify the nonnegative integer $J_i({\bf b})$  to be the number of elements of $\mathbf{b}$ equal to $i$, denoting
  \be
J({\bf b}):= (J_1({\bf b},) \dots J_M({\bf b})),
 \ee
 and define ${\bf n}^{(i)}(\mathbf{b})$ to be the $J_i$-tuple of such elements:
\be
	{\bf n}^{(i)}({\bf b})=\left(n^{(i)}_{1},\dots,n^{(i)}_{J_i}\right)\,,
\ee
\be
	n^{(i)}_{1}<n^{(i)}_{2}<\dots<n^{(i)}_{J_i}\,,\quad b_{n^{(i)}_{j_i}}=i\,.
\ee
	The $J_i$-tuple of partitions $\tilde{\nub}^{(i)}$ is chosen as
\be
	\tilde{\nub}^{(i)}=\left(\tilde{\nu}^{(i,1)},\dots,\tilde{\nu}^{(i,J_i)}\right)\,,
\ee
\be
	\tilde{\nu}^{(i,j_i)}=\nu^{(n^{(i)}_{j_i})}\,.
\ee
Repeating this for all values of $i$ between $1$ and $M$ we get a map
\bea
\TT^N_{m, d_-} : \grM^{(N)}_{m,d_-} \times B_{m, M} &\& \ra  \bigsqcup_{J, \  \ell(J)=M , |J|=m}\grM^{(N)}_{J, d_-}\cr
\TT^N_{m, d_-} :	(\nub,{\bf b}) &\& \mapsto (\tilde{\nub}^{(1)},\dots,\tilde{\nub}^{(M)})  
\eea
It is evident that this map is invertible.

	It remains to associate such sets with rationally weighted, doubly labelled constellations. 
	From (\ref{Hexpand}),(\ref{Wexpand}) we obtain
	\bea
	\label{WconstNum}
		&\&\sum_{l=0}^{d_{+}}\sideset{}{'}\sum_{\mub\in\mathfrak{M}_{l,d_{+}}^{(|\mu|)}}W_{G_{{\bf c},{\bf d}}}(\mub;\nub)H(\mub,\nub,\mu,\nu) \cr
	&\&	=
		(-1)^{d-d_{+}-\#\nub}\sum_{\tilde{\mub}\in\mathfrak{M}_{L,d_{+}}^{(|\mu|)}}
		\prod_{i=1}^{L}c_{i}^{\ell^{*}(\mu^{(i)})}\sum_{\mathbf{b}\in B_{\#\nub,d_-}}\prod_{i=1}^{\#\nub}d_{b_i}^{\ell^{*}(\nu^{(i)})}
		H(\tilde{\mub},\nub,\mu,\nu)
	\eea
	and
	\bea
	\label{WconstDen}
		&\&\sum_{m=0}^{d_-}\sideset{}{'}\sum_{\nub\in\mathfrak{M}^{(|\mu|)}_{m,d_-}}(-1)^{d_--m}\sum_{\mathbf{b}\in B_{m,M}}\prod_{i=1}^{m}d_{b_i}^{\ell^{*}(\nu^{(i)})}
		H(\tilde{\mub},\nub,\mu,\nu)\cr
		&\& = \sideset{}{'}\sum_{\substack{
			\tilde{\nub}^{(1)},\dots,\tilde{\nub}^{(M)}\\
			|\nu^{(i,j_i)}|=N\,, \sum_{i=1}^{M}\ell^{*}(\tilde{\nub}^{(i)})=d_-
		}}
		(-1)^{d_--\sum_{i=1}^{M}\#\tilde{\nub}^{(i)}}
		\prod_{i=1}^{M}\prod_{j_i=1}^{\#\tilde{\nub}^{(i)}}d_{i}^{\ell^{*}(\tilde{\nu}^{(i,j_i)})} 
		H(\tilde{\mub},\tilde{\nub}^{(1)},\dots,\tilde{\nub}^{(M)},\mu,\nu). \cr
		&\&
	\eea
	Comparing (\ref{double_label_weight}) with (\ref{WconstNum}), (\ref{WconstDen}), 
	we see that, for a constellation $\GG^J_{{\bf c},{\bf d}}$ with spectrum $J(\GG^J_{{\bf c},{\bf d}})$ from the equivalence class 
	$\left(\nu(\GG^J_{{\bf c},{\bf d}}), \mub(\GG^J_{{\bf c},{\bf d}}),\nub(\GG^J_{{\bf c},{\bf d}}),\mu(\GG^J_{{\bf c},{\bf d}})\right)$,  
	the weight can be expressed as
	\be
		W_{\GG^J_{\mathbf{c},\mathbf{d}}}={1\over N!}\beta^{d(\GG^J_{{\bf c},{\bf d}})}\gamma^{N(\GG^J_{{\bf c},{\bf d}})}
		W_{G_{{\bf c},{\bf d}}}(\mub(\GG^J_{{\bf c},{\bf d}}),\nub(\GG^J_{{\bf c},{\bf d}}))p_{\mu}(\mathbf{s})p_{\nu}(\mathbf{t})\,.
	\ee
	
Since the partitions \mbox{$\left(\nu(\GG^J_{{\bf c},{\bf d}}), \mub(\GG^J_{{\bf c},{\bf d}}),\nub(\GG^J_{{\bf c},{\bf d}}),\mu(\GG^J_{{\bf c},{\bf d}})\right)$} 
are recovered from the equivalence class of the constellation,
 the summation in (\ref{tausum}) can be interpreted as a taken over all equivalence classes of doubly labelled constellations with fixed weight generating function.
	Removing the factor $N!H(\mub,\nub^{(1)},\dots,\nub^{(M)},\mu,\nu)$, which
is the number of constellations in the equivalence class,  we replace the sum over equivalence classes 
with a sum over constellations, which concludes the proof.
\end{proof}

\section{Examples of rational weight generating functions and constellations}
\label{examples_rat_constel}

\subsection{$2$D Toda $\tau$-functions for two simple cases; matrix integrals}
\label{MM_rational_gener_tau}

The two simplest examples of nonpolynomial rational weight generating functions, are:
\be
G_{(\emptyset), (d_1)}  = {1\over 1 - z d_1}, \quad \text{and } G_{(c_1), (d_1)}  = {1 + c_1 z \over 1 - d_1 z},\quad \text{with } c_1\neq 0, d_1 \neq 0.
\ee
Both of these lead to weighted Hurwitz numbers  having interesting combinatorial interpretations.

For the first, we have
\bea
\tau^{(G_{(\emptyset), (d_1)}, \beta, \gamma)}({\bf t},{\bf s})
&\&:= \sum_\lambda \gamma^{|\lambda|} r^{(G_{(\emptyset), (d_1)},  \beta)}_\lambda s_\lambda({\bf t}) s_\lambda({\bf s})
\label{tau_schur_G_0d1} \\
&\& =\sum_{\substack{\mu, \nu \\ |\mu|=|\nu|}} \gamma^{|\mu|}
\sum_{d=0}^\infty \beta^d  H^d_{G_{(\emptyset), (d_1)}}(\mu, \nu)\,  p_\mu({\bf t}) p_\nu({\bf s}).
\label{tau_G_0d1_H}
\eea
where
\be
 r^{(G_{(\emptyset),(d_1)}, \beta)}_\lambda  =  {1\over (-\beta d_1)^{|\lambda|}  \left(-{1\over \beta d_1}\right)_\lambda},
 \ee
 and
 \be
 (a)_\lambda := \prod_{(i,j)\in \lambda} (a+j-i).
 \ee
 It well-known \cite{GGN1, GGN2} that the weighted Hurwitz numbers $H^d_{G_{(\emptyset), (d_1)}}(\mu, \nu)$ in this case may equivalently
be interpreted as (normalized) enumerations of $d$-step paths in the Cayley graph of $\SS_N$ generated
by transpositions $(ab)$ (with $b> a$) starting from an element in the conjugacy class $\cyc(\mu)$, 
and ending in the class $\cyc(\nu)$, such that in the  consecutive sequence of steps $(a_1,b_1) \cdots (a_d, b_d)$ 
along edges, the second elements $(b_1, \dots, b_d)$ form a weakly increasing sequence
\be
b_i \le b_{i+1}, \quad 1\le i \le d-1.
\ee

It is also well-known \cite{GGN1, GGN2} that, if we restrict the  $2$D Toda flow variables  ${\bf t}$ and ${\bf s}$
 in (\ref{tau_schur_G_0d1}), (\ref{tau_G_0d1_H}) to equal the trace invariants
\be
t_i = {\tfrac1 i} \tr(A^i), \quad s_i = {\tfrac1 i} \tr(B^i)
\label{tr_invar_AB}
\ee
of a pair of diagonal $n\times n$ matrices
\be
A = \diag(a_1, \dots, a_n), \quad B = \diag(b_1, \dots, b_n), 
\ee
and let
\be
\beta := -{1\over n d_1},
\label{beta_n_d1}
\ee
then  $\tau^{(G_{(\emptyset), (d_1)}, \beta, \gamma)}({\bf t},{\bf s})$ may be expressed
as the Harish-Chandra-Itzykson-Zuber integral over the group $U(n)$ of unitary $n \times n$ matrices
\be
\tau^{(G_{(\emptyset) (d_1)},-\frac{1}{nd_1}, \gamma)}([A],[B]) =
{1\over \VV_n}\int_{U\in U(n)} e^{\gamma n \tr(UAU^\dag B)}d\mu_H(U) = 
{( \prod_{k=1}^{N-1}k! )\, \det(e^{\gamma n a_i b_j})_{1\le i,j \le n} \over  (n\gamma)^n\Delta({\bf a}) \Delta ({\bf b})},
\label{IZHC_integral}
\ee
where $d\mu(U)$ is the Haar measure on $U(n)$, whose volume is
\be
\VV_n := \int_{U(n)} d\mu_H(U) ={(2\pi)^{\tfrac12 n (n+1)} \over \prod_{k=1}^{n-1}k!},
\ee
 $[A]$ and $[B]$ denote the infinite
sequences of normalized trace invariants defined in (\ref{tr_invar_AB}) and $\Delta({\bf a})$
$\Delta({\bf b})$ are the Vandermonde determinants in the elements
\be
{\bf a} =(a_1, \dots, a_n), \quad {\bf b} =(b_1, \dots, b_n).
\ee

For the second case, we have
\bea
\tau^{(G_{(c_1),( d_1)}, \beta, \gamma)}({\bf t},{\bf s})
&\&:= \sum_\lambda \gamma^{|\lambda|} r^{(G_{(c_1), (d_1)},  \beta)}_\lambda s_\lambda({\bf t}) s_\lambda({\bf s})
\label{tau_schur_G_c1d1} \\
&\& =\sum_{\substack{\mu, \nu \\ |\mu|=|\nu|}} \gamma^{|\mu|}
\sum_{d=0}^\infty \beta^d  H^d_{G_{(c_1), (d_1)}}(\mu, \nu)\,  p_\mu({\bf t}) p_\nu({\bf s}).
\label{tau_G_c1d1_H}
\eea
where
\be
 r^{(G_{(c_1), (d_1)}, \beta)}_\lambda  = \left( - \frac{c_1}{d_1}\right)^{|\lambda|}
 {\left({1\over \beta c_1}\right)_\lambda \over\left(- {1\over \beta d_1}\right)_\lambda}.
 \ee

An alternative combinatorial interpretation  for $ H^d_{G_{(c_1), (d_1))}}(\mu, \nu)$ also exists
in terms of enumeration of monotonic paths in the Cayley graph  \cite{GH1}. If we express $H^d_{G_{(c_1), (d_1))}}(\mu, \nu)$ 
as a polynomial in the parameters $(c_1, d_1)$
 \be
 H^d_{G_{(c_1), (d_1))}}(\mu, \nu) = \sum_{j=0}^d c_1^j d_1^{d-j} E^d_j(\mu, \nu),
 \ee
 then $E^d_j(\mu, \nu)$ gives a (normalized) enumeration of $d$-step paths in the Cayley graph
 of $\SS_N$ generated by transpositions which consist of a sequence of $j$ strictly monotonically
 increasing steps, followed by a sequence of $d-j$ weakly monotonically increasing ones. 
 
If  we restrict the $2$D Toda flow variables ${\bf t}$ and ${\bf s}$  again to equal the normalized trace invariants 
$[A]$ and $[B]$ of a pair of diagonal $n\times n$ matrices,
 and choose $\beta$ as in eq.~(\ref{beta_n_d1}), then $\tau^{(G_{(c_1), (d_1)}, \beta, \gamma)}({\bf t},{\bf s})$ can 
 again   be expressed as a matrix integral over the group $U(n)$, which can also be evaluated  explicitly as a 
 finite determinant \cite{GH1, HO1}:
 \bea
\tau^{(G_{(c_1), (d_1)},-\frac{1}{n d_1}, \gamma)}([A],[B]) &\&=
{1\over \VV_n}\int_{U\in U(n)} \left(\det(1 -z U A U^\dag B)\right)^{n{d_1\over c_1} } d\mu_H(U) \cr
&\& = 
\left(\prod_{k=1}^{n-1} {k!\over (1 - n(1 + {d_1\over c_1})_k} \right) 
{\left(\det(1-z a_i b_j)_{1\le i,j \le n} \right)^{n(1+{d_1 \over c_1}) -1}
\over z^{{1\over 2}n(n-1)} \Delta({\bf a}) \Delta ({\bf b})}, \cr
&\&
\label{HO_integral}
\eea
where
\be
z := -{ \gamma c_1 \over d_1}.
\ee

All other $\tau$-functions  $\tau^{(G_{{\bf c}, {\bf d}}, \beta, \gamma)}$ generating either polynomially 
 or rationally  weighted Hurwitz numbers may also be given matrix integral representations \cite{AC2, AMMN2, BH}
 when the flow parameters ${\bf t}$ and ${\bf s}$ are restricted to the trace invariants of a pair of $n\times n$
  normal matrices $(A, B)$.

\subsection{Examples of weighted constellations}
\label{constell_rat_simplest}

We now display some further examples of singly and doubly labelled constellations corresponding
to weight generating functions $G_{(c_1, c_2, c_3), (\emptyset)}$,  $G_{(\emptyset), (d_1)}$, $G_{(c_1), (d_1)}$
and $G_{(\emptyset), (d_1,d_2)}$,
both on Riemann surfaces of genus $g=0$, that are five-sheeted branched covers of the Riemann sphere
with five branch points,  and of genus $g=1$, that are three-sheeted branched covers, also with five branch points.

Figure \ref{fig:planar_N5_constell_3} shows another example of a doubly weighted labelled planar constellation, of type $G_{(\emptyset), (d_1)}$  
for $N=5$ with five branch points corresponding to the monodromy factorization
\bea
h_{(0,0))} h_{(1,1)} h_{(1,2)} h_{(1,3)} h_{(1,0)}&\& = \Ib_5, \cr
h_{(0,0)} =  (123),\  h_{(1,1)} = (153) , \  h_{(1,2)}  = (15)(23)   , \ h_{(1,3)} &\&= (14) , \  h_{(1,0)}= (14).
\eea

 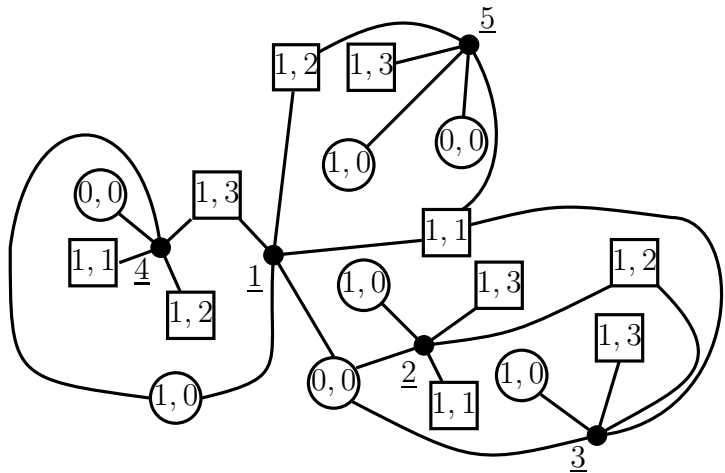
\begin{figure}[H]
 \label{fig:planar_N5_constell_3}
\begin{center}
	\begin{tikzpicture}
		\tikzmath{
			\sts=0.13;
			\x1=0;
			\y1=0;
			\x2=2.0;
			\y2=-1.2;
			\x3=4.3;
			\y3=-2.4;
			\x4=-1.5;
			\y4=0.1;
			\x5=2.6;
			\y5=2.8;
			\qsd=1.2*0.28;
			\qrd=1.1*0.28;
			}
		\draw[fill=black] (\x5,\y5)  circle   (\sts)  node[anchor=south west]{\underline{5}};
		\draw[fill=black] (\x1,\y1)  circle   (\sts)  node[anchor=north east]{\underline{1}};
		\draw[fill=black] (\x2,\y2)  circle   (\sts);
		\draw (\x2-0.2,\y2-0.4) node{\underline{2}};
		\draw[fill=black] (\x3,\y3)  circle   (\sts)  node[anchor=north east]{\underline{3}};
		\draw[fill=black] (\x4,\y4)  circle   (\sts)  node[anchor=north east]{\underline{4}};
		\draw[very thick,fill=white] (\x1,\y5-0.6) rectangle (\x1+0.6,\y5);
		\draw (\x1+0.3,\y5-0.3) node(v215) {$1,2$};
		\draw[very thick,fill=white] (\x5-0.3-\qrd,\y1+0.3-\qrd) rectangle(\x5-0.3+\qrd,\y1+0.3+\qrd);
		\draw (\x5-0.3,\y1+0.3) node(v115) {$1,1$};
		\draw[very thick,fill=white] (\x5-1.3-\qrd,\y5-0.3-\qrd) rectangle (\x5-1.3+\qrd,\y5-0.3+\qrd);
		\draw (\x5-1.3,\y5-0.3) node(v35) {$1,3$};
		\draw[very thick,fill=white] (\x5-1.6,\y5-1.6) circle (\qsd);
		\draw (\x5-1.6,\y5-1.6) node(vinf5) {$1,0$};
		\draw[very thick,fill=white] (\x5-0.1,\y5-1.3) circle (\qsd);
		\draw (\x5-0.1,\y5-1.3) node(v05) {$0,0$};
		\draw[very thick,fill=white] (0.5*\x1+0.5*\x4-\qrd,\y4+0.7-\qrd) rectangle (0.5*\x1+0.5*\x4+\qrd,\y4+0.7+\qrd);
		\draw (0.5*\x1+0.5*\x4,\y4+0.7) node (v314) {$1,3$};
		\draw[very thick,fill=white] (\x4-0.8,\y4+0.7) circle (\qsd);
		\draw (\x4-0.8,\y4+0.7) node (v04) {$0,0$};
		\draw[very thick,fill=white] (\x4-0.9-\qrd,\y4-0.2-\qrd) rectangle(\x4-0.9+\qrd,\y4-0.2+\qrd);
		\draw (\x4-0.9,\y4-0.2) node (v14) {$1,1$};
		\draw[very thick,fill=white] (\x4+0.4-\qrd,\y4-0.9-\qrd) rectangle (\x4+0.4+\qrd,\y4-0.9+\qrd);
		\draw (\x4+0.4,\y4-0.9) node (v24) {$1,2$};
		\draw[very thick,fill=white] (\x4+0.2,\y4-2) circle (\qsd);
		\draw (\x4+0.2,\y4-2) node (vinf4) {$1,0$};
		\draw[very thick,fill=white] (\x1*0.6+\x2*0.4,\y2-0.5) circle (\qsd);
		\draw (\x1*0.6+\x2*0.4,\y2-0.5) node (v012) {$0,0$};
		\draw[very thick,fill=white] (\x2+0.4-\qrd,\y2-0.8-\qrd) rectangle(\x2+0.4+\qrd,\y2-0.8+\qrd);
		\draw (\x2+0.4,\y2-0.8) node (v12) {$1,1$};
		\draw[very thick,fill=white] (\x3+0.5-\qrd,\y2+1.1-\qrd) rectangle (\x3+0.5+\qrd,\y2+1.1+\qrd);
		\draw (\x3+0.5,\y2+1.1) node (v223) {$1,2$};
		\draw[very thick,fill=white] (\x2+1-\qrd,\y2+0.8-\qrd) rectangle (\x2+1.0+\qrd,\y2+0.8+\qrd);
		\draw (\x2+1,\y2+0.8) node (v32) {$1,3$};
		\draw[very thick,fill=white] (\x2-0.8,\y2+0.8) circle (\qsd);
		\draw (\x2-0.8,\y2+0.8) node (vinf2) {$1,0$};
		\draw[very thick,fill=white] (\x3+0.3-\qrd,\y3+1.3-\qrd) rectangle (\x3+0.3+\qrd,\y3+1.3+\qrd);
		\draw (\x3+0.3,\y3+1.3) node (v33) {$1,3$};
		\draw[very thick,fill=white] (\x3-1.0,\y3+0.8) circle (\qsd);
		\draw (\x3-1.0,\y3+0.8) node (vinf3) {$1,0$};
		\draw[very thick] (\x5-1.0,\y5-0.25) -- (\x5,\y5);
		\draw[very thick] (\x5-1.35,\y5-1.35) -- (\x5,\y5);
		\draw[very thick] (v05) -- (\x5,\y5);
		\draw[very thick] (\x4,\y4) -- (v314) -- (\x1,\y1);
		\draw[very thick] (\x4-0.55,\y4+0.45) -- (\x4,\y4);
		\draw[very thick] (\x4-0.55,\y4-0.2) -- (\x4,\y4);
		\draw[very thick] (v24) -- (\x4,\y4);
		\draw[very thick] (v12) -- (\x2,\y2);
		\draw[very thick] (\x2+0.7,\y2+0.5)  -- (\x2,\y2);
		\draw[very thick] (\x2-0.55,\y2+0.55) -- (\x2,\y2);
		\draw[very thick] (\x3+0.3,\y3+1.0) -- (\x3,\y3);
		\draw[very thick] (\x3-0.75,\y3+0.55)  -- (\x3,\y3);
		\draw[very thick] (\x1,\y1) -- (\x1*0.6+\x2*0.4,\y2-0.15);
		\draw[very thick] (\x1*0.6+\x2*0.4+0.3,\y2-0.3) -- (\x2,\y2);
		\draw[very thick] (\x1,\y1) --  (v215);
		\draw[very thick] (\x1+0.6,\y5-0.1) .. controls (\x1+1.5,\y5+0.4) and  (\x5-0.7,\y5+0.4) .. (\x5,\y5);
		\draw[very thick] (\x1,\y1) --  (\x5-0.6,\y1+0.2);
		\draw[very thick] (\x5-0.1,\y1+0.6).. controls (\x5+0.5,\y1+1) and (\x5+0.5,\y5-0.8) .. (\x5,\y5);
		\draw[very thick] (\x4,\y4) .. controls (\x4-0.25,\y4+2) and (\x4-1.75,\y4+2) .. (\x4-2,\y4) .. controls (\x4-2,\y4-1.75) .. (\x4-0.15,\y4-2);
		\draw[very thick] (\x4+0.55,\y4-2) .. controls (\x1+0.3,\y1-1.6) and (\x1-0.1,\y1-1.5) .. (\x1,\y1);
		\draw[very thick](\x1*0.6+\x2*0.4+0.25,\y2-0.75)   .. controls (\x2*0.7+\x3*0.3, \y3-0.4) .. (\x3,\y3) .. controls (\x3+2,\y3+0.2) and (\x3+2,\y1+0.5) .. (\x3+1,\y1+0.5)
		.. controls (\x3-0.5,\y1+0.7)  .. (\x5-0.0,\y1+0.4);
		\draw[very thick] (\x2,\y2) .. controls (\x2*0.5+\x3*0.5+0.05,\y2+0.15).. (\x3+0.2,\y2+0.8);
		\draw[very thick] (\x3+0.8,\y2+0.8).. controls (\x3+2,0.7*\y2+0.3*\y3) and (\x3+1,0.5*\y2+0.5*\y3) .. (\x3,\y3);
	\end{tikzpicture}
\end{center}
\caption{\footnotesize  Example of a doubly labelled planar constellation with $N=5$, $L=0$, $M=1$, $J_1 =3$.}
\end{figure}
Up to the change in labelling, this is identical to the factorizations (\ref{factor_N5k3}) and (\ref{N5_L1_M1_factoriz}),
so the unweighted constellation is  the same as those in Figures \ref{fig:planar_N5_constell_1} and \ref{fig:planar_N5_constell_2}.
However, the double labelling implies a different set of weighting parameters, according to the rules listed in Definition \ref{double_labelled_rational_weighting}.
This case corresponds to $(L,M, J)=(0, 1, (3))$ and therefore,
by eq.~(\ref{double_label_weight})  has total weight
\be
W_{\GG^{(3)}_{(\emptyset), (d_1)}} = {1\over 5!} \gamma^5 \beta^5 d_1^5\   p_{(2,1,1,1)}({\bf t})\  p_{(3,1,1)}({\bf s}).
\ee

Figure \ref{fig:planar_N5_constell_4}  provides a further double labelling,
of type $G_{(\emptyset), (d_1, d_2))}$,  with $(J_1,J_2) =(1,2)$, of the same unweighted constellation that appears in 
Figures \ref{fig:planar_N5_constell_1}, \ref{fig:planar_N5_constell_2} and \ref{fig:planar_N5_constell_3}.
The associated weighting parameters are again determined according to the rules listed in 
Definition \ref{double_labelled_rational_weighting}.
This case corresponds to $(L,M, J)=(0, 2, (1,2))$ and therefore,
by eq.~(\ref{double_label_weight})  has total weight
\be
W_{\GG^{(1,2)}_{(\emptyset) (d_1,d_2)}} = {1\over 5!} \gamma^5  \beta^5 d_1^2 d_2^3\   p_{(2,1,1,1)}({\bf t})\  p_{(3,1,1)}({\bf s}).
\ee

  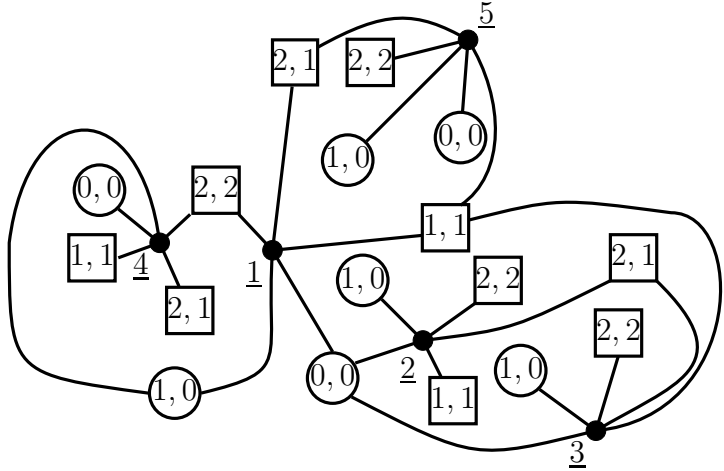
\begin{figure}[H]
 \label{fig:planar_N5_constell_4}
\begin{center}
	\begin{tikzpicture}
		\tikzmath{
			\sts=0.13;
			\x1=0;
			\y1=0;
			\x2=2.0;
			\y2=-1.2;
			\x3=4.3;
			\y3=-2.4;
			\x4=-1.5;
			\y4=0.1;
			\x5=2.6;
			\y5=2.8;
			\qsd=1.2*0.28;
			\qrd=1.1*0.28;
			}
		\draw[fill=black] (\x5,\y5)  circle   (\sts)  node[anchor=south west]{\underline{5}};
		\draw[fill=black] (\x1,\y1)  circle   (\sts)  node[anchor=north east]{\underline{1}};
		\draw[fill=black] (\x2,\y2)  circle   (\sts);
		\draw (\x2-0.2,\y2-0.4) node{\underline{2}};
		\draw[fill=black] (\x3,\y3)  circle   (\sts)  node[anchor=north east]{\underline{3}};
		\draw[fill=black] (\x4,\y4)  circle   (\sts)  node[anchor=north east]{\underline{4}};
		\draw[very thick,fill=white] (\x1,\y5-0.6) rectangle (\x1+0.6,\y5);1,0
		\draw (\x1+0.3,\y5-0.3) node(v215) {$2,1$};
		\draw[very thick,fill=white] (\x5-0.3-\qrd,\y1+0.3-\qrd) rectangle(\x5-0.3+\qrd,\y1+0.3+\qrd);
		\draw (\x5-0.3,\y1+0.3) node(v115) {$1,1$};
		\draw[very thick,fill=white] (\x5-1.3-\qrd,\y5-0.3-\qrd) rectangle (\x5-1.3+\qrd,\y5-0.3+\qrd);
		\draw (\x5-1.3,\y5-0.3) node(v35) {$2,2$};
		\draw[very thick,fill=white] (\x5-1.6,\y5-1.6) circle (\qsd);
		\draw (\x5-1.6,\y5-1.6) node(vinf5) {$1,0$};
		\draw[very thick,fill=white] (\x5-0.1,\y5-1.3) circle (\qsd);
		\draw (\x5-0.1,\y5-1.3) node(v05) {$0,0$};
		\draw[very thick,fill=white] (0.5*\x1+0.5*\x4-\qrd,\y4+0.7-\qrd) rectangle (0.5*\x1+0.5*\x4+\qrd,\y4+0.7+\qrd);
		\draw (0.5*\x1+0.5*\x4,\y4+0.7) node (v314) {$2,2$};
		\draw[very thick,fill=white] (\x4-0.8,\y4+0.7) circle (\qsd);
		\draw (\x4-0.8,\y4+0.7) node (v04) {$0,0$};
		\draw[very thick,fill=white] (\x4-0.9-\qrd,\y4-0.2-\qrd) rectangle(\x4-0.9+\qrd,\y4-0.2+\qrd);
		\draw (\x4-0.9,\y4-0.2) node (v14) {$1,1$};
		\draw[very thick,fill=white] (\x4+0.4-\qrd,\y4-0.9-\qrd) rectangle (\x4+0.4+\qrd,\y4-0.9+\qrd);
		\draw (\x4+0.4,\y4-0.9) node (v24) {$2,1$};
		\draw[very thick,fill=white] (\x4+0.2,\y4-2) circle (\qsd);
		\draw (\x4+0.2,\y4-2) node (vinf4) {$1,0$};
		\draw[very thick,fill=white] (\x1*0.6+\x2*0.4,\y2-0.5) circle (\qsd);
		\draw (\x1*0.6+\x2*0.4,\y2-0.5) node (v012) {$0,0$};
		\draw[very thick,fill=white] (\x2+0.4-\qrd,\y2-0.8-\qrd) rectangle(\x2+0.4+\qrd,\y2-0.8+\qrd);
		\draw (\x2+0.4,\y2-0.8) node (v12) {$1,1$};
		\draw[very thick,fill=white] (\x3+0.5-\qrd,\y2+1.1-\qrd) rectangle (\x3+0.5+\qrd,\y2+1.1+\qrd);
		\draw (\x3+0.5,\y2+1.1) node (v223) {$2,1$};
		\draw[very thick,fill=white] (\x2+1-\qrd,\y2+0.8-\qrd) rectangle (\x2+1.0+\qrd,\y2+0.8+\qrd);
		\draw (\x2+1,\y2+0.8) node (v32) {$2,2$};
		\draw[very thick,fill=white] (\x2-0.8,\y2+0.8) circle (\qsd);
		\draw (\x2-0.8,\y2+0.8) node (vinf2) {$1,0$};
		\draw[very thick,fill=white] (\x3+0.3-\qrd,\y3+1.3-\qrd) rectangle (\x3+0.3+\qrd,\y3+1.3+\qrd);
		\draw (\x3+0.3,\y3+1.3) node (v33) {$2,2$};
		\draw[very thick,fill=white] (\x3-1.0,\y3+0.8) circle (\qsd);
		\draw (\x3-1.0,\y3+0.8) node (vinf3) {$1,0$};
		\draw[very thick] (\x5-1.0,\y5-0.25) -- (\x5,\y5);
		\draw[very thick] (\x5-1.35,\y5-1.35) -- (\x5,\y5);
		\draw[very thick] (v05) -- (\x5,\y5);
		\draw[very thick] (\x4,\y4) -- (v314) -- (\x1,\y1);
		\draw[very thick] (\x4-0.55,\y4+0.45) -- (\x4,\y4);
		\draw[very thick] (\x4-0.55,\y4-0.2) -- (\x4,\y4);
		\draw[very thick] (v24) -- (\x4,\y4);
		\draw[very thick] (v12) -- (\x2,\y2);
		\draw[very thick] (\x2+0.7,\y2+0.5)  -- (\x2,\y2);
		\draw[very thick] (\x2-0.55,\y2+0.55) -- (\x2,\y2);
		\draw[very thick] (\x3+0.3,\y3+1.0) -- (\x3,\y3);
		\draw[very thick] (\x3-0.75,\y3+0.55)  -- (\x3,\y3);
		\draw[very thick] (\x1,\y1) -- (\x1*0.6+\x2*0.4,\y2-0.15);
		\draw[very thick] (\x1*0.6+\x2*0.4+0.3,\y2-0.3) -- (\x2,\y2);
		\draw[very thick] (\x1,\y1) --  (v215);
		\draw[very thick] (\x1+0.6,\y5-0.1) .. controls (\x1+1.5,\y5+0.4) and  (\x5-0.7,\y5+0.4) .. (\x5,\y5);
		\draw[very thick] (\x1,\y1) --  (\x5-0.6,\y1+0.2);
		\draw[very thick] (\x5-0.1,\y1+0.6).. controls (\x5+0.5,\y1+1) and (\x5+0.5,\y5-0.8) .. (\x5,\y5);
		\draw[very thick] (\x4,\y4) .. controls (\x4-0.25,\y4+2) and (\x4-1.75,\y4+2) .. (\x4-2,\y4) .. controls (\x4-2,\y4-1.75) .. (\x4-0.15,\y4-2);
		\draw[very thick] (\x4+0.55,\y4-2) .. controls (\x1+0.3,\y1-1.6) and (\x1-0.1,\y1-1.5) .. (\x1,\y1);
		\draw[very thick](\x1*0.6+\x2*0.4+0.25,\y2-0.75)   .. controls (\x2*0.7+\x3*0.3, \y3-0.4) .. (\x3,\y3) .. controls (\x3+2,\y3+0.2) and (\x3+2,\y1+0.5) .. (\x3+1,\y1+0.5)
		.. controls (\x3-0.5,\y1+0.7)  .. (\x5-0.0,\y1+0.4);
		\draw[very thick] (\x2,\y2) .. controls (\x2*0.5+\x3*0.5+0.05,\y2+0.15).. (\x3+0.2,\y2+0.8);
		\draw[very thick] (\x3+0.8,\y2+0.8).. controls (\x3+2,0.7*\y2+0.3*\y3) and (\x3+1,0.5*\y2+0.5*\y3) .. (\x3,\y3);
	\end{tikzpicture}
\end{center}
\caption{\footnotesize  Example of a doubly labelled planar constellation with $N=5$, $L=0$, $M=2$, $J_1 =1$, $J_2 =2$.}
\end{figure}

Figure \ref{fig:torus_constell_N3_1} shows a doubly labelled constellation with $N=3$,
 and $(L,M, J)=(3, 0, (0))$, so the weight generating function is the cubic polynomial $G_{(c_1, c_2, c_3), (\emptyset)}$
and this is therefore  equivalent to a singly labelled constellation.  It has five branch points with monodromy factorization
\bea
h_{(0,0))} h_{(1,0)} h_{(2,0)} h_{(3,0)} h_{(4,0)}&\& = \Ib_5, \cr
h_{(0,0)} =  (123),\  h_{(1,0)} = (12) , \  h_{(2,0)}  = (23)   , \ h_{(3,0)} &\&= (12) , \  h_{(4,0)}= (12).
\label{N3_torus_factorization}
\eea
The genus is $g=1$, so it is mapped on a torus.
The weighting parameters are again determined by the rules of Definition \ref{double_labelled_rational_weighting}.
so, by eq.~(\ref{double_label_weight})  the total weight is
\be
W_{\GG^0_{(c_1, c_2, c_3), (\emptyset)}} = {1\over 3!} \gamma^3 \beta^3 c_1 c_2 c_3\   p_{(2,1)}({\bf t})\  p_{(3)}({\bf s}).
\ee

\begin{figure}[H]
 \label{fig:torus_constell_N3_1}
	\begin{center}
		\tikzmath{
			\sts=0.13;
			\qsd=0.28;
			}
		\begin{tikzpicture}
			\draw[very thick,blue] (0,0) ellipse[y radius=2cm,x radius=5cm];
			\draw[very thick,blue] (-2.7,0.7) arc[start angle=210,end angle=330, y radius=1.5cm, x radius=3cm];
			\draw[very thick,blue] (2.0,0.35) arc[start angle=20, end angle=160, x radius=2.2cm, y radius=1 cm];
			\draw[very thick,fill=black] (1.8,-1.0) circle[radius=\sts];
			\draw[very thick,fill=black] (-2.5,-0.2) circle[radius=\sts];
			\draw[very thick,fill=black] (4,0.3) circle[radius=\sts];
			\draw (4,0.3) node[anchor=north west] {$\underline{3}$};
			\draw (-2.5,-0.2) node[below=2pt] {$\underline{2}$};
			\draw (1.8,-1.0) node [anchor=south west] {$\underline{1}$};
			\draw[very thick] (0,-0.5) circle[radius=1.2*\qsd];
			\draw (0,-0.52) node {$0,0$};
			\draw[very thick] (-2.4,-0.1) arc[start angle=180, end angle=265, y radius=0.5cm, x radius=2.2cm];
			\draw[very thick] (0.3,-0.6) arc[start angle=275, end angle=325, y radius=2cm, x radius=5.0cm];
			\draw[very thick] (-1.3,-1.03) circle[radius=1.2*\qsd];
			\draw (-1.3,-1.03) node {$1,0$};
			\draw[very thick] (-2.4,-0.1) arc[start angle=180, end angle=265, y radius=0.5cm, x radius=2.2cm];
			\draw[very thick,rotate=70,dash pattern= on 1pt off 1pt ] (-2.5,2.0) arc [start angle=180, end angle=360,x radius=1.05cm, y radius =0.5cm];
			\draw[very thick,rotate=70] (-0.45,1.9) arc [start angle=0,end angle=60, x radius=1.1cm,y radius=0.5cm];
			\draw[very thick] (-2.7,-1.6) .. controls (-2.8,-1.2) and (-2.6,-1.4) .. (-1.6,-1.03);
			\draw[very thick,rotate=70,dash pattern= on 1pt off 1pt ] (0.8,-1.6) arc [start angle=0, end angle=180,x radius=1.15cm, y radius =0.5cm];
			\draw[very thick,rotate=-15] (-0.7,-1.2) arc[start angle=270, end angle=320, y radius=2cm, x radius=3.5cm];
			\draw[very thick] (0.3,-0.62)--(1.8,-1.0);
			\draw[very thick] (-0.7,-1.55) circle[radius=1.2*\qsd];
			\draw (-0.7,-1.55) node {$2,0$};
			\draw[very thick] (-0.6+\qsd,-1.55) arc [start angle=270, end angle=330, x radius=2.5cm, y radius=1cm];
			\draw[very thick,rotate=70] (-1.5,-1.6) arc [start angle=180, end angle=270,x radius=1.15cm, y radius =0.5cm];
			\draw[very thick] (3.2,0.4) circle[radius=1.2*\qsd];
			\draw (3.2,0.4) node {$4,0$};
			\draw[very thick] (2.1,1.0) circle[radius=1.2*\qsd];
			\draw (2.1,1.0) node {$3,0$};
			\draw[very thick] (2.4,0.9) .. controls (2.8,0.75) and (2.8,0.45) .. (1.8,0.2);
			\draw[very thick] (2.45,0.95) arc[start angle=90,end angle=0,x radius=1.5cm, y radius=0.5cm];
			\draw[very thick] (3.2,-0.6) circle[radius=1.2*\qsd];
			\draw (3.2,-0.6) node {$1,0$};
			\draw[very thick] (-4.4,0.3) circle[radius=1.2*\qsd];
			\draw (-4.4,0.3) node {$4,0$};
			\draw[very thick] (-3.4,0.3) circle[radius=1.2*\qsd];
			\draw (-3.4,0.3) node {$3,0$};
			\draw[very thick] (-2,1.0) circle[radius=1.2*\qsd];
			\draw (-2,1.0) node {$2,0$};
			\draw[very thick] (-4.4,-0.05) arc [start angle=220, end angle=300, x radius=1.5cm, y radius=1cm];
			\draw[very thick] (-2.5,-0.2) .. controls (-3.2,0.6) and (-3.2,0.8) .. (-2.3,1.0);
			\draw[very thick] (-1.7,1.0) .. controls (2.2,2.2) and (4,1) .. (4,0.3);
			\draw[very thick] (-4.4,0.6) arc [start angle=170,end angle=0, x radius=4.5cm, y radius=1.5cm];
			\draw[very thick] (1.8,-1.0) .. controls (2.8,-1.5) and (4.7,-1)..(4.53,0.35);
			\draw[very thick] (-3.15,0.05) arc[start angle=225, end angle=270, x radius=1.2cm, y radius=1cm];
			\draw[very thick] (4,0.3) arc[start angle=45,end angle=70,x radius=1.2cm, y radius=0.75cm];
			\draw[very thick] (4.0,0.3) arc[start angle=360, end angle=310, x radius=1.2cm, y radius=1cm];
		\end{tikzpicture}
	\end{center}
	\caption{A constellation  with $N=3$, $L=3$, $M=0$, and $g=1$.}
\end{figure}
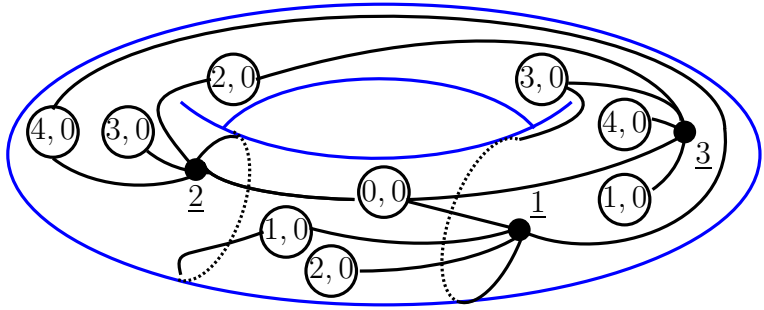

Figure \ref{fig:torus_constell_N3_1_period_parallel} shows the same constellation as Figure \ref{fig:torus_constell_N3_1},  
displayed on the lattice of period parallelograms.

\begin{figure}[H]
 \label{fig:torus_constell_N3_1_period_parallel}
	\begin{center}
		\tikzmath{
			\sts=0.13;
			\qsd=0.28;
			}
		\begin{tikzpicture}
			\foreach \xoff in {0,7.3}
			\foreach \yoff in {0,6}
			{
			\draw[very thick,fill=black] (1.4+\xoff,-3.0+\yoff) circle[radius=\sts];
			\draw[very thick,fill=black] (-2+\xoff,0+\yoff) circle[radius=\sts];
			\draw[very thick,fill=black] (3+\xoff,1+\yoff) circle[radius=\sts];
			\draw (3+\xoff,1+\yoff) node[anchor=north west] {$\underline{3}$};
			\draw (-2+\xoff,0+\yoff) node[below=2pt] {$\underline{2}$};
			\draw (1.4+\xoff,-3.0+\yoff) node[below=2pt] {$\underline{1}$};
			\draw[very thick] (1.4+\xoff,-1+\yoff) circle[radius=1.2*\qsd];
			\draw (1.4+\xoff,-1+\yoff) node {$0,0$};
			\draw[very thick] (0+\xoff,-2+\yoff) circle[radius=1.2*\qsd];
			\draw (0+\xoff,-2+\yoff) node {$1,0$};
			\draw[very thick] (3.3+\xoff,-2+\yoff) circle[radius=1.2*\qsd];
			\draw (3.3+\xoff,-2+\yoff) node {$4,0$};
			\draw[very thick] (0.4+\xoff,-4+\yoff) circle[radius=1.2*\qsd];
			\draw (0.4+\xoff,-4+\yoff) node {$2,0$};
			\draw[very thick] (-2+\xoff,1+\yoff) circle[radius=1.2*\qsd];
			\draw (-2+\xoff,1+\yoff) node {$2,0$};
			\draw[very thick] (3+\xoff,-3.5+\yoff) circle[radius=1.2*\qsd];
			\draw (3+\xoff,-3.5+\yoff) node {$3,0$};
			\draw[very thick] (-2.8+\xoff,0+\yoff) circle[radius=1.2*\qsd];
			\draw (-2.8+\xoff,0+\yoff) node {$3,0$};
			\draw[very thick] (2+\xoff,1+\yoff) circle[radius=1.2*\qsd];
			\draw (2+\xoff,1+\yoff) node {$4,0$};
			\draw[very thick] (3+\xoff,0+\yoff) circle[radius=1.2*\qsd];
			\draw (3+\xoff,0+\yoff) node {$1,0$};
			\draw[very thick] (-2+\xoff,0+\yoff) -- (-2+\xoff,0.7+\yoff);
			\draw[very thick] (-2+\xoff,0+\yoff) -- (-2.5+\xoff,0+\yoff);
			\draw[very thick] (3+\xoff,1+\yoff) -- (3+\xoff,0.3+\yoff);
			\draw[very thick] (3+\xoff,1+\yoff) -- (2.3+\xoff,1+\yoff);
			\draw[very thick] (1.4+\xoff,-3.0+\yoff) -- (1.4+\xoff,-1.3+\yoff);
			\draw[very thick] (1.4+\xoff,-3.0+\yoff) -- (0.2+\xoff,-2.3+\yoff);
			\draw[very thick] (1.4+\xoff,-3.0+\yoff) -- (3.1+\xoff,-2.3+\yoff);
			\draw[very thick] (1.4+\xoff,-3.0+\yoff) -- (2.75+\xoff,-3.3+\yoff);
			\draw[very thick] (1.4+\xoff,-3.0+\yoff) -- (0.65+\xoff,-3.75+\yoff);
			\draw[very thick] (-2+\xoff,0+\yoff) -- (1.05+\xoff,-0.95+\yoff);
			\draw[very thick] (3+\xoff,1+\yoff) -- (1.75+\xoff,-0.95+\yoff);
			\draw[very thick, dash pattern = on 1pt  off 1pt] (\xoff-0.5*7.3+0.3,\yoff-4.5) rectangle (\xoff+0.5*7.3+0.3,\yoff+1.5);
			\draw[very thick](-2.35+\xoff,1+\yoff)--(-0.5*7.3+0.3+\xoff,1+\yoff);
			\draw[very thick](3+\xoff,1+\yoff)--(0.5*7.3+0.3+\xoff,1+\yoff);
			\draw[very thick](3+\xoff,-3.8+\yoff)--(3+\xoff,-4.5+\yoff);
			\draw[very thick](3+\xoff,1+\yoff)--(3+\xoff,1.5+\yoff);
			\draw[very thick](-2+\xoff,0+\yoff)--(-1.25+\xoff,1.5+\yoff);
			\draw[very thick](-0.2+\xoff,-2.3+\yoff)--(-1.25+\xoff,-4.5+\yoff);
			\draw[very thick](3.5+\xoff,-1.8+\yoff)--(0.5*7.3+0.3+\xoff,-1.35+\yoff);
			\draw[very thick](-2+\xoff,0+\yoff)--(-0.5*7.3+0.3+\xoff,-1.35+\yoff);
			}
		\end{tikzpicture}
	\end{center}
	\caption{The same constellation as Figure \ref{fig:torus_constell_N3_1},  displayed on the lattice of period parallelograms.}
\end{figure}
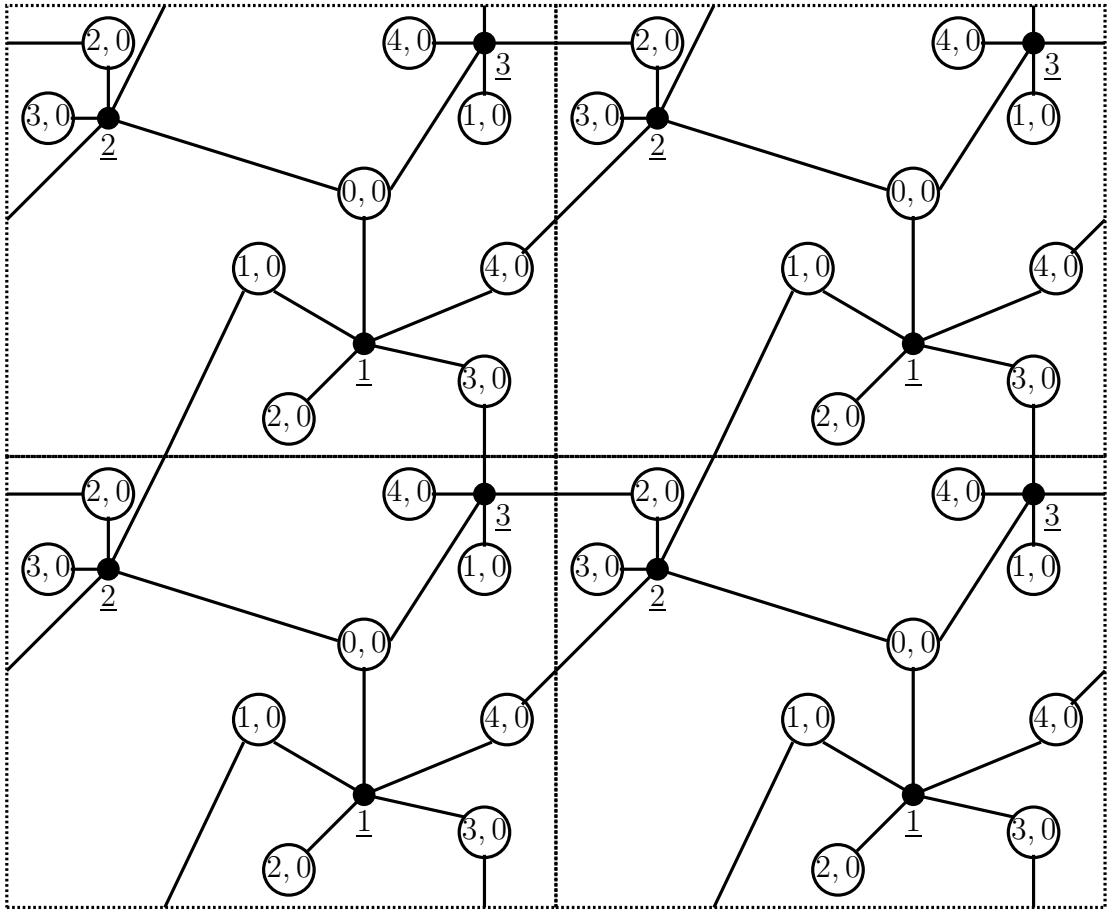

Figures \ref{fig:torus_constell_N3_2},  \ref{fig:torus_constell_N3_3} and \ref{fig:torus_constell_N3_4} show
three further types of double labelling of the same  constellation as Figure \ref{fig:torus_constell_N3_1},
corresponding to $(L, M, J)$ equal to $(0,1, (3))$, $(1,1,(2))$ and $(0,2,(1,2))$, respectively, 
and weight generating functions  $G_{(0),(d_1)}$, $G_{(c_1),(d_1)}$ , $G_{(c_1),(d_1, d_2)}$.
Up to changes in labelling, the group elements comprising the corresponding factorization of $\Ib_3 \in \SS_3$ 
are the same as in (\ref{N3_torus_factorization}). 
However, the weighting parameters,  determined according to the rules in Definition 
\ref{double_labelled_rational_weighting}, are different in each case. The corresponding weighted constellations
 therefore contribute different terms to the sums in (\ref{tau_G_cd_H}), defining different 
 $\tau$-functions  $\tau^{(G_{{\bf s}, {\bf d}}, \beta, \gamma)}({\bf t}, {\bf s})$ for each case.
 
Figure \ref{fig:torus_constell_N3_2} corresponds to $(L,M, J)=(0, 1, (3))$, with weight generating
function $G_{(\emptyset),(d_1)}$ and therefore, by eq.~(\ref{double_label_weight})  has total weight
\be
W_{\GG^{(3)}_{(\emptyset), (d_1)}} = {1\over 3!} \gamma^3  \beta^3d_1^3\   p_{(2,1)}({\bf t})\  p_{(3)}({\bf s}).
\ee


\begin{figure}[H]
 \label{fig:torus_constell_N3_2}
	\begin{center}
		\tikzmath{
			\sts=0.13;
			\qsd=1.2*0.28;
			}
		\begin{tikzpicture}
			\draw[very thick,blue] (0,0) ellipse[y radius=2cm,x radius=5cm];
			\draw[very thick,blue] (-2.7,0.7) arc[start angle=210,end angle=330, y radius=1.5cm, x radius=3cm];
			\draw[very thick,blue] (2.0,0.35) arc[start angle=20, end angle=160, x radius=2.2cm, y radius=1 cm];
			\draw[very thick,fill=black] (1.8,-1.0) circle[radius=\sts];
			\draw[very thick,fill=black] (-2.5,-0.2) circle[radius=\sts];
			\draw[very thick,fill=black] (4,0.3) circle[radius=\sts];
			\draw (4,0.3) node[anchor=north west] {$\underline{3}$};
			\draw (-2.5,-0.2) node[below=2pt] {$\underline{2}$};
			\draw (1.8,-1.0) node [anchor=south west] {$\underline{1}$};
			\draw[very thick] (0,-0.5) circle[radius=\qsd];
			\draw (0,-0.52) node {$0,0$};
			\draw[very thick] (-2.4,-0.1) arc[start angle=180, end angle=265, y radius=0.5cm, x radius=2.2cm];
			\draw[very thick] (0.35,-0.6) arc[start angle=275, end angle=325, y radius=2cm, x radius=5.0cm];
			\draw[very thick] (-1.4-\qsd,-1.03-\qsd) rectangle (-1.4+\qsd,-1.03+\qsd);
			\draw (-1.4,-1.03) node {$1,1$};
			\draw[very thick] (-2.4,-0.1) arc[start angle=180, end angle=265, y radius=0.5cm, x radius=2.2cm];
			\draw[very thick,rotate=70,dash pattern= on 1pt off 1pt ] (-2.5,2.0) arc [start angle=180, end angle=360,x radius=1.05cm, y radius =0.5cm];
			\draw[very thick,rotate=70] (-0.45,1.9) arc [start angle=0,end angle=60, x radius=1.1cm,y radius=0.5cm];
			\draw[very thick] (-2.7,-1.6) .. controls (-2.8,-1.2) and (-2.6,-1.4) .. (-1.75,-1.03);
			\draw[very thick,rotate=70,dash pattern= on 1pt off 1pt ] (0.8,-1.6) arc [start angle=0, end angle=180,x radius=1.15cm, y radius =0.5cm];
			\draw[very thick,rotate=-15] (-0.75,-1.2) arc[start angle=270, end angle=320, y radius=2cm, x radius=3.55cm];
			\draw[very thick] (0.35,-0.62)--(1.8,-1.0);
			\draw[very thick] (-0.6-\qsd,-1.55-\qsd) rectangle (-0.6+\qsd,-1.55+\qsd);
			\draw (-0.6,-1.55) node {$1,2$};
			\draw[very thick] (-0.6+\qsd,-1.55) arc [start angle=270, end angle=330, x radius=2.5cm, y radius=1cm];
			\draw[very thick,rotate=70] (-1.5,-1.6) arc [start angle=180, end angle=270,x radius=1.15cm, y radius =0.5cm];
			\draw[very thick] (3.2,0.4) circle[radius=\qsd];
			\draw (3.2,0.4) node {$1,0$};
			\draw[very thick] (2.1-\qsd,1.0-\qsd) rectangle (2.1+\qsd,1.0+\qsd);
			\draw (2.1,1.0) node {$1,3$};
			\draw[very thick] (2.45,0.9) .. controls (2.8,0.75) and (2.8,0.45) .. (1.8,0.2);
			\draw[very thick] (2.45,0.95) arc[start angle=90,end angle=0,x radius=1.5cm, y radius=0.5cm];
			\draw[very thick] (3.2-\qsd,-0.6-\qsd) rectangle (3.2+\qsd,-0.6+\qsd);
			\draw (3.2,-0.6) node {$1,1$};
			\draw[very thick] (-4.4,0.3) circle[radius=\qsd];
			\draw (-4.4,0.3) node {$1,0$};
			\draw[very thick] (-3.4-\qsd,0.3-\qsd) rectangle (-3.4+\qsd,0.3+\qsd);
			\draw (-3.4,0.3) node {$1,3$};
			\draw[very thick] (-2-\qsd,1.1-\qsd) rectangle (-2+\qsd,1.1+\qsd);
			\draw (-2,1.1) node {$1,2$};
			\draw[very thick] (-4.4,-0.05) arc [start angle=220, end angle=300, x radius=1.5cm, y radius=1cm];
			\draw[very thick] (-2.5,-0.2) .. controls (-3.2,0.6) and (-3.2,0.8) .. (-2.35,1.0);
			\draw[very thick] (-1.65,1.0) .. controls (2.2,2.2) and (4,1) .. (4,0.3);
			\draw[very thick] (-4.4,0.65) arc [start angle=170,end angle=-2, x radius=4.49cm, y radius=1.5cm];
			\draw[very thick] (1.8,-1.0) .. controls (2.8,-1.5) and (4.7,-1)..(4.53,0.35);
			\draw[very thick] (-3.1,0.00) arc[start angle=225, end angle=260, x radius=1.2cm, y radius=1cm];
			\draw[very thick] (4,0.3) arc[start angle=45,end angle=70,x radius=1.2cm, y radius=0.75cm];
			\draw[very thick] (4.0,0.3) arc[start angle=360, end angle=310, x radius=1.2cm, y radius=1cm];
		\end{tikzpicture}
	\end{center}
	\caption{The same constellation as Figure \ref{fig:torus_constell_N3_1}, with double labelling parameters $L=0$, $M=1$, $J_1=3$.}
\end{figure}
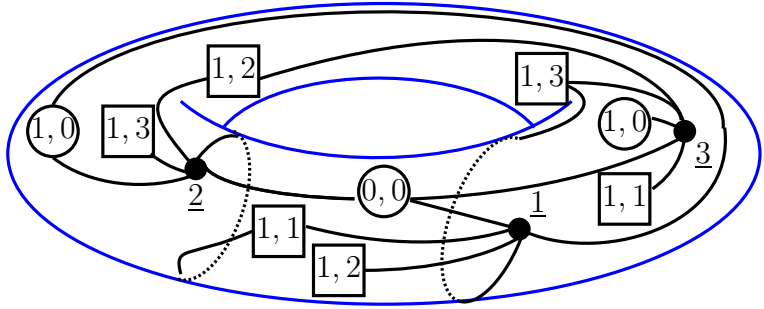

Figure \ref{fig:torus_constell_N3_3} corresponds to $(L,M, J)=$(1,1,(2))$ $, with weight generating
function $G_{(c_1),(d_1)}$ and therefore has total weight
\be
W_{\GG^{(2)}_{(c_1), (d_1)}} = {1\over 3!} \gamma^3   \beta^3 c_1 d_1^2\   p_{(2,1)}({\bf t})\  p_{(3)}({\bf s}).
\ee

\begin{figure}[H]
 \label{fig:torus_constell_N3_3}
	\begin{center}
		\tikzmath{
			\sts=0.13;
			\qsd=1.2*0.28;
			}
		\begin{tikzpicture}
			\draw[very thick,blue] (0,0) ellipse[y radius=2cm,x radius=5cm];
			\draw[very thick,blue] (-2.7,0.7) arc[start angle=210,end angle=330, y radius=1.5cm, x radius=3cm];
			\draw[very thick,blue] (2.0,0.35) arc[start angle=20, end angle=160, x radius=2.2cm, y radius=1 cm];
			\draw[very thick,fill=black] (1.8,-1.0) circle[radius=\sts];
			\draw[very thick,fill=black] (-2.5,-0.2) circle[radius=\sts];
			\draw[very thick,fill=black] (4,0.3) circle[radius=\sts];
			\draw (4,0.3) node[anchor=north west] {$\underline{3}$};
			\draw (-2.5,-0.2) node[below=2pt] {$\underline{2}$};
			\draw (1.8,-1.0) node [anchor=south west] {$\underline{1}$};
			\draw[very thick] (0,-0.5) circle[radius=\qsd];
			\draw (0,-0.52) node {$0,0$};
			\draw[very thick] (-2.4,-0.1) arc[start angle=180, end angle=265, y radius=0.5cm, x radius=2.2cm];
			\draw[very thick] (0.35,-0.6) arc[start angle=275, end angle=325, y radius=2cm, x radius=5.0cm];
			\draw[very thick] (-1.4,-1.03) circle[radius=\qsd];
			\draw (-1.4,-1.03) node {$1,0$};
			\draw[very thick] (-2.4,-0.1) arc[start angle=180, end angle=265, y radius=0.5cm, x radius=2.2cm];
			\draw[very thick,rotate=70,dash pattern= on 1pt off 1pt ] (-2.5,2.0) arc [start angle=180, end angle=360,x radius=1.05cm, y radius =0.5cm];
			\draw[very thick,rotate=70] (-0.45,1.9) arc [start angle=0,end angle=60, x radius=1.1cm,y radius=0.5cm];
			\draw[very thick] (-2.7,-1.6) .. controls (-2.8,-1.2) and (-2.6,-1.4) .. (-1.75,-1.03);
			\draw[very thick,rotate=70,dash pattern= on 1pt off 1pt ] (0.8,-1.6) arc [start angle=0, end angle=180,x radius=1.15cm, y radius =0.5cm];
			\draw[very thick,rotate=-15] (-0.75,-1.2) arc[start angle=270, end angle=320, y radius=2cm, x radius=3.55cm];
			\draw[very thick] (0.35,-0.62)--(1.8,-1.0);
			\draw[very thick] (-0.6-\qsd,-1.55-\qsd) rectangle (-0.6+\qsd,-1.55+\qsd);
			\draw (-0.6,-1.55) node {$1,1$};
			\draw[very thick] (-0.6+\qsd,-1.55) arc [start angle=270, end angle=330, x radius=2.5cm, y radius=1cm];
			\draw[very thick,rotate=70] (-1.5,-1.6) arc [start angle=180, end angle=270,x radius=1.15cm, y radius =0.5cm];
			\draw[very thick] (3.2,0.4) circle[radius=\qsd];
			\draw (3.2,0.4) node {$2,0$};
			\draw[very thick] (2.1-\qsd,1.0-\qsd) rectangle (2.1+\qsd,1.0+\qsd);
			\draw (2.1,1.0) node {$1,2$};
			\draw[very thick] (2.45,0.9) .. controls (2.8,0.75) and (2.8,0.45) .. (1.8,0.2);
			\draw[very thick] (2.45,0.95) arc[start angle=90,end angle=0,x radius=1.5cm, y radius=0.5cm];
			\draw[very thick] (3.2,-0.6) circle[radius=\qsd];
			\draw (3.2,-0.6) node {$1,0$};
			\draw[very thick] (-4.4,0.3) circle[radius=\qsd];
			\draw (-4.4,0.3) node {$2,0$};
			\draw[very thick] (-3.4-\qsd,0.3-\qsd) rectangle (-3.4+\qsd,0.3+\qsd);
			\draw (-3.4,0.3) node {$1,2$};
			\draw[very thick] (-2-\qsd,1.1-\qsd) rectangle (-2+\qsd,1.1+\qsd);
			\draw (-2,1.1) node {$1,1$};
			\draw[very thick] (-4.4,-0.05) arc [start angle=220, end angle=300, x radius=1.5cm, y radius=1cm];
			\draw[very thick] (-2.5,-0.2) .. controls (-3.2,0.6) and (-3.2,0.8) .. (-2.35,1.0);
			\draw[very thick] (-1.65,1.0) .. controls (2.2,2.2) and (4,1) .. (4,0.3);
			\draw[very thick] (-4.4,0.65) arc [start angle=170,end angle=-2, x radius=4.49cm, y radius=1.5cm];
			\draw[very thick] (1.8,-1.0) .. controls (2.8,-1.5) and (4.7,-1)..(4.53,0.35);
			\draw[very thick] (-3.1,0.00) arc[start angle=225, end angle=260, x radius=1.2cm, y radius=1cm];
			\draw[very thick] (4,0.3) arc[start angle=45,end angle=70,x radius=1.2cm, y radius=0.75cm];
			\draw[very thick] (4.0,0.3) arc[start angle=360, end angle=310, x radius=1.2cm, y radius=1cm];
		\end{tikzpicture}
	\end{center}
	\caption{The same constellation as Figure \ref{fig:torus_constell_N3_1}, with double labelling parameters $L=1$, $M=1$ $J_1=2$.}
\end{figure}
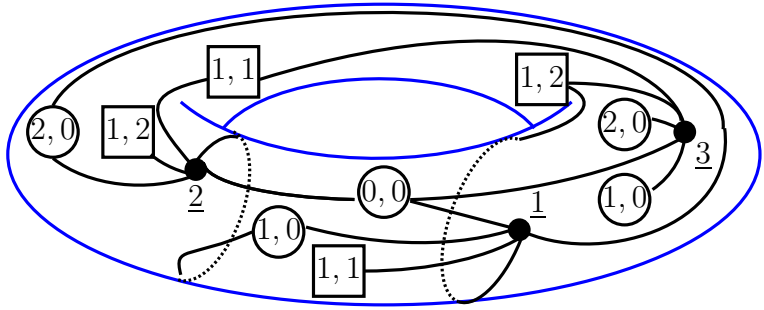

Figure \ref{fig:torus_constell_N3_4} corresponds to $(L,M, J) =(0,2,(1,2)) $, with weight generating
function $G_{(0),(d_1, d_2)}$ and therefore  has total weight
\be
W_{\GG^{(1,2)}_{(\emptyset), (d_1,d_2)}} = {1\over 3!} \gamma^3   \beta^3 d_1 d_2^2\   p_{(2,1)}({\bf t})\  p_{(3)}({\bf s}).
\ee

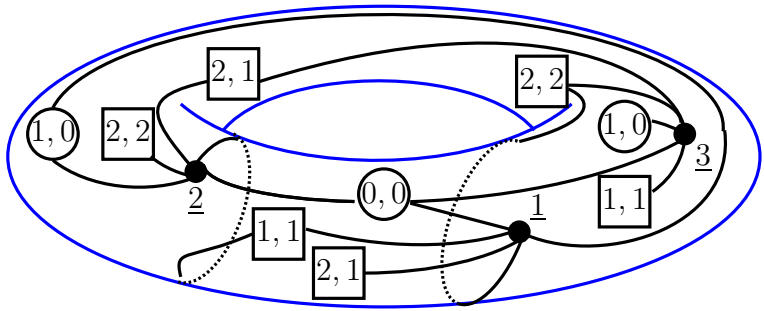
\begin{figure}[H]
 \label{fig:torus_constell_N3_4}
	\begin{center}
		\tikzmath{
			\sts=0.13;
			\qsd=1.2*0.28;
			}
		\begin{tikzpicture}
			\draw[very thick,blue] (0,0) ellipse[y radius=2cm,x radius=5cm];
			\draw[very thick,blue] (-2.7,0.7) arc[start angle=210,end angle=330, y radius=1.5cm, x radius=3cm];
			\draw[very thick,blue] (2.0,0.35) arc[start angle=20, end angle=160, x radius=2.2cm, y radius=1 cm];
			\draw[very thick,fill=black] (1.8,-1.0) circle[radius=\sts];
			\draw[very thick,fill=black] (-2.5,-0.2) circle[radius=\sts];
			\draw[very thick,fill=black] (4,0.3) circle[radius=\sts];
			\draw (4,0.3) node[anchor=north west] {$\underline{3}$};
			\draw (-2.5,-0.2) node[below=2pt] {$\underline{2}$};
			\draw (1.8,-1.0) node [anchor=south west] {$\underline{1}$};
			\draw[very thick] (0,-0.5) circle[radius=\qsd];
			\draw (0,-0.52) node {$0,0$};
			\draw[very thick] (-2.4,-0.1) arc[start angle=180, end angle=265, y radius=0.5cm, x radius=2.2cm];
			\draw[very thick] (0.35,-0.6) arc[start angle=275, end angle=325, y radius=2cm, x radius=5.0cm];
			\draw[very thick] (-1.4-\qsd,-1.03-\qsd) rectangle (-1.4+\qsd,-1.03+\qsd);
			\draw (-1.4,-1.03) node {$1,1$};
			\draw[very thick] (-2.4,-0.1) arc[start angle=180, end angle=265, y radius=0.5cm, x radius=2.2cm];
			\draw[very thick,rotate=70,dash pattern= on 1pt off 1pt ] (-2.5,2.0) arc [start angle=180, end angle=360,x radius=1.05cm, y radius =0.5cm];
			\draw[very thick,rotate=70] (-0.45,1.9) arc [start angle=0,end angle=60, x radius=1.1cm,y radius=0.5cm];
			\draw[very thick] (-2.7,-1.6) .. controls (-2.8,-1.2) and (-2.6,-1.4) .. (-1.75,-1.03);
			\draw[very thick,rotate=70,dash pattern= on 1pt off 1pt ] (0.8,-1.6) arc [start angle=0, end angle=180,x radius=1.15cm, y radius =0.5cm];
			\draw[very thick,rotate=-15] (-0.75,-1.2) arc[start angle=270, end angle=320, y radius=2cm, x radius=3.55cm];
			\draw[very thick] (0.35,-0.62)--(1.8,-1.0);
			\draw[very thick] (-0.6-\qsd,-1.55-\qsd) rectangle (-0.6+\qsd,-1.55+\qsd);
			\draw (-0.6,-1.55) node {$2,1$};
			\draw[very thick] (-0.6+\qsd,-1.55) arc [start angle=270, end angle=330, x radius=2.5cm, y radius=1cm];
			\draw[very thick,rotate=70] (-1.5,-1.6) arc [start angle=180, end angle=270,x radius=1.15cm, y radius =0.5cm];
			\draw[very thick] (3.2,0.4) circle[radius=\qsd];
			\draw (3.2,0.4) node {$1,0$};
			\draw[very thick] (2.1-\qsd,1.0-\qsd) rectangle (2.1+\qsd,1.0+\qsd);
			\draw (2.1,1.0) node {$2,2$};
			\draw[very thick] (2.45,0.9) .. controls (2.8,0.75) and (2.8,0.45) .. (1.8,0.2);
			\draw[very thick] (2.45,0.95) arc[start angle=90,end angle=0,x radius=1.5cm, y radius=0.5cm];
			\draw[very thick] (3.2-\qsd,-0.6-\qsd) rectangle (3.2+\qsd,-0.6+\qsd);
			\draw (3.2,-0.6) node {$1,1$};
			\draw[very thick] (-4.4,0.3) circle[radius=\qsd];
			\draw (-4.4,0.3) node {$1,0$};
			\draw[very thick] (-3.4-\qsd,0.3-\qsd) rectangle (-3.4+\qsd,0.3+\qsd);
			\draw (-3.4,0.3) node {$2,2$};
			\draw[very thick] (-2-\qsd,1.1-\qsd) rectangle (-2+\qsd,1.1+\qsd);
			\draw (-2,1.1) node {$2,1$};
			\draw[very thick] (-4.4,-0.05) arc [start angle=220, end angle=300, x radius=1.5cm, y radius=1cm];
			\draw[very thick] (-2.5,-0.2) .. controls (-3.2,0.6) and (-3.2,0.8) .. (-2.35,1.0);
			\draw[very thick] (-1.65,1.0) .. controls (2.2,2.2) and (4,1) .. (4,0.3);
			\draw[very thick] (-4.4,0.65) arc [start angle=170,end angle=-2, x radius=4.49cm, y radius=1.5cm];
			\draw[very thick] (1.8,-1.0) .. controls (2.8,-1.5) and (4.7,-1)..(4.53,0.35);
			\draw[very thick] (-3.1,0.00) arc[start angle=225, end angle=260, x radius=1.2cm, y radius=1cm];
			\draw[very thick] (4,0.3) arc[start angle=45,end angle=70,x radius=1.2cm, y radius=0.75cm];
			\draw[very thick] (4.0,0.3) arc[start angle=360, end angle=310, x radius=1.2cm, y radius=1cm];
		\end{tikzpicture}
	\end{center}
\caption{The same constellation as Figure \ref{fig:torus_constell_N3_1}, with double labelling parameters $L=0$, $M=2$, $J_1=1$,  $J_2=2$.}
\end{figure}

 \bigskip \bigskip
\noindent 
\small{ {\it Acknowledgements.} 
This work was partially supported by the Natural Sciences and Engineering Research Council of Canada (NSERC) 
and the Fonds de Recherche du Qu\'ebec, Nature et Technologies (FRQNT).  
\bigskip
\bigskip 


 \newcommand{\arxiv}[1]{\href{http://arxiv.org/abs/#1}{arXiv:{#1}}}

\bigskip
\noindent

\end{document}